\documentclass[a4paper,10pt]{amsart}

\usepackage[T1]{fontenc}
\usepackage{lmodern}
\usepackage{dsfont}
\usepackage{amsmath}
\usepackage{amsfonts}
\usepackage{amssymb}
\usepackage{amsthm}
\usepackage{mathtools}
\usepackage{thmtools}
\usepackage{xcolor}
\usepackage{tikz}
\usetikzlibrary{cd,arrows,decorations.pathmorphing,fit}
\usepackage[inline]{enumitem}
\usepackage{xfrac}
\usepackage{xr}
\usepackage[pdftex,
            pdfauthor={},
            pdftitle={Equivalence Relations on Vertex Operator Algebras, II: Witt Equivalence and Orbifolds},
            pdfsubject={},
            pdfkeywords={},
            pdfproducer={},
            pdfcreator={}]{hyperref}

\theoremstyle{plain}
\newtheorem{thm}{Theorem}[section]

\newtheorem{thmph}[thm]{The\-o\-rem\textsuperscript{ph}}
\newtheorem{prop}[thm]{Proposition}
\newtheorem{propph}[thm]{Propo\-si\-tion\textsuperscript{ph}}
\newtheorem{cor}[thm]{Corollary}
\newtheorem{lem}[thm]{Lemma}

\newtheorem{conj}[thm]{Conjecture}
\newtheorem{conjph}[thm]{Con\-jec\-ture\textsuperscript{ph}}

\newtheorem*{thm*}{Theorem}
\newtheorem*{conj*}{Conjecture}
\newtheorem*{claim*}{Claim}
\newtheorem*{prop*}{Proposition}
\theoremstyle{definition}
\newtheorem{defi}[thm]{Definition}
\newtheorem{defiph}[thm]{Def\-i\-ni\-tion\textsuperscript{ph}}

\newtheorem*{nota*}{Notation}
\newtheorem{rem}[thm]{Remark}
\newtheorem{remph}[thm]{Re\-mark\textsuperscript{ph}}
\newtheorem{ex}[thm]{Example}
\newtheorem{exph}[thm]{Ex\-am\-ple\textsuperscript{ph}}

\newtheorem{prob}[thm]{Problem}

\newcommand{\Q}{\mathbb{Q}}
\newcommand{\Z}{\mathbb{Z}}
\newcommand{\Ns}{\mathbb{Z}_{>0}}
\newcommand{\N}{\mathbb{Z}_{\geq0}}
\newcommand{\C}{\mathbb{C}}
\newcommand{\R}{\mathbb{R}}

\newcommand{\tr}{\operatorname{tr}}

\newcommand{\Aut}{\operatorname{Aut}}
\newcommand{\Hom}{\operatorname{Hom}}

\newcommand{\rk}{\operatorname{rk}}

\newcommand{\sign}{\operatorname{sign}}
\newcommand{\voa}{vertex operator algebra}
\newcommand{\Voa}{Vertex operator algebra}
\newcommand{\VOA}{Vertex Operator Algebra}
\newcommand{\vosa}{vertex operator subalgebra}

\newcommand{\fpvosa}{fixed-point vertex operator subalgebra}

\newcommand{\vac}{\textbf{1}}
\newcommand{\ch}{\operatorname{ch}}
\newcommand{\Ch}{\operatorname{Ch}}
\newcommand{\id}{\operatorname{id}}

\newcommand{\ee}{\mathfrak{e}}
\newcommand{\g}{\mathfrak{g}}
\newcommand{\hh}{\mathfrak{h}}

\newcommand{\Inn}{\operatorname{Inn}}

\renewcommand{\sl}{\mathfrak{sl}}

\newcommand{\Irr}{\operatorname{Irr}}

\newcommand{\orb}{\operatorname{orb}}

\newcommand{\strat}{strong\-ly rational}

\newcommand{\II}{I\!I}

\newcommand{\Com}{\operatorname{Com}}
\newcommand{\no}{\,{\raise0.25em\hbox{$\mathop{\hphantom{\cdot}}\limits^{_{\circ}}_{^{\circ}}$}}\,}

\newcommand{\Rep}{\operatorname{Rep}}

\newcommand{\Vect}{\operatorname{Vect}}
\newcommand{\Cat}{\mathcal{C}}

\newcommand{\loc}{\mathrm{loc}}

\newcommand{\gDD}{\II_{12,0}(2_{\II}^{-10}4_{\II}^{-2})}

\newcommand{\gJJ}{\II_{6,0}(2_{\II}^{+4}4_{\II}^{-2}3^{+5})}

\makeatletter
\newcommand*{\addFileDependency}[1]{
\typeout{(#1)}
\@addtofilelist{#1}
\IfFileExists{#1}{}{\typeout{No file #1.}}
}\makeatother

\begin{document}

\title[Equivalence Relations on Vertex Operator Algebras, II]{Equivalence Relations on Vertex Operator Algebras, II: Witt Equivalence and Orbifolds}
\author[Sven Möller and Brandon~C. Rayhaun]{Sven Möller\textsuperscript{\lowercase{a}} and Brandon~C. Rayhaun\textsuperscript{\lowercase{b},\lowercase{c}}}

\begin{abstract}
When can two \strat{} \voa{}s or 1+1d rational conformal field theories (RCFTs) be related by topological manipulations? For \voa{}s, the term ``topological manipulations'' refers to operations like passing to a conformal extension or restricting to a conformal subalgebra; for RCFTs, topological manipulations include operations like gauging (or orbifolding) a finite part of a generalized global symmetry or interpolating to a new theory via a topological line interface of finite quantum dimension.

Inspired by results in the theory of even lattices and tensor categories, we say that two \strat{} \voa{}s are \emph{Witt equivalent} if their central charges agree and if their modular tensor categories are Witt equivalent. Two RCFTs are said to be Witt equivalent if their central charges agree and if their associated 2+1d topological field theories can be separated by a topological surface. We argue that Witt equivalence is necessary for two theories to be related by topological manipulations. We conjecture that it is also sufficient, and give proofs in various special cases.

We relate this circle of ideas to the problem of classifying RCFTs, and to lore concerning deformation classes of quantum field theories. We use the notion of Witt equivalence to argue, conditional on a mild technical hypothesis, that all of the finite symmetries of the $\mathrm{SU}(2)_1$ Wess--Zumino--Witten model are invertible. We also sketch a ``quantum Galois theory'' for chiral CFTs, which generalizes prior mathematical literature by incorporating non-invertible symmetries; we illustrate this non-invertible Galois theory in the context of the monster CFT $V^\natural$, for which we produce two Fibonacci lines. Finally, we discuss $p$-neighborhood of \voa{}s, which is a special topological manipulation related to $\mathbb{Z}_p$-orbifolding.
\end{abstract}

\thanks{\textsuperscript{a}{Universität Hamburg, Hamburg, Germany}}
\thanks{\textsuperscript{b}{Yang Institute for Theoretical Physics, State University of New York, Stony Brook, NY, United States of America}}
\thanks{\textsuperscript{c}{Institute for Advanced Study, Princeton, NJ, United States of America}}

\thanks{Email: \href{mailto:math@moeller-sven.de}{\nolinkurl{math@moeller-sven.de}}, \href{mailto:brandonrayhaun@gmail.com}{\nolinkurl{brandonrayhaun@gmail.com}}}

\maketitle

\setcounter{tocdepth}{1}
\tableofcontents
\setcounter{tocdepth}{2}

\pagebreak


\section{Introduction}


\subsection{For Mathematicians}

(An introduction for physicists follows in \autoref{sec:physintro}.) \Voa{}s and their categories of representations provide an axiomatization of chiral algebras and fusion rings of $2$-dimensional conformal field theories in physics. While originally introduced to solve the monstrous moonshine conjecture \cite{Bor86,FLM88,Bor92}, by now, \voa{}s are ubiquitous in mathematics and appear, e.g., in the context of algebraic geometry, group theory, Lie theory and the theory of automorphic forms. Some applications to mathematical physics will be addressed in \autoref{sec:physintro}.

The structure theory of (suitably regular) \voa{}s bears many similarities to the theory of even lattices over the integers. This is especially manifest when considering classification problems in not too large central charges, as studied in, e.g., \cite{MS23,MR23,Ray23,HM23}. For larger central charges, a complete classification is not realistic. However, one can ask how \voa{}s can be organized into equivalence classes under equivalence relations that are coarser than being isomorphic. In a series of two papers, the first part being \cite{MR24a}, we aim to address this question.

\medskip

While in \cite{MR24a} we studied \voa{} analogs of the lattice genus \cite{Hoe03,Mor21}, in this work we generalize the coarser notions (see \autoref{fig:lat}) of Witt and rational equivalence to (\strat{}) \voa{}s (see also \cite{Hoe03,SY24}).

Recall (see \autoref{defi:lattice_witt}) that two even lattices $L$ and $M$ are \emph{Witt equivalent} if their discriminant forms $L'/L$ and $M'/M$ are Witt equivalent (\cite{Wit37}, see \autoref{defi:metric_witt}) as metric groups, i.e.\ as finite abelian groups equipped with a quadratic form, and if they have the same signature over $\R$, $\sign(L)=\sign(M)$. This can be reformulated as saying that the even lattices $L$ and $M$ admit same-rank (or finite-index) extensions to even lattices in the same genus (see \autoref{fig:wittlat}).

There is also the notion of rational equivalence (see \autoref{defi:lattice_rat}). Two even lattices $L$ and $M$ are said to be \emph{rationally equivalent} if $L\otimes_\Z\Q\cong M\otimes_\Z\Q$ as $\Q$-vector spaces equipped with quadratic forms. Equivalently, there is an even lattice~$K$ that is isomorphic to full-rank sublattices of $L$ and $M$ (see \autoref{fig:ratlat}).

In fact, it is not difficult to see that the notions of Witt and rational equivalence of even lattices are actually the same (see \autoref{prop:lat_ratwit} and \autoref{fig:lat}); however, in view of the generalization to \voa{}s, which we shall discuss momentarily, it is useful to keep them separate.

\medskip

Both Witt and rational equivalence can be generalized to \strat{} \voa{}s. In parallel to the discussion in \cite{MR24a} regarding the lattice genus, the pattern that emerges is that each equivalence relation on even lattices (genus, Witt and rational equivalence) splits up into two equivalence relations for \voa{}s, one that is a more ``classical'' analog and one that is a more honest ``quantum'' analog. The former will always be a refinement of the latter. We depict this in the main diagram, \autoref{fig:summary}.

\begin{figure}[ht]
\begin{tikzpicture}[decoration={snake,post length=3pt,amplitude=1pt,segment length=4pt}]
\def \tt {-2}
\tikzstyle{equiv}=[draw,rectangle, rounded corners, minimum width=3cm, minimum height=1cm]
\tikzstyle{a}=[thick,->,>=stealth,shorten >=2pt,shorten <=2pt]
\tikzstyle{t}=[thick,<->,>=stealth,shorten >=2pt,shorten <=2pt]
\tikzstyle{b}=[thick,<->,>=stealth,shorten >=2pt,shorten <=2pt]
\tikzstyle{d}=[thick,dashed,->,>=stealth,shorten >=2pt,shorten <=2pt]
\tikzstyle{dd}=[thick,double equal sign distance, -Implies]
\tikzstyle{s}=[draw,decorate,->]
\node[equiv,align=center] (witt) at (0,0) {Witt Equiv.\ \\ (Def.\ \ref{defi:weakwittequivalence})};
\node[equiv,align=center] (orbifold) at (8.8,0) {Orbifold Equiv.\ \\ (Def.\ \ref{defi:orbequiv})};
\node[equiv,align=center] (rational) at (0,-5) {Strong Witt Equiv.\ \\ (Def.\ \ref{defi:strongwittequivalence})};
\node[equiv,align=center] (inner) at (8.8,-5) {Inner Orb.\ Equiv.\ \\ (Def.\ \ref{defi:innerorbequ})};
\node[equiv,align=center] (bulk) at (4.4,-2.5) {Bulk Genus \\
(Def.\ \ref{defi:bulkgenus})};
\node[equiv,align=center] (hyperbolic) at (4.4,-7.5) {Hyperbolic Genus \\ (Def.\ \ref{defi:hypgen})};
\node[] at (5.95,-.8) {Prop.\ \ref{prop:wittbulkorb}};
\node[] at (0,-6.5) {\emph{VOA Concepts}};
\draw[b] (5.7,-1.73) to [bend left=20] (5,-.15);
\draw[d] ([shift={(0,-.1)}]witt.east) to node[anchor=north,shift={(-1,0)}] {Conj.\ \ref{conj:wittorb}} ([shift={(0,-.1)}]orbifold.west);
\draw[a] ([shift={(0,.1)}]orbifold.west) to node[anchor=south] {Thm.\ \ref{prop:orbimplieswitt}} ([shift={(0,.1)}]witt.east);
\draw[d] (bulk) -- node[anchor=west,shift={(-.5,-.4)}] {Conj.\ \ref{conj:anisobulkorb}} (orbifold);
\draw[a] (hyperbolic) -- node[anchor=west,shift={(-.5,-.4)}] {Cor.\ \ref{cor:orbifoldrelation}} (inner);
\draw[a] (rational) -- node[anchor=west] {Prop.\ \ref{prop:weakstrongwitt}} (witt);
\draw[a] (inner) -- node[anchor=east] {Rem.\ \ref{rem:innerorbtoorb}} (orbifold);
\draw[a] (hyperbolic) -- node[anchor=west,shift={(-.3,.35)}] {Prop.\ \ref{prop:hypstrongwitt}} (rational);
\draw[a] (bulk) -- node[anchor=east,shift={(.32,-.35)}] {Rem.\ \ref{rem:bulkimplieswitt}} (witt);
\draw[t] (rational) -- node[anchor=south,shift={(1.4,0)}] {Thm.\ \ref{thm:rationalinnerorb}} (inner);
\draw[draw=white,fill=white] (4.25,-5.15) rectangle ++(0.3,0.3);
\draw[a] (hyperbolic) -- node[anchor=east,shift={(0,.85)}] {\cite[Thm.\ \ref*{MR1cor:hypbulk}]{MR24a}} (bulk);
\path[s] (4.4,-10.5) to node[anchor=west,shift={(0.1,0)}] {\emph{Generalize}} (4.4,-8.5);
\node[] at (0,-10.75+\tt) {\emph{Lattice Concepts}};
\node[equiv,align=center] (lattice_witt) at (0,-9+\tt) {Witt Equiv.\ \\ (Def.\ \ref{defi:lattice_witt})};
\node[equiv,align=center] (lattice_rat) at (8.8,-9+\tt) {Rational\ Equiv.\ \\ (Def.\ \ref{defi:lattice_rat})};
\node[equiv,align=center] (lattice_gen) at (4.4,-11.5+\tt) {Genus \\ (Def.\ \ref{defi:latticegenus})};
\draw[a] (lattice_gen) -- (lattice_witt);
\draw[a] (lattice_gen) -- (lattice_rat);
\draw[t] (lattice_witt) -- (lattice_rat);
\end{tikzpicture}
\caption{Interrelations between the various equivalence relations on \strat{} \voa{}s studied in our work. The solid arrows are rigorously established. The dashed arrows are conjectural.
}
\label{fig:summary}
\end{figure}
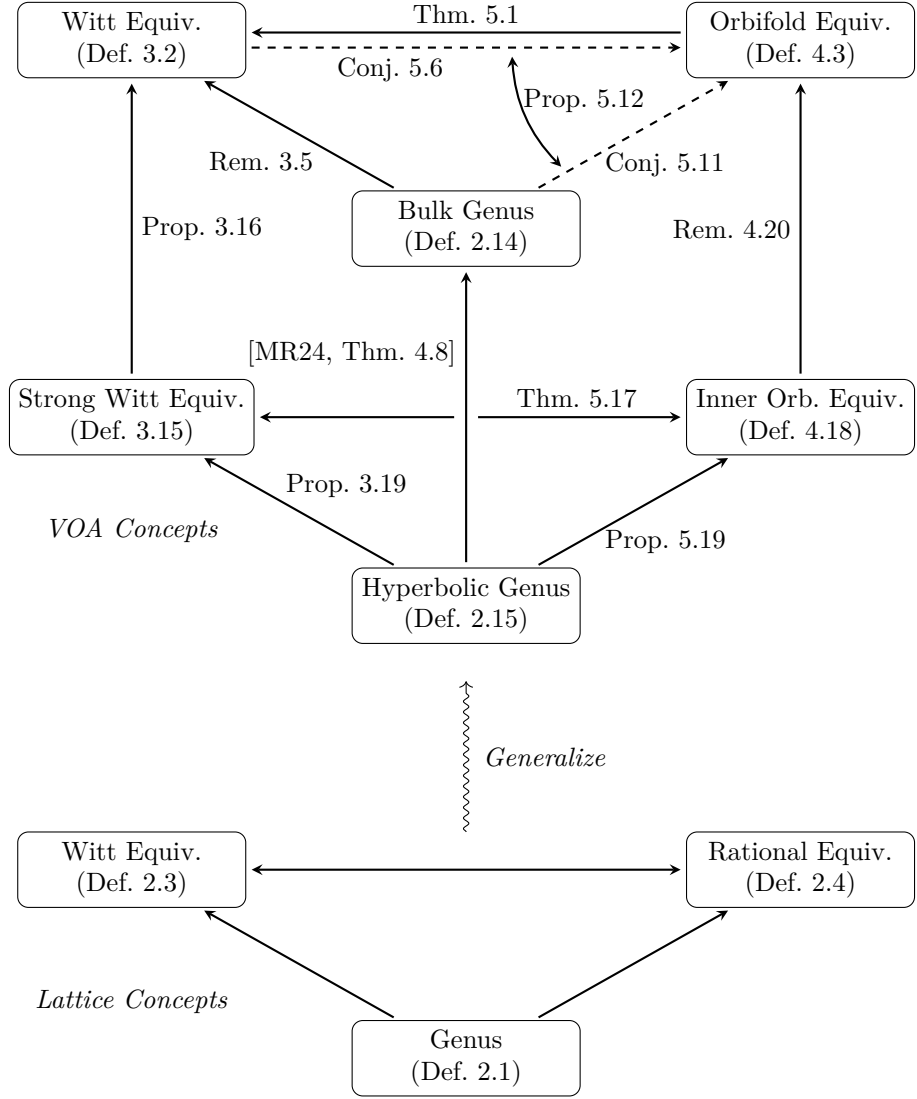

To understand Witt equivalence of \voa{}s, we first need to recall the corresponding notion for modular tensor categories (here, in the usual semisimple sense) \cite{DMNO13,FSV13,Kon14}. Two modular tensor categories $\Cat$ and $\mathcal{D}$ are \emph{Witt equivalent} if there is an equivalence of ribbon fusion categories of the form
\begin{equation*}
\Cat\boxtimes\overline{\mathcal{D}}\cong Z(\mathcal{F})
\end{equation*}
for some spherical fusion category $\mathcal{F}$ (see \autoref{defn:Wittequivalencecategories}). We denote the Witt class of a modular tensor category $\Cat$ by $[\Cat]$. Without motivating this definition here much further, we note that it reduces exactly to Witt equivalence of metric groups, \autoref{defi:metric_witt}, if one considers the (pseudo-unitary) pointed modular tensor categories $\Cat(D)$, which are in bijection with metric groups $D$ \cite{JS93,EGNO15}.

Witt equivalence generalizes to \voa{}s as follows. (Take note of the formal analogy to Witt equivalence of even lattices in \autoref{defi:lattice_witt}.) Recall that the representation category $\Rep(V)$ of a \strat{} \voa{} $V$ is a modular tensor category \cite{Hua08b}. Two \strat{} \voa{}s $V$ and $V'$ are \emph{(weakly) Witt equivalent} (see \autoref{defi:weakwittequivalence}) if they have the same central charge, $c(V)=c(V')$, and if $\Rep(V)$ and $\Rep(V')$ are Witt equivalent as modular tensor categories. We denote the Witt class of a \strat{} \voa{} by the pair $([\Rep(V)],c(V))$.

Based on the correspondence between conformal extensions of \strat{} \voa{}s and condensable algebras in modular tensor categories \cite{KO02,HKL15,CMSY25} (see \autoref{prop:algext}), we can reformulate Witt equivalence as follows, also assuming for simplicity that the \voa{}s are \emph{positive}, i.e.\ that all irreducible modules except for the adjoint module are positively graded. This situation is depicted in \autoref{fig:witt}.

\smallskip

\noindent \textbf{\autoref{prop:alternativecharacterizationWittgenus}.} \emph{Two positive, \strat{} \voa{}s $V$ and~$V'$ are Witt equivalent if and only if there exist \strat{} \voa{}s $W$ and $W'$ in the same bulk genus that contain $V$ and $V'$, respectively, as conformal subalgebras (implying that $W$ and $W'$ are positive as well).}

\smallskip

While the ``quantum'' (or weak) version of Witt equivalence is based on the corresponding notion for modular tensor categories, the ``classical'' (or strong) version, by contrast, applies Witt equivalence only to the associated lattice of the \voa{}. (This dichotomy is analogous to the difference between bulk and hyperbolic genus, as given in \cite[\autoref*{MR1defi:bulkgenus}]{MR24a} and \cite[\autoref*{MR1defi:commgenus}]{MR24a}, which becomes even more apparent in the unpointed case \cite[\autoref*{MR1cor:unpointed}]{MR24a}.) Recall from \cite{Mas14,HM23} that any \strat{} \voa{} $V$ can be written as a simple-current extension of a dual pair $C\otimes V_L\subset V$, where $L$ is the associated lattice of $V$, $V_L$ the corresponding lattice \voa{} and $C$ the Heisenberg commutant (with vanishing weight-$1$ space, $C_1=\{0\}$). Now, two \strat{} \voa{}s are \emph{strongly Witt equivalent} if their Heisenberg commutants are isomorphic and their associated lattices are Witt equivalent as lattices (see \autoref{defi:strongwittequivalence}).

\medskip

Rational equivalence of lattices generalizes to what we call orbifold equivalence of \voa{}s (see \autoref{defi:orbequiv}). Two \strat{} \voa{}s $V$ and $V'$ are \emph{orbifold equivalent} if there is a sequence of \strat{} \voa{}s $V=V_1,V_2,\dots,V_n=V'$ such that $V_i$ and $V_{i+1}$ share a common (up to isomorphism) \strat{} conformal subalgebra $W_i$, as depicted in \autoref{fig:orbequiv}.

This definition is supposed to capture what it means for two \strat{} \voa{}s to be related under (repeated) ``generalized orbifolds'', such as gauging a fusion category of topological line operators. However, as these notions are not yet (fully) mathematically rigorous, we shall content ourselves with the above definition (see \cite{Rie25,DNRX25} for some work in this direction).

We point out (see \autoref{rem:nontrans}) that being related by a common \strat{} conformal subalgebra is likely not an equivalence relation, forcing us to explicitly build in the transitivity into our definition. The potential counterexample follows from the conjectural classification of the positive, \strat{} \voa{}s of central charge $c=1$ (see \autoref{ex:cc1}) and is ultimately a consequence of the famous ADE-classification of the finite subgroups of $\mathrm{PSL}(2,\C)$, which correspond to Du Val singularities.

There is also the ``classical'' (or strong) version of orbifold equivalence (see \autoref{defi:innerorbequ} and \autoref{rem:innerorbtrans}). Two \strat{} \voa{}s $V$ and $V'$ are \emph{inner orbifold equivalent} if they share a \strat{} conformal subalgebra $W$ such that the concrete realization of $W$ inside $V$ (and similarly for $V'$) can be chosen such that $W$ and $V$ have the same Heisenberg commutant. Here, it actually suffices to demand that $V$ and $V'$ be connected ``in one step'' in order to obtain an equivalence relation. We have chosen the terminology ``inner'' because, in some special cases, inner orbifold equivalence reduces to a (group-like) orbifold construction under inner automorphisms.

\medskip

We then study how Witt and rational equivalence of \voa{}s are related. The following is the main conjecture of this text, where we again for simplicity assume that the \voa{}s are positive.

\smallskip

\noindent\textbf{\autoref{conj:wittorb}.} \emph{Two positive, \strat{} \voa{}s belong to the same Witt class if and only if they are orbifold equivalent.}

\smallskip

A variant of this in the special case of holomorphic \voa{}s was conjectured in \cite{Joh21}.

With the machinery developed in this text and in \cite{MR24a} it is relatively straightforward to prove the reverse implication in the conjecture.

\smallskip

\noindent\textbf{\autoref{prop:orbimplieswitt}.} \emph{If two \strat{} \voa{}s $V$ and $V'$ are orbifold equivalent, then they are Witt equivalent.}

\smallskip

We expect the forward direction in \autoref{conj:wittorb} to be much harder to establish. Indeed, we argue in \autoref{rem:fakemoonshine} that, as a very narrow special case, this would likely imply the uniqueness of the moonshine module \cite{FLM88}, \autoref{conj:moonshineuniqueness}.

Instead, we prove the following partial result, which shows that strongly Witt equivalent \voa{}s are necessarily orbifold equivalent, and in fact inner orbifold equivalent.

\smallskip

\noindent\textbf{\autoref{thm:rationalinnerorb}.} \emph{Two \strat{} \voa{}s $V$ and $V'$ are inner orbifold equivalent if and only if they are strongly Witt equivalent. }

\smallskip

We also argue that, in order to prove \autoref{conj:wittorb}, it suffices to show that any two positive, \strat{} \voa{}s in the same bulk genus are orbifold equivalent (see \autoref{conj:anisobulkorb} and \autoref{prop:wittbulkorb}).

The advantage of the formulation in \autoref{conj:anisobulkorb} is that it is possible to collect evidence in its favor for small central charges. Indeed, it follows from \cite{DM04,ELMS21,MS23} that all \strat{}, \emph{holomorphic} \voa{}s~$V$ (i.e.\ with $\Rep(V)\cong\Vect$) of central charge $c\leq24$, except for possible ``fake'' moonshine modules, are orbifold equivalent. Hence:

\smallskip

\noindent\textbf{\autoref{prop:witttrivialorbifoldingVOAs}.} \emph{Assuming \autoref{conj:moonshineuniqueness} on the uniqueness of the moonshine module, \autoref{conj:wittorb} is true for the infinitely many positive, \strat{} \voa{}s belonging to the Witt classes $([\Vect],c)$ with $c\leq 24$.}

\smallskip

Evidently, bulk (genus) equivalence for \voa{}s is a refinement of Witt equivalence (see \autoref{rem:bulkimplieswitt}), and conjecturally, as we just wrote, one of orbifold equivalence as well (see \autoref{conj:anisobulkorb}). Moreover, we show in \autoref{prop:hypstrongwitt} and \autoref{thm:rationalinnerorb} that hyperbolic equivalence implies strong Witt and inner orbifold equivalence, thus completing the main diagram in \autoref{fig:summary}.

\medskip

We study the concepts introduced in this work in a number of instructive examples. We consider the Witt classes $([\Vect],c)$ with Witt trivial representation categories in \autoref{ex:trivialwitt}. Then, in \autoref{ex:discretewitt}, we consider the positive, \strat{} \voa{}s of central charge $c<1$, i.e.\ the discrete series Virasoro minimal models. In \autoref{ex:cc1}, we study the positive, \strat{} \voa{}s of central charge $c=1$, conjecturally given by rank-$1$ lattice \voa{}s and their conformal subalgebras.

The (conjectural) relation between Witt and orbifold equivalence is explored for the positive, \strat{} \voa{}s of central charge $c=1$ in \autoref{ex:cc1wittorbifold}. Moreover, we consider the \strat{}, holomorphic \voa{}s of central charge $c=16,24$ in \autoref{ex:c=16orbifolds} and \autoref{ex:deepholes}, respectively.

\medskip

In \autoref{sec:wittgensym}, we discuss the relationship between Witt equivalence and generalized global symmetries, through a concept called symmetry-subalgebra duality (see also \cite{Rie25}). In particular, we argue that all the finite symmetries of the $\mathrm{SU}(2)_1$ Wess--Zumino--Witten model are invertible \cite{CLR24} (see \autoref{exph:SU(2)1symmetries}), and we also find two distinct Fibonacci symmetries of the moonshine module $V^\natural$ (see \autoref{exph:fibmonster}, and also \cite{FH24}, which appeared independently during the work on this manuscript). Along the way, we give a proof that the Fischer \voa{} $\smash{\textsl{VF}_{24}^\natural}$ \cite{HLY12} is \strat{}.

\medskip

Finally, in \autoref{sec:neighborhood} we study $p$-neighborhood of \voa{}s. This is a specialization of orbifold equivalence where we allow two \strat{} \voa{}s $V$ and $V'$ to be related by \emph{one} common \strat{} conformal subalgebra $W$ of index~$p$ for a prime $p$ (see \autoref{defi:neighbor}).

This definition is inspired by the notion of neighborhood for lattices \cite{Kne57}. For even lattices, the connected components of the $p$-neighborhood graph for a fixed prime $p$ are usually identical to the lattice genera (see \autoref{sec:lat}).

However, for \strat{} \voa{}s the picture is less clear. For simplicity, we restrict to the case of holomorphic \voa{}s. Naively, based on the dichotomy between orbifold and Witt equivalence and their inner versions (see \autoref{fig:summary}), one would expect that iterated $p$-neighborhood corresponds to the bulk genus and iterated inner $p$-neighborhood to the hyperbolic genus. However, the $p$-neighborhood graph of a bulk genus of holomorphic \voa{}s is usually disconnected (see \autoref{ex:schellekens3} and \cite{Mon98,HM23}), while one might suspect that it becomes connected if one is allowed to mix different primes~$p$ (see \autoref{rem:neighborbulk}).

The picture becomes clearer in the inner case. In \autoref{conj:hypgenpneighbor} we essentially propose that, except for finitely many primes, the inner $p$-neighborhood graph of a hyperbolic genus is connected for a fixed prime~$p$. We are able to prove a slightly weaker statement in \autoref{prop:hypgenpneighbor2}.

In \autoref{subsec:neighborexamples}, we provide empirical evidence for many of the formulated conjectures pertaining to $p$-neighborhood.


\subsection{For Physicists}\label{sec:physintro}

When are two quantum field theories (QFTs) related by topological manipulations? In this work, performing ``topological manipulations'' on a QFT $\mathcal{T}$ refers to operations like gauging a finite part of a generalized global symmetry of $\mathcal{T}$, or (relatedly) constructing a topological interface of finite quantum dimension between $\mathcal{T}$ and another QFT.

It is generally challenging to answer this question for a given pair of theories $\mathcal{T}$ and $\mathcal{T}'$. For example, the full symmetry structure of $\mathcal{T}$ will often be unknown, and even for those symmetries that have been identified, it can be quite involved to determine all the ways to gauge them, and whether any of these gaugings will connect $\mathcal{T}$ with $\mathcal{T}'$.

One natural way to circumvent such difficulties is to attempt to define readily computable physical invariants of a QFT that do not change when it is subjected to topological manipulations. Should one succeed, the matching of the invariants of $\mathcal{T}$ with those of $\mathcal{T}'$ would furnish a necessary condition for $\mathcal{T}'$ to be obtained from $\mathcal{T}$ by gauging. Ideally, these invariants would further be \emph{complete} in the sense that matching them would also be sufficient to demonstrate that $\mathcal{T}$ and $\mathcal{T}'$ are connected by topological manipulations, but this is much harder to achieve.

In the context of 2+1d topological quantum field theories (TQFTs), it was recently rigorously established that two theories can be related by a kind of generalized orbifold procedure \cite{MR23b} if and only if they can be separated by a topological interface \cite{Mul24}. Whether two TQFTs admit a topological interface between them is in turn sharply characterized within the framework of tensor categories via the notion of \emph{Witt equivalence} \cite{DNO13,DMNO13}, see \autoref{defn:Wittequivalencecategories}.

Witt equivalence is useful because it can be practically tested. For instance, there are computable observables of 2+1d TQFTs, known as higher central charges \cite{NSW19,NRWZ22,KKOSS22}, which define quantities that are constant along a Witt class. Therefore, given a pair of TQFTs, one can concretely check whether they have the same higher central charges: if some of the higher central charges disagree, then one can rigorously conclude that the theories cannot be related by topological manipulations, whereas if all of them agree, the test is generally inconclusive.\footnote{However, in the case of abelian TQFTs, the matching of just a finite number of higher central charges is both necessary and sufficient to establish that a pair of theories can be separated by a topological interface \cite{KKOSS22}.}

For 1+1d (not necessarily topological) QFTs, it was also recently shown that two theories are related by orbifolding \cite{DHVW85,DHVW86,BT18} if and only if they can be separated by a topological line interface of finite quantum dimension \cite{DLWW24}.\footnote{As we discuss in \autoref{sec:orb}, in order for this statement to be true, one must expand their notion of ``orbifolding'' to include the gauging of algebra objects of finite quantum dimension inside categories with \emph{infinitely} many simple topological line operators.} However, 1+1d QFTs are much more complicated objects than 2+1d TQFTs: for example, an infinite amount of data is generally required to determine a 1+1d QFT (like the scaling dimensions and operator product expansions of local operators in the case of conformal field theories), whereas the specification of a 2+1d topological phase requires only finitely many numbers (fusion coefficients, F-symbols, R-matrices, and so on). Correspondingly, criteria as tractable as Witt equivalence are generally lacking, and therefore testing whether two 1+1d QFTs are related by orbifolding, or can be separated by a topological line interface, is typically a hard problem, for the reasons described a few paragraphs ago. For example, there are almost no conformal field theories (CFTs) for which the \emph{full} category of topological line operators is known, setting aside the minimal models.\footnote{It is often erroneously assumed that the category of Verlinde lines of an RCFT is the full category of topological line defects of the theory. A simple counterexample is the $\mathrm{SU}(2)_1$ Wess--Zumino--Witten model: the unique non-trivial Verlinde line generates just a subgroup $\Z_2$ of the $\mathrm{SO}(4)$ symmetry group.}

\medskip

It is reasonable to expect that the situation is more favorable for 1+1d \emph{rational} conformal field theories (RCFTs), given their close relationship with 2+1d TQFTs \cite{Wit89,EMSS89}. To explain this relationship, recall that any 1+1d CFT possesses a left-moving chiral algebra $V$, consisting of the holomorphic local operators of the theory, as well as a right-moving chiral algebra $W$, consisting of the anti-holomorphic local operators of the theory \cite{Zam85,MS89}. In an RCFT, $V$ can be thought of as living on the boundary of a 2+1d TQFT, and likewise for $W$. The defining data of a 2+1d TQFT is a tuple $(\mathcal{C},c)$ consisting of a modular tensor category (MTC) $\mathcal{C}$, which encodes the algebraic properties of the anyons, and a rational number $c$ congruent to the chiral central charge of $\mathcal{C}$ modulo $8$ \cite{BK01,Tur10}. (We sometimes abusively abbreviate $(\mathcal{C},c)$ to just $\mathcal{C}$ when specifying $c$ is unimportant.) The 2+1d TQFT supporting $V$ on its boundary is $(\Rep(V),c(V))$, where $c(V)$ is the central charge of $V$ and $\Rep(V)$ is the representation category of $V$, which is known to admit the structure of a modular tensor category \cite{MS89,Hua08b}. Similar remarks apply to $W$.

The full $S^1$ Hilbert space $\mathcal{H}$ of an RCFT decomposes into finitely many irreducible representations of the chiral algebras,
\begin{equation*}
\mathcal{H}\cong \bigoplus_{i,j} M_{ij} V^i\otimes\overline{W}^j,
\end{equation*}
and $\mathcal{H}$ actually defines a Lagrangian algebra object in $\Rep(V)\boxtimes \overline{\Rep(W)}$ (see \autoref{sec:cat} for the definition of Lagrangian algebra). The defining data of an RCFT are $V$, $W$ and this Lagrangian algebra $\mathcal{H}$ \cite{FRS02}, and we compactly denote the corresponding CFT using the symbol ${_V}\mathcal{H}_W$, or simply $\mathcal{H}$ if we do not need to specify the precise chiral algebras.

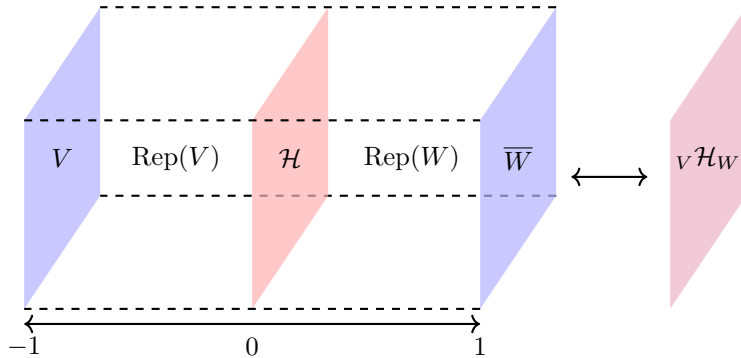
\begin{figure}[ht]
\begin{center}
\begin{tikzpicture}
\filldraw[black,fill=blue!30,, fill opacity=.7, thick,draw=none] (-6,0) -- (-6,-1-1.5) -- (-5,-1) -- (1-6,1.5) -- cycle;
\draw[black, thick,dashed](1,-1)--(-5,-1);
\filldraw[black,fill=red!30, fill opacity=.7, thick,draw=none] (-3,0) -- (1-3,1.5) -- (1-3,-1) -- (-3,-1-1.5) -- cycle;
\filldraw[black,fill=blue!30,fill opacity=.7, thick,draw=none] (-3+3,0) -- (1-3+3,1.5) -- (1-3+3,-1) -- (-3+3,-1-1.5) -- cycle;
\draw[black, thick,dashed] (-6,0)--(0,0);
\draw[black, thick,dashed] (0,-2.5) -- (-6,-1-1.5);
\draw[black, thick,dashed] (1,1.5)--(1-6,1.5);
\node[] at (-.5-.4,-.5) {$\Rep(W)$};
\node[] at (-.5-4+.5,-.5) {$\Rep(V)$};
\node[] at (-2.5,-.5) {$\mathcal{H}$};
\node[] at (-2.5-3,-.5) {$V$};
\node[] at (-2.5+3,-.5) {$\overline{W}$};
\draw[<->,thick] (-6,-2.7) -- (0,-2.7);
\node[] at (-6,-3) {$-1$};
\node[] at (-3,-3) {$0$};
\node[] at (0,-3) {$1$};

\tikzstyle{s}=[draw,decorate,<->,thick]
\path[s] (1.2,-.75) to node[anchor=south,shift={(0.1,0)}] {} (2.2,-.75);
\filldraw[black,fill=purple!30,fill opacity=.7, thick,draw=none] (-3+6-.5,0) -- (1-3+6-.5,1.5) -- (1-3+6-.5,-1) -- (-3+6-.5,-1-1.5) -- cycle;
\node[] at (-2.5+6-.5,-.5) {${_V}\mathcal{H}_W$};
\end{tikzpicture}
\end{center}
\caption{The 2+1d representation of an RCFT ${_V}\mathcal{H}_W$.}\label{fig:KS}
\end{figure}

There is a useful 2+1d representation of a 1+1d RCFT \cite{FRS02,KS11}, depicted in \autoref{fig:KS}. For ease of exposition, let us restrict to theories with vanishing gravitational anomaly, $c(V)-c(W)=0$. The ingredients of this picture are as follows.
\begin{enumerate}
\item One works with a spacetime manifold of the form $\Sigma\times [-1,1]$, with $\Sigma$ a $2$-dimensional surface. The region $\Sigma\times [0,1]$ is in the phase defined by the TQFT $\Rep(W)$, and $\Sigma\times [-1,0]$ is in the phase defined by $\Rep(V)$.
\item At the boundaries $\Sigma\times\{1\}$ and $\Sigma\times \{-1\}$, we impose gapless chiral boundary conditions defined by $\overline{W}$ and $V$, respectively.
\item On the surface $\Sigma \times \{0\}$, a topological interface determined by $\mathcal{H}$ interpolates between the two configurations.\footnote{A Lagrangian algebra of $\Rep(V)\boxtimes \overline{\Rep(W)}$ defines a boundary condition of the corresponding TQFT. By unfolding along this boundary, one obtains a topological interface between $\Rep(V)$ and $\Rep(W)$.}
\end{enumerate}
Upon dimensional reduction of this configuration to 1+1d along the interval direction, one recovers the RCFT ${_V}\mathcal{H}_W$.

Now, the existence of the topological interface determined by $\mathcal{H}$ demonstrates that, by definition, $\Rep(V)$ is Witt equivalent to $\Rep(W)$. This allows us to introduce the following important invariant of ${_V}\mathcal{H}_W$,
\begin{equation*}
\mathsf{W}(\mathcal{H})\coloneqq [\Rep(V)]=[\Rep(W)],
\end{equation*}
where $[\mathcal{C}]$ is the Witt class of the modular tensor category $\mathcal{C}$. This invariant enters prominently in the following proposal, which is one of our main results.

\smallskip

\noindent\textbf{\autoref{thm:orbimplieswittRCFT}.} \emph{The triple $(\mathsf{W}(\mathcal{H}),c(V),c(W))$ associated with an RCFT ${_V}\mathcal{H}_W$ is invariant under topological manipulations. Thus, if $(\mathsf{W}(\mathcal{H}),c(V),c(W))$ is not equal to $(\mathsf{W}(\mathcal{H}'),c(V'),c(W'))$, then the RCFTs ${_V}\mathcal{H}_W$ and ${_{V'}}\mathcal{H}'_{W'}$ cannot be related by gauging a finite part of a generalized global symmetry, or separated by a topological line interface of finite quantum dimension.}

\smallskip

The power of \autoref{thm:orbimplieswittRCFT} rests on the fact that the triple $(\mathsf{W}(\mathcal{H}),c(V),c(W))$ is a very coarse shadow of a full RCFT and correspondingly is far easier to work with. Furthermore, category-theoretic techniques from 2+1d TQFT can be used to compute $\mathsf{W}(\mathcal{H})$, and matching this invariant is generally more effective than trying to discover whether there exists a sequence of orbifolds, or a topological line interface, which interpolates between two given 1+1d QFTs.

As an example, \autoref{thm:orbimplieswittRCFT} shows that there is no finite symmetry of the $\mathrm{SU}(2)_1$ Wess--Zumino--Witten model that can be orbifolded to obtain the (bo\-son\-ized) free Dirac fermion, nor a topological interface of finite quantum dimension interpolating between the two, in spite of the fact that both theories are RCFTs with left- and right-moving central charges equal to $1$. This follows immediately from the fact that the Witt class of $\mathrm{SU}(2)_1(=\mathrm{U}(1)_2)$ Chern--Simons theory is different from that of $\mathrm{U}(1)_4$ Chern--Simons theory. See the analysis of the $\mathrm{U}(1)_{2N_1}\times \mathrm{U}(1)_{-2N_2}$ TQFT in Section~2.2 of \cite{KKOSS22} in terms of higher central charges, and \autoref{exph:orbifoldingc=1} for an alternative approach.

Let us give some intuition for where \autoref{thm:orbimplieswittRCFT} comes from. We present a heuristic physical picture here, and relegate a more mathematically rigorous treatment to the main text. The basic starting point is the assertion that any configuration that one can draw involving 1+1d RCFTs can be inflated to a suitable 2+1d configuration compactified on an interval, which moreover recovers the starting configuration upon dimensional reduction back to 1+1d. \autoref{fig:KS} is the most basic example of this.

Suppose then that two RCFTs ${_V}\mathcal{H}_W$ and ${_{V'}}\mathcal{H}'_{W'}$ can be related by orbifolding. By starting with ${_V}\mathcal{H}_W$ and performing the orbifold in half of spacetime, one obtains a topological line interface $I$ between ${_V}\mathcal{H}_W$ and ${_{V'}}\mathcal{H}'_{W'}$. Upon inflating to 2+1d, one expects to find the following.\footnote{This proposal is similar in spirit to the one in \cite{CRZ24}, that (not necessarily topological) line interfaces between 1+1d QFTs with symmetry inflate to topological interfaces between the corresponding symmetry TQFTs. See also \cite{Cop26,CDHHT24,CHO25,BCPS25,GHU25} for related ideas.}
\begin{enumerate}
\item The topological line interface $I$ inflates to a pair of $2$-dimensional topological interfaces $\mathcal{I}_1$ and $\mathcal{I}_2$, where $\mathcal{I}_1$ interpolates between $\Rep(V)$ and $\Rep(V')$, and $\mathcal{I}_2$ interpolates between $\Rep(W)$ and $\Rep(W')$.
\item The topological interface $\mathcal{I}_1$ terminates on a topological line junction $I_1$ between the gapless chiral boundary conditions defined by $V$ and $V'$, and likewise the topological interface $\mathcal{I}_2$ terminates on a topological line junction $I_2$ between the boundaries defined by $W$ and $W'$.
\item The interfaces $\mathcal{I}_1$ and $\mathcal{I}_2$ meet in the middle of the interval, where they form a kind of four-way junction with the topological interfaces $\mathcal{H}$ and $\mathcal{H}'$. We do not bother to give a name to the line junction that decorates the intersection of these four interfaces.
\end{enumerate}
See \autoref{fig:KSinterface} for a picture of this setup. The existence of, say, the interface $\mathcal{I}_1$ shows that $\Rep(V)$ and $\Rep(V')$ are Witt equivalent, and hence we recover the condition in \autoref{thm:orbimplieswittRCFT} that $\mathsf{W}(\mathcal{H})=\mathsf{W}(\mathcal{H}')$. The fact that $c(V)=c(V')$ and $c(W)=c(W')$ in order for the RCFTs $\mathcal{H}$ and $\mathcal{H}'$ to be related by topological manipulations follows from the standard fact that orbifolding does not change the central charges of a theory.

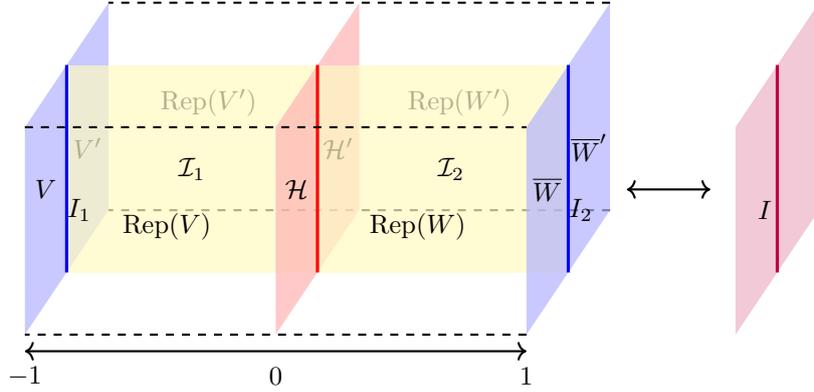
\begin{figure}[ht]
\begin{center}
\begin{tikzpicture}[scale=1.1]
\draw[black,fill=blue!30,fill opacity=.7, thick,draw=none] (-6,0) -- (-6,-1-1.5) -- (-5,-1) -- (1-6,1.5)  -- cycle;

\node[] at (-2.5-3+.25,-.5+.25) {$V'$};
\draw[black, thick,dashed](1,-1)--(-5,-1);
\node[] at (-.5-4+.5+.2,-.5-.5+1.3) {$\Rep(V')$};
\node[] at (-.5-4+.5+.2+3,-.5-.5+1.3) {$\Rep(W')$};

\draw[draw=black,thick,fill=yellow!30,fill opacity=.7,draw=none] (-6+.5,.75) -- (.5-3,.75) -- (.5-3,-2.5+.75) -- (-6+.5,-2.5+.75) -- cycle;
\draw[very thick,blue] (-6+.5,-2.5+.75)--(-6+.5,.75);
\filldraw[black,fill=red!30,fill opacity=.7, thick,draw=none] (-3,0) -- (1-3,1.5) -- (1-3,-1) -- (-3,-1-1.5) -- cycle;
\node[] at (-2.5+.25,-.5+.25) {$\mathcal{H}'$};
\draw[draw=black, thick,fill=yellow!30,fill opacity=.7,draw=none] (-3+.5,.75) -- (.5,.75) -- (.5,-2.5+.75) -- (-3+.5,-2.5+.75) -- cycle;
\draw[very thick,red] (-3+.5,-2.5+.75)--(-3+.5,.75);\filldraw[black,fill=blue!30,fill opacity=.7, thick,draw=none] (-3+3,0) -- (1-3+3,1.5) -- (1-3+3,-1) -- (-3+3,-1-1.5) -- cycle;

\draw[black, thick,dashed] (-6,0)--(0,0);
\draw[black, thick,dashed] (0,-2.5) -- (-6,-1-1.5);
\draw[black, thick,dashed] (1,1.5)--(1-6,1.5);
\node[] at (-.5-.4,-.5) {$\mathcal{I}_2$};
\node[] at (-.5-4+.5,-.5) {$\mathcal{I}_1$};
\node[] at (-2.5-.25,-.5-.25) {$\mathcal{H}$};
\node[] at (-2.5-3-.25,-.5-.25) {$V$};
\node[] at (-2.5+3-.25,-.5-.25) {$\overline{W}$};
\node[] at (-2.5+3+.25,-.5+.25) {$\overline{W}'$};
\node[] at (-.5-4+.5-.5+.2,-.5-.5-.2) {$\Rep(V)$};
\node[] at (-.5-4+.5-.5+.2+3,-.5-.5-.2) {$\Rep(W)$};
\node[] at (-6+.65,-1) {$I_1$};
\node[] at (.65,-1) {$I_2$};
\draw[<->,thick] (-6,-2.7) -- (0,-2.7);
\node[] at (-6,-3) {$-1$};
\node[] at (-3,-3) {$0$};
\node[] at (0,-3) {$1$};

\tikzstyle{s}=[draw,decorate,<->,thick]
\path[s] (1.2,-.75) to node[anchor=south,shift={(0.1,0)}] {} (2.2,-.75);
\filldraw[black,fill=purple!30, fill opacity=.7, thick,draw=none] (-3+6-.5,0) -- (1-3+6-.5,1.5) -- (1-3+6-.5,-1) -- (-3+6-.5,-1-1.5) -- cycle;
\draw[very thick,purple] (-3+6-.5+.5,0+.75)--(-3+6-.5+.5,.75-1-1.5);
\node[] at (-2.5+6-.15-.5,-.5-.5) {$I$};
\draw[very thick,blue] (.5,-2.5+.75)--(.5,.75);

\end{tikzpicture}
\end{center}
\caption{The 2+1d representation of a topological line interface $I$ between two 1+1d RCFTs ${_V}\mathcal{H}_W$ and ${_{V'}}\mathcal{H}'_{W'}$.}\label{fig:KSinterface}
\end{figure}

It is natural to wonder whether $(\mathsf{W}(\mathcal{H}),c(V),c(W))$ defines a \emph{complete} invariant. In other words, is it guaranteed that ${_V}\mathcal{H}_W$ and ${_{V'}}\mathcal{H}'_{W'}$ can be related by topological manipulations if their invariants $(\mathsf{W}(\mathcal{H}),c(V),c(W))$ and $(\mathsf{W}(\mathcal{H}'),c(V'),c(W'))$ match?

The reason that this question is difficult to answer can be seen in 2+1d. Indeed, the matching of these invariants guarantees that one can construct topological interfaces $\mathcal{I}_1$ and $\mathcal{I}_2$ as in \autoref{fig:KSinterface}, but it is much more subtle to determine whether $\mathcal{I}_1$ and $\mathcal{I}_2$ possess the necessary topological line junctions $I_1$ and $I_2$ to complete the picture, which is a more dynamical question. Nonetheless, we codify our expectation into the following conjecture.

\smallskip

\noindent\textbf{\autoref{conj:wittorbRCFT}.} \emph{If two RCFTs ${_V}\mathcal{H}_W$ and ${_{V'}}\mathcal{H}'_{W'}$ have matching invariants $(\mathsf{W}(\mathcal{H}),c(V),c(W))$ and $(\mathsf{W}(\mathcal{H}'),c(V'),c(W'))$, respectively, then one can be obtained from the other via topological manipulations. More succinctly, the triple $(\mathsf{W}(\mathcal{H}),c(V),c(W))$ defines a complete invariant.}

\smallskip

There are a number of reasons to believe this conjecture.

First of all, we show in \autoref{prop:witttrivialorbifolding}, using the results of \cite{MS23}, that \autoref{conj:wittorbRCFT} is true in the case that $c(V),c(W)\leq 24$ and $\mathsf{W}(\mathcal{H})=[\Vect]$ if one assumes the widely believed conjecture concerning the uniqueness of the moonshine module, \autoref{conj:moonshineuniqueness}. In particular, for sufficiently low central charges, \autoref{prop:witttrivialorbifolding} is true for holomorphically factorizing CFTs, and more generally any RCFT whose, e.g., left-moving chiral algebra is attached to a bulk 2+1d TQFT that admits a topological boundary condition.

Second, it turns out that by carefully translating a certain mathematical result about even lattices (\autoref{prop:lat_ratwit}) into physics, one can deduce that \autoref{conj:wittorbRCFT} holds for any pair of \emph{free} RCFTs. (By free RCFT ${_V}\mathcal{H}_W$, we mean a theory for which $c(V)+c(W)$ is equal to the number of linearly independent $\mathrm{U}(1)$ symmetries.) We generalize this idea to obtain the following partial result in the direction of \autoref{conj:wittorbRCFT}.

\smallskip

\noindent\textbf{\autoref{thmph:rcftinnerorb=stronglywitt}.} \emph{\autoref{conj:wittorbRCFT} holds for any two RCFTs ${_V}\mathcal{H}_W$ and ${_{V'}}\mathcal{H}'_{W'}$ for which $V/\mathrm{U}(1)^r\cong V'/\mathrm{U}(1)^{r'}$ and $W/\mathrm{U}(1)^s\cong W'/\mathrm{U}(1)^{s'}$, where $r$, $r'$, $s$ and $s'$ are the ranks of the Lie algebras formed by the spin-1 Noether currents in $V$, $V'$, $W$ and $W'$, respectively.}

\smallskip

The chiral algebra $V/\mathrm{U}(1)^r$ is the coset of $V$ by its maximal $\mathrm{U}(1)^r$ Kac--Moody algebra. The $\mathrm{U}(1)^r$ Kac--Moody algebra admits a description in terms of chiral free bosons, and therefore defines what we call the \emph{free sector} of $V$. By quotienting it out, we may think of what is left over, i.e.\ $V/\mathrm{U}(1)^r$, as a kind of \emph{interacting sector} of $V$. Thus, \autoref{thmph:rcftinnerorb=stronglywitt} essentially asserts that \autoref{conj:wittorbRCFT} holds for any pair of RCFTs whose interacting sectors coincide. The fact that \autoref{conj:wittorbRCFT} holds for free theories is recovered from \autoref{thmph:rcftinnerorb=stronglywitt} in the special case that the interacting sectors of $V$, $V'$, $W$ and $W'$ are all trivial. This condition may appear baroque, but we explain in \autoref{exph:generalizednarain} that it is satisfied by pairs of theories that are related by current-current (i.e.\ $J\bar J$) deformations.

\medskip

Having sketched our results, we now explain their significance. One source of motivation comes from the problem of classifying RCFTs. Orbifolding \cite{DHVW85,DHVW86} is one of the two main procedures by which one can produce new CFTs from known ones (the other being the method of taking cosets \cite{GKO85,GKO86}), and it is natural to ask how effective this procedure is at moving one around theory space. Indeed, throughout the decades, there have been various proposals for how these two tools, orbifolding and cosets, can be used in combination to generate all RCFTs from some smaller set of seed theories. We claim that, on the strength of one additional mild assumption (called weak reconstruction, \autoref{conj:weakreconstruction}), one is led to the following ``structure theorem'' on RCFTs.

\smallskip

\noindent\textbf{\autoref{thmph:classification}.} \emph{Assuming \autoref{conj:wittorbRCFT} and weak reconstruction, \autoref{conj:weakreconstruction}, any RCFT can be obtained from the affine Kac--Moody algebra $\smash{\mathsf{E}_{8,1}^{\otimes n}\otimes\overline{\mathsf{E}}_{8,1}^{\otimes m}}$ for some $n,m\in\N$ by performing topological manipulations and then quotienting decoupled degrees of freedom.}

\smallskip

Here, by ``quotienting decoupled degrees of freedom'' we refer to the operation $\mathcal{T}_1\otimes\mathcal{T}_2\to \mathcal{T}_1$ of passing from a trivially decoupled tensor product of RCFTs to one of the tensor factors. Although it is not always terribly practical, this structural result gives one an abstract sense for what all RCFTs are expected to ``look like'', at least if one believes our \autoref{conj:wittorbRCFT}.

\smallskip

\autoref{conj:wittorbRCFT} is also related to other fundamental conjectures in QFT. For example, it has been suggested \cite{Sei19} that any two QFTs in the same spacetime dimension, with the same symmetries and the same anomalies, can be connected by deformations. Here, ``deformations'' include processes that move ``up and down RG flows'' in the sense described in \cite{GJW21}. As a special case, this conjecture sets the expectation that any QFT $\mathcal{T}$ with vanishing gravitational anomalies should admit some boundary condition. Indeed, the conjecture asserts that a theory $\mathcal{T}$ without gravitational anomalies should be deformable to the trivially gapped theory $\mathcal{T}_{\mathrm{triv}}$; by activating such a deformation in half of spacetime, one obtains an interface between $\mathcal{T}$ and $\mathcal{T}_{\mathrm{triv}}$, which can be equivalently thought of as a boundary of $\mathcal{T}.$

Establishing this conjecture even in the narrow setting of RCFTs is a formidable challenge. Indeed, although one might naively believe that any RCFT admits Cardy boundary conditions \cite{Car89}, this construction only applies to theories whose left- and right-moving chiral algebras coincide, and misses out more general ``heterotic'' RCFTs. We demonstrate that settling this lore is no more and no less difficult than proving \autoref{conj:wittorbRCFT}.\footnote{We thank Zohar Komargodski, Sahand Seifnashri and Shu-Heng Shao for inspiring discussions about their unpublished work \cite{KSS22}, in which they obtained a special case of \autoref{thmph:recharacterizationmainconj} that applies to holomorphically factorizing CFTs.}

\smallskip

\noindent\textbf{\autoref{thmph:recharacterizationmainconj}.} \emph{Our main conjecture, \autoref{conj:wittorbRCFT}, is equivalent to \autoref{conj:RCFTboundaries}, which asserts that every RCFT with vanishing gravitational anomaly admits a boundary condition with finite $g$-function (or boundary entropy).}

\smallskip

Finally, it is natural to extend the notion of Witt equivalence of 3d TQFTs to 2d RCFTs by defining two theories ${_V}\mathcal{H}_W$ and ${_{V'}}\mathcal{H}'_{W'}$ to be Witt equivalent if their invariants $(\mathsf{W}(\mathcal{H}),c(V),c(W))$ and $(\mathsf{W}(\mathcal{H}'),c(V'),c(W'))$ are equal. A natural question then is to understand the space of all theories that are Witt equivalent to a fixed RCFT ${_V}\mathcal{H}_W$, i.e.\ to understand the Witt class of ${_V}\mathcal{H}_W$. We explain in \autoref{prop:rcftsymmetries} how this problem is closely related to the problem of classifying the global symmetries of ${_V}\mathcal{H}_W$. We emphasize that this perspective has practical implications by showing that it leads to a proof of the following result.

\smallskip

\noindent\textbf{\autoref{exph:SU(2)1symmetries}.} \emph{All the finite global symmetries of the $\mathrm{SU}(2)_1$ Wess--Zumino--Witten model are invertible.}


\subsection*{Notation}
All Lie algebras and vertex (operator) algebras will be over the base field $\C$. All categories will be enriched over $\Vect=\Vect_\C$. Usually, $q$ denotes a formal variable or the complex variable $q= e^{2\pi i\tau}$ for $\tau\in\mathbb{H}$ with the complex upper half-plane $\mathbb{H}\coloneqq\{z\in\C\mid\operatorname{Im}(z)>0\}$.

This is a mathematical manuscript. However, from time to time we use the notation Theorem$^{\mathrm{ph}}$, Proposition$^{\mathrm{ph}}$, etc.\ to denote statements that are established or well-defined only at a physics level of rigor.


\subsection*{Acknowledgments}
We thank Yichul Choi, Thomas Creutzig, Aaron Hofer, Gerald Höhn, Hannes Knötzele, Zohar Komargodski, Robert McRae, Yuto Moriwaki, Nils Scheithauer, Sahand Seifnashri, Shu-Heng Shao, Yifan Wang, Xiao-Gang Wen, Harshit Yadav and Hiroshi Yamauchi for helpful discussions. We thank the referees for their many valuable comments and suggestions.

Sven Möller acknowledges support from the DFG through the Emmy Noether Program and the CRC 1624 \emph{Higher Structures, Moduli Spaces and Integrability}, project numbers 460925688 and 506632645, respectively. Brandon Rayhaun gratefully acknowledges NSF grant PHY-2210533.


\section{Preliminaries}\label{sec:prelim}

In this section, we review some background material on lattices, modular tensor categories, \voa{}s, and their relationships to concepts in physics. We will also assume much of the content covered in \cite[\autoref*{MR1sec:prelim}]{MR24a}.


\subsection{Lattices and Metric Groups}\label{sec:lat}

Recall the definition of an even lattice $L$, and the notion of its dual lattice $L'\supset L$, which is generally not even. Further recall that a metric group $D=(D,Q)$ is a finite abelian group $D$ equipped with a non-degenerate quadratic form $Q\colon D\to \mathbb{Q}/\Z$, which also defines an associated symmetric bilinear form $B$. The discriminant form of $L$ is the metric group defined by $L'/L$ with the induced quadratic form \cite{Nik80}, and the discriminant of $L$ is $d(L) =|L'/L|$. For more details, see also \cite[\autoref*{MR1sec:lat}]{MR24a}.

\begin{defi}[Lattice Genus]\label{defi:latticegenus}
Two even lattices $L$ and $M$ belong to the same genus if any of the following equivalent conditions hold.
\begin{enumerate}[label=(\alph*)]
\item\label{item:latgen1} $L\otimes_\Z\R\cong M\otimes_\Z\R$ and  $L\otimes_\Z\Z_p\cong M\otimes_\Z\Z_p$ for all primes $p$.
\item\label{item:latgen2} $L'/L\cong M'/M$ as metric groups and $\sign(L)=\sign(M)$.
\item\label{item:latgen3} $L\oplus\II_{1,1}\cong M\oplus\II_{1,1}$.
\end{enumerate}
Here, and only here, $\Z_p$ denotes the ring of $p$-adic integers, and $\II_{1,1}$ is the unique even, unimodular lattice of signature $(1,1)$.
\end{defi}

If $D$ is a metric group, and $H$ is an isotropic subgroup (i.e.\ a subgroup for which $Q(h)=0+\Z$ for all $h\in H$), then $H$ is contained in its orthogonal complement $H^\perp$ and the \emph{subquotient} $H^\perp/H$ inherits the structure of a metric group. An isotropic subgroup $H$ of the discriminant form $L'/L$ of an even lattice~$L$ defines an extension of $L$ to a larger even lattice $K\subset L'$ of the same rank,
\begin{equation}\label{eqn:lattice_ext}
K\coloneqq\bigcup_{\lambda+L\in H}(\lambda+L).
\end{equation}
Furthermore, all same-rank extensions $K$ of $L$ are obtained in this way. The discriminant form of $K$ is given by the subquotient $K'/K\cong H^\perp/H$.

\medskip

Another useful notion is that of Witt equivalence of metric groups \cite{Wit37}.
\begin{defi}[Witt Equivalence]\label{defi:metric_witt}
Two metric groups $D$ and $D'$ are \emph{Witt equivalent} if they admit isometric subquotients, i.e.\ if there are isotropic subgroups $H\subset D$ and $H'\subset D'$ such that $H^\perp/H\cong(H')^\perp/H'$ as metric groups.
\end{defi}
This defines an equivalence relation. It turns out that each Witt class of metric groups has a unique anisotropic representative (i.e.\ a representative not having any isotropic subgroups). It can be reached from any representative by taking the subquotient associated with a maximal isotropic subgroup (see, e.g., Lemma~A.31 of \cite{DGNO10}).

As a special case, one can look at self-dual, isotropic subgroups $H$, i.e.\ those isotropic subgroups satisfying $H=H^\perp$. In other words, the corresponding subquotient $H^\perp/H=\{0\}$ is trivial. A metric group is in the trivial Witt class if and only if it possesses a self-dual, isotropic subgroup or equivalently a trivial subquotient.

Evidently, any self-dual, isotropic subgroup $H$ of a given metric group $D$ is maximal. The converse is in general not true. However, if $D$ possesses at least one self-dual, isotropic subgroup (meaning that $D$ is in the trivial Witt class), then every isotropic subgroup is contained in some self-dual, isotropic subgroup. Hence, in that case, any maximal isotropic subgroup is self-dual.

\medskip

We extend the notion of Witt equivalence to even lattices as follows.
\begin{defi}[Witt Equivalence]\label{defi:lattice_witt}
Two even lattices $L$ and $M$ are \emph{Witt equivalent} if their discriminant forms $L'/L$ and $M'/M$ are Witt equivalent as metric groups and if they have the same signature, $\sign(L)=\sign(M)$.
\end{defi}
This defines an equivalence relation. It can be reformulated, following the discussion around equation~\eqref{eqn:lattice_ext}, as saying that the even lattices $L$ and $M$ admit same-rank extensions to even lattices $N_1$ and $N_2$, respectively, in the same genus, as depicted in \autoref{fig:wittlat}. As a special case, if $M$ is a full-rank sublattice of $L$, then they are in the same Witt class.
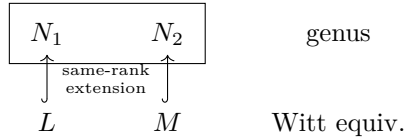
\begin{figure}[ht]
\centering
\begin{tikzcd}
[/tikz/execute at end picture={
\node (large) [rectangle, draw, fit=(A3) (A4)] {};
}]
|[alias=A3]| N_1 & |[alias=A4]| N_2&\text{genus}\\
|[alias=A1]| L \arrow[u,hook,"~\shortstack{\tiny same-rank\\\tiny extension}" right] & |[alias=A2]| M\arrow[u,hook]&\text{Witt equiv.}
\end{tikzcd}
\caption{Witt equivalence of even lattices $L$ and $M$.}
\label{fig:wittlat}
\end{figure}

There is also the notion of rational equivalence (see, e.g., \cite{Gan91}), which we depict in \autoref{fig:ratlat}.
\begin{defi}[Rational Equivalence]\label{defi:lattice_rat}
Two even lattices $L$ and $M$ are said to be \emph{rationally equivalent} if any one of the following two conditions hold.
\begin{enumerate}[label=(\alph*)]
\item $L\otimes_\Z\Q\cong M\otimes_\Z\Q$ as $\Q$-vector spaces with symmetric bilinear forms.
\item There is an even lattice $K$ that is isomorphic to full-rank (or finite-index) sublattices of $L$ and $M$.
\end{enumerate}
\end{defi}

\begin{figure}[ht]
\centering
\begin{tikzcd}[column sep=small]
L& & M&\text{Rational equiv.}\\
& K\arrow[hookrightarrow]{ul}{\text{ext.}}\arrow[hookrightarrow]{ur}[swap]{\text{ext.}}
\end{tikzcd}
\caption{Rational equivalence of even lattices $L$ and $M$.}
\label{fig:ratlat}
\end{figure}
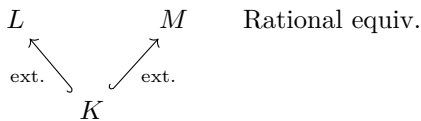

\begin{rem}\label{rem:trans}
Once again, rational equivalence is an equivalence relation. Indeed, it is not difficult to see that sharing a full-rank sublattice is a transitive property. Suppose that $L_1$ and $L_2$ have the common full-rank sublattice $K_1$ (which we assume to be explicitly realized as a sublattice of $L_2$) and that $L_2$ and $L_3$ have the common full-rank sublattice $K_2$ (again assumed to be realized as a sublattice of $L_2$). Then the intersection $K_1\cap K_2$ is a full-rank sublattice of $K_1$ and $K_2$ and hence isomorphic to full-rank sublattices of $L_1$ and of $L_3$, proving the transitivity.

We discuss this rather simple fact because this transitivity property will fail for the corresponding notion for \strat{} \voa{}s (see \autoref{defi:orbequiv}, \autoref{fig:orbequiv} and \autoref{rem:nontrans}).
\end{rem}

By definition, for even lattices $L$ and $M$, being isomorphic (or integrally equivalent) and being $p$-neighbors (see below) each imply rational equivalence. Also, it follows from the Hasse principle that $L$ and $M$ are rationally equivalent if and only if all the local quadratic spaces $L\otimes\Q_p$ and $M\otimes\Q_p$ for $p$ a prime including $\infty$ are isomorphic. In particular, lattices in the same genus are rationally equivalent.

We now show that Witt equivalence of even lattices is the same as rational equivalence. However, we have given these two notions separate names because it is far more subtle to establish whether they continue to be the same once they are generalized to the setting of \voa{}s.

\begin{prop}\label{prop:lat_ratwit}
Two even lattices are rationally equivalent if and only if they are Witt equivalent.
\end{prop}

\begin{proof}
The forward direction is almost immediate. Assume that $L_1$ and $L_2$ are rationally equivalent even lattices. Thus, there exists an even lattice $K$ which participates as a full-rank sublattice of both $L_1$ and $L_2$. Hence, by the comment following \autoref{defi:lattice_witt}, $L_1$ and $L_2$ each belong to the same Witt class as $K$.

\smallskip

We show the reverse direction. That is, assume that the even lattices $L_1$ and $L_2$ belong to the same Witt class. Then, by picking maximal isotropic subgroups $H_i<L_i'/L_i$, one obtains new even lattices $M_i$ whose discriminant forms are given by the corresponding subquotients defined by the $H_i$, i.e.\ $M_i'/M_i\cong H_i^\perp/H_i$. Because $L_1$ is a full-rank sublattice of $M_1$, it follows that $L_1$ is rationally equivalent to $M_1$. Similar comments apply replacing $L_1$ with $L_2$ and $M_1$ with $M_2$. Thus, if we demonstrate that $M_1$ is rationally equivalent to $M_2$, it will immediately follow that $L_1$ is rationally equivalent to $L_2$.

To see this, note that $M_1$ and $M_2$ belong to the same genus. Indeed, they have the same signature. Furthermore, the $M_i$ have anisotropic discriminant forms because the $H_i$ were chosen to be maximal. On the other hand, each Witt class of metric groups admits a unique anisotropic representative. The discriminant forms of the $M_i$ belong to the same Witt class because they are obtained as subquotients of metric groups which belong to the same Witt class by assumption. So it follows that $M_1'/M_1\cong M_2'/M_2$, and hence that $M_1$ and $M_2$ belong to the same genus. On the other hand, we remarked above that being in the same genus implies rational equivalence. This completes the reverse direction.
\end{proof}

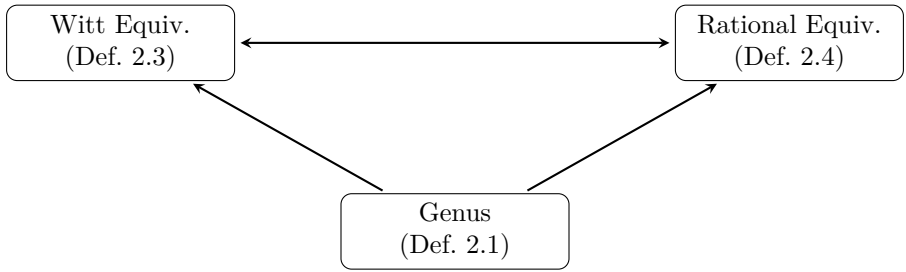
\begin{figure}[ht]
\begin{tikzpicture}
\def \tt {-2}
\tikzstyle{equiv}=[draw,rectangle, rounded corners, minimum width=3cm, minimum height=1cm]
\tikzstyle{a}=[thick,->,>=stealth,shorten >=2pt,shorten <=2pt]
\tikzstyle{t}=[thick,<->,>=stealth,shorten >=2pt,shorten <=2pt]
\tikzstyle{b}=[thick,<->,>=stealth,shorten >=2pt,shorten <=2pt]
\tikzstyle{d}=[thick,dashed,->,>=stealth,shorten >=2pt,shorten <=2pt]
\tikzstyle{dd}=[thick,double equal sign distance, -Implies]
\node[equiv,align=center] (lattice_witt) at (0,-9+\tt) {Witt Equiv.\ \\ (Def.\ \ref{defi:lattice_witt})};
\node[equiv,align=center] (lattice_rat) at (8.8,-9+\tt) {Rational\ Equiv.\ \\ (Def.\ \ref{defi:lattice_rat})};
\node[equiv,align=center] (lattice_gen) at (4.4,-11.5+\tt) {Genus \\ (Def.\ \ref{defi:latticegenus})};
\draw[a] (lattice_gen) -- (lattice_witt);
\draw[a] (lattice_gen) -- (lattice_rat);
\draw[t] (lattice_witt) -- (lattice_rat);
\end{tikzpicture}
\caption{Interrelations between the various equivalence relations on even lattices.}
\label{fig:lat}
\end{figure}

We summarize the different ways to organize even lattices into equivalence classes in \autoref{fig:lat}. One key mathematical goal of this text is to generalize these notions appropriately to (\strat{}) \voa{}s. This was done for the notion of genus in \cite{Hoe03,Mor21,MR24a}, and is carried out below for Witt equivalence and rational equivalence in \autoref{sec:witt} and \autoref{sec:orb}, respectively. This goal is motivated by the observation that every even lattice $L$ defines a lattice conformal vertex algebra $V_L$ (which is a \strat{} \voa{} if $L$ is positive-definite), allowing us to translate these lattice-theoretic notions across this construction.

We shall observe that each lattice equivalence in \autoref{fig:lat} can be generalized to (at least) two inequivalent notions in the setting of \voa{}s, a finer one that is morally still a lattice concept and a coarser one that we claim is a more honest vertex algebraic generalization (see \autoref{fig:summary}).

\medskip

Finally, we also recall the notion of lattice neighborhood \cite{Kne57}.
\begin{defi}[Neighborhood]\label{defi:latneigh}
For a positive integer $n\in\Ns$, two even lattices $L$ and $M$ (or isomorphism classes thereof) are called \emph{$n$-neighbors}, necessarily of the same discriminant $d(L)=d(M)$, if they contain a common full-rank sublattice $K$ in both of them satisfying $L/K\cong \mathbb{Z}_n$ and $M/K\cong \mathbb{Z}_n$ as groups.
\end{defi}
That is, $K'/K$ has two isotropic subgroups isomorphic to $\Z_n$ and the corresponding extensions of $K$ are (isomorphic to) $L$ and $M$. The definition is depicted in \autoref{fig:latneigh}.

\begin{figure}[ht]
\centering
\begin{tikzcd}
L\arrow[-,dashed]{rr}{n\text{-neighbor}}& & M\\
& K\arrow[hookrightarrow]{ul}{\Z_n\text{-ext.}}\arrow[hookrightarrow]{ur}[swap]{\Z_n\text{-ext.}}
\end{tikzcd}
\caption{Neighborhood of even lattices $L$ and $M$.}
\label{fig:latneigh}
\end{figure}

The lattice genus of $L$, \autoref{defi:latticegenus}, is closely related to the connected component of $L$ in the lattice $p$-neigh\-bor\-hood graph when $p$ is a prime not dividing $d(L)$. For instance, if $L$ and $M$ are even lattices that are iterated $p$-neighbors for some prime $p$ with $p\nmid d(L)$, then $L$ and $M$ belong to the same genus.

Conversely, suppose that $L$ and $M$ are even lattices of rank at least $3$ belonging to the same genus, and further assume that this genus consists of only one spinor genus.\footnote{In general, the genus of a lattice decomposes into $2^k$ many spinor genera for some $k\in\N$. Often, $k=0$, i.e.\ genus and spinor genus coincide. See, e.g., \cite{Ome71}.} Then $L$ and $M$ are iterated $p$-neighbors of one another for any odd prime~$p$ with $p\nmid d(L)$; see, e.g., \cite{Sch91} and the references cited therein (or \cite{Sch22,Che22}). More precisely, $p$-neighbors of $L$ for a fixed prime $p$ cover (at least) the whole spinor genus of $L$, while different spinor genera in the genus of $L$ may be connected by $p$-neigh\-bor\-hood for some specially selected primes $p$ \cite{BH83,Sch91}. We summarize this as follows: considering two even lattices $L$ and $M$ of rank at least $3$, if they belong to the same genus, then they are iterated $p$-neighbors of one another for possibly different primes~$p$.

We shall consider the generalization of $n$-neigh\-bor\-hood to \voa{}s in \autoref{sec:neighborhood}, but it seems not to admit such a clear twofold pattern as in \autoref{fig:summary}.


\subsection{Categories}\label{sec:cat}

Recall fusion categories, braided fusion categories, ribbon fusion categories and modular tensor categories, i.e.\ non-degenerate ribbon fusion categories (see, e.g., \cite{BK01,Tur10,EGNO15}). All these are assumed to be semisimple abelian and finite, as is the usual convention. In particular, recall that each metric group $(D,Q)$ defines a pointed modular tensor category denoted by $\mathcal{C}(D)$ \cite{JS93}. This modular tensor category is pseudo-unitary, and we choose on it the unique (spherical) pivotal structure such that all categorical dimensions are positive. Further categorical notions are discussed in \cite[\autoref*{MR1sec:cat}]{MR24a}. For instance, for a braided (or ribbon) fusion category $\mathcal{C}$, we denote by $\overline{\mathcal{C}}$ the braid-reversed (and ribbon-reversed) category of $\mathcal{C}$.

\medskip

We briefly review certain notions of algebras in braided tensor categories and the conditions under which their categories of modules yield new braided tensor categories \cite{Par95,KO02}. Such algebras are closely related to conformal extensions of \voa{}s \cite{HKL15,CKM24,CMSY25}, as we shall discuss in \autoref{sec:voa}, and appear in the context of anyon condensation \cite{FSV13,Kon14}.

In a tensor category $\Cat$, there is the notion of an (associative, unital) algebra $A$ with a multiplication morphism $m\colon A\otimes A\to A$, and we denote by $\Cat_A$ the category of, say, right $A$-modules inside $\Cat$. Further, if $\Cat$ is a braided tensor category and $A$ is a commutative algebra (satisfying $m\circ c_{A,A}=m$), then $\Cat_A$ is a tensor category, and we can moreover consider the category $\Cat_A^\loc$ of local (or dyslectic) modules of $A$ inside $\Cat$, which is again a braided tensor category.

If $\Cat$ is a modular tensor category, one may ask under what conditions on the commutative algebra $A$ the local modules $\Cat_A^\loc$ form again a modular tensor category. This is answered in \cite{KO02,FFRS06,DMNO13} in the semisimple setting, which will be relevant for us, and in \cite{LW23,SY24} in the non-semisimple case.

Given an algebra $A$ in a fusion category $\Cat$, we say that $A$ is \emph{separable} or equivalently (over the base field $\C$) \emph{semisimple} if $\Cat_A$ is a semisimple category (see, e.g., Proposition~2.7 in \cite{DMNO13}, cf.\ Theorem~6.10 in \cite{KZ17}), and $A$ is called \emph{haploid} if $\Hom_\Cat(\textbf{1},A)\cong\C$. Returning to the braided setting, we say that an algebra $A$ is \emph{pre-condensable} if it is separable, haploid and commutative.

\begin{rem}
There are further the notions of \emph{simple} and \emph{exact} algebras (see, e.g., \cite{SY24}), but it turns out that for commutative, haploid algebras in the present setting, i.e.\ in braided fusion categories, they are each equivalent to separability or semisimplicity.
\end{rem}
\begin{proof}
That separability implies exactness is proved, e.g., in Lemma~5.3 of \cite{SY24}. On the other hand, exactness (together with the haploid property, which implies indecomposability) implies simplicity by Theorem~12.1 in \cite{EO21}. Finally, simplicity implies semisimplicity (or equivalently the separability) by Theorem~2.26 in \cite{CMSY25}. This completes the proof. (See also Example~7.5.4 in \cite{EGNO15} for the connection between exactness and semisimplicity.)
\end{proof}

The main result of \cite{DMNO13} (see Section~3.5) states that if $A$ is a pre-con\-dens\-able algebra in a braided fusion category $\Cat$, then $\Cat_A^\loc$ is again a braided fusion category. We call $\Cat_A^\loc$ the \emph{condensation} of $\Cat$ by $A$. Moreover, if $\Cat$ is non-degenerate, then so is $\Cat_A^\loc$.

\begin{rem}\label{rem:separable_exact}
In \cite{SY24}, Theorem~5.21 (see also Theorem~2.23 in \cite{CMSY25}), this statement is generalized to braided finite tensor categories (i.e.\ braided fusion categories without the semisimplicity assumption). The notion of pre-condensable algebras considered in this setting are exact, haploid, commutative algebras.
\end{rem}

We point out that \cite{DMNO13,SY24}, in contrast to \cite{KO02,LW23}, do not assume that the algebra $A$ in $\Cat$ has non-zero categorical dimension $\dim_\Cat(A)\neq0$. This is an important point, as this condition is generally difficult to verify, unless we know, for example, that the category $\Cat$ is pseudo-unitary, in which case the condition is automatically satisfied.

\medskip

Now, let $\Cat$ be a modular tensor category, i.e.\ a non-degenerate ribbon fusion category. We call an algebra $A$ in $\Cat$ \emph{condensable} if it is pre-condensable and has trivial ribbon twist $\theta_A=\id_A$. Then it follows from the results in \cite{DMNO13} that $\Cat_A^\loc$ is a modular tensor category, and in particular inherits a ribbon twist. Once again, we call $\Cat_A^\loc$ the \emph{condensation} of $\Cat$ by $A$.

Again, we compare this result to \cite{KO02,LW23}. They work with what they call rigid $\Cat$-algebras with trivial ribbon twist (or equivalently, rigid Frobenius algebras with trivial ribbon twist), which are the same as condensable algebras with non-zero categorical dimension (see Proposition~3.11 in \cite{LW23}). In other words, the condition $\dim_\Cat(A)\neq0$ may be dropped in the context of condensation.

\begin{rem}
The condensation result is generalized to the non-semisimple setting in \cite{SY24}, Theorem~5.21. There, they use haploid, commutative, exact algebras that are symmetric Frobenius for condensing non-semisimple modular tensor categories. For commutative algebras, symmetric Frobenius is equivalent to Frobenius with trivial ribbon twist (see Proposition~2.25 in \cite{FFRS06}). On the other hand, if we replace exactness by (the in the non-semisimple setting potentially stronger notion of) separability, then by Proposition~3.11 in \cite{LW23}, the Frobenius property follows if the ribbon twist is trivial and $\dim_\Cat(A)\neq0$.
\end{rem}

A condensable algebra $A$ in a modular tensor category is called \emph{Lagrangian} if $\Cat_A^\loc\cong\Vect$, or equivalently if $\operatorname{FPdim}(A)^2=\operatorname{FPdim}(\Cat)$ \cite{DMNO13}. A modular tensor category is said to be \emph{anisotropic} if it admits no non-trivial condensable algebras \cite{DMNO13}.

\medskip

Witt equivalence of metric groups, \autoref{defi:metric_witt}, can be generalized to modular tensor categories. The notion is defined for braided fusion categories in \cite{DMNO13}, and in the non-semisimple setting in \cite{SY24}. Here, we modify the definition to also respect the ribbon structure (see also \cite{Kon14}). Recall from Theorem~1.2 in \cite{Mue03} (and Theorem~2.3 in \cite{ENO05}) that the center of a spherical fusion category is a modular tensor category.
\begin{defi}[Witt Equivalence]\label{defn:Wittequivalencecategories}
Two modular tensor categories $\Cat$ and $\mathcal{D}$ are \emph{Witt equivalent} if there is an equivalence of ribbon fusion categories of the form
\begin{equation*}
\Cat\boxtimes\overline{\mathcal{D}}\cong Z(\mathcal{F}),
\end{equation*}
where $\mathcal{F}$ is some spherical fusion category, and we recall that $\overline{\mathcal{D}}$ denotes the ribbon-reversal of $\mathcal{D}$.
\end{defi}
Witt equivalence defines an equivalence relation and we write $[\mathcal{C}]$ for the Witt equivalence class of $\mathcal{C}$.

We call $[\Vect]$ the trivial Witt class. Its representatives are exactly the modular tensor categories of the form $Z(\mathcal{F})$ for some spherical fusion category.

There are two equivalent formulations of Witt equivalence which will be useful in the sequel (cf.\ Corollary~5.9 and Proposition~5.15 in \cite{DMNO13}, \cite{Kon14}).
\begin{prop}\label{prop:equivalentformulationWitt}
Let $\Cat$ and $\mathcal{D}$ be modular tensor categories. Then the following are equivalent:
\begin{enumerate}
\item $\Cat$ and $\mathcal{D}$ are Witt equivalent.
\item There exist condensable algebras $A$ and $B$ of $\Cat$ and $\mathcal{D}$, respectively, such that the condensations $\mathcal{C}_A^\loc\cong\mathcal{D}_B^\loc$ as ribbon fusion categories.
\item There exists a modular tensor category $\mathcal{E}$ admitting two condensable algebras $A$ and $B$ such that $\mathcal{E}_A^\loc\cong\mathcal{C}$ and $\mathcal{E}_B^\loc\cong\mathcal{D}$ as ribbon fusion categories.
\end{enumerate}
\end{prop}

The following example is immediate:
\begin{ex}[Anyon Condensation]\label{ex:condwitt}
If $\mathcal{C}$ is a modular tensor category and $A$ a condensable algebra, then $\mathcal{C}_A^\loc$ is Witt equivalent to $\mathcal{C}$ as $\mathcal{C}\boxtimes\smash{\overline{\mathcal{C}_A^\loc}}\cong Z(\mathcal{C}_A)$.
\end{ex}

We note that every Witt class contains infinitely many modular tensor categories: for example, as noted above, the modular tensor categories in the trivial Witt class $[\Vect]$ are precisely of the form $Z(\mathcal{F})$ with $\mathcal{F}$ a spherical fusion category, of which there are infinitely many. However, every Witt class admits a unique distinguished representative (cf.\ Theorem~5.13 in \cite{DMNO13}).
\begin{prop}\label{prop:anisorep}
Every Witt class contains an anisotropic representative, which is unique up to ribbon fusion equivalence and denoted by $\mathcal{C}_{\mathrm{an}}$. It can be reached from any representative $\Cat\in[\mathcal{C}_{\mathrm{an}}]$ by taking the condensation $\Cat_A^\loc\cong\Cat_{\mathrm{an}}$ associated with a maximal condensable algebra $A$ in $\mathcal{C}$.
\end{prop}

Returning to the above example, a modular tensor category $\Cat$ is in the trivial Witt class $[\Vect]$ if and only if it possesses a Lagrangian algebra (i.e.\ a condensable algebra $A$ satisfying $\Cat_A^\loc\cong\Vect$) if and only if $\mathcal{C}_{\mathrm{an}}=\Vect$.

\medskip

For pointed modular tensor categories $\mathcal{C}(D)$, characterized by metric groups $D=(D,Q)$, Witt equivalence precisely reproduces Witt equivalence of metric groups. Indeed, the condensable algebras $A=A(H)$ in $\mathcal{C}(D)$ are indexed by the isotropic subgroups $H$ of $D$, and $\Cat_{A(H)}^\loc\cong\mathcal{C}(H^\bot/H)$.\footnote{Strictly speaking, the condensable algebras $A(H)$ in $\mathcal{C}(D)$ also depend on a choice of $2$-cochain $\epsilon\colon H\times H\to\C^\times$, but this does not affect the condensation up to equivalence. A very explicit treatise on this is given, e.g., in \cite{GLM24}, Section~5.2.} Then, $A(H)$ is Lagrangian if and only if $H$ is self-dual in $D$. Hence, the Witt class $[\Cat(D)]=[\Vect]$ is trivial if and only if $D$ contains a self-dual, isotropic subgroup $H$.

We remark that for the modular tensor categories $\mathcal{C}(D)$, the associativity and commutativity of the algebra $A$ already forces the quadratic form $Q$ to be trivial on the subgroup $H$ of $D$, meaning that the condition that the ribbon twist be trivial is redundant here. This changes if we consider the pointed modular tensor categories $\mathcal{C}(D,Q,\theta)$ where the quadratic forms $Q$ and $\theta$ defined by the braiding and the ribbon twist, respectively, do not coincide. One can show that their associated symmetric bilinear forms must coincide, or equivalently that they differ by a homomorphism $D\to\{\pm1\}$ (which here coincides with the categorical dimensions). Now, the associativity and commutativity of the algebra still amount to the vanishing of $Q$ but the trivial-twist condition to the vanishing of $\theta$. (See Section~2.4 in \cite{HM23} and Section~5.2 in \cite{GLM24} for more details.)


\subsection{Vertex Algebras}\label{sec:voa}

We refer readers to, e.g., \cite{Bor86,FLM88,Kac98,FBZ04} and \cite[\autoref*{MR1sec:voa}]{MR24a} for any background on (conformal) vertex algebras, \voa{}s and their representation categories.

\medskip

In this text we shall mainly be concerned with \emph{\strat{}} \voa{}s $V$. These are \voa{}s that are simple, rational, $C_2$-cofinite, self-contragredient and of CFT-type. Their representation categories $\Rep(V)$ are modular tensor categories \cite{Hua08b}. We recall that a \voa{} is of CFT-type if it has $L_0$-grading $V=\bigoplus_{n\geq0}V_n$ and $V_0=\C\vac$.

We further call a \strat{} \voa{} $V$ \emph{positive} if $\rho(M)>0$ for every irreducible module $M\in\Irr(V)$ except for $V$ itself. Here, $\rho(M)$ denotes the smallest $L_0$-eigenvalue of $M$. This implies that the modular tensor category $\Rep(V)$ has only positive categorical dimensions and is hence pseudo-unitary \cite{DJX13,DLN15}.

We recall from \cite{HM23} (based on \cite{Mas14}, see also \cite[\autoref*{MR1sec:asslat}]{MR24a}) that any \strat{} \voa{} $V$ can be written as a simple-current extension (a special case of the conformal extensions discussed below) of a dual pair $C\otimes V_L$. The latter consists of a lattice \voa{} $V_L$ (see below) associated with a positive-definite, even lattice $L$ called the \emph{associated lattice} and of the \emph{Heisenberg commutant} $C$, which is \strat{} and has zero weight-$1$ space, $C_1=\{0\}$. We also recall from \cite[\autoref*{MR1sec:cartan}]{MR24a} the notion of \emph{Cartan subalgebra} (or maximal toral subalgebra) of a \voa{}. In fact, the Heisenberg commutant is obtained as the commutant of this Cartan subalgebra, and the lattice \voa{} $V_L$ as the double commutant.

\medskip

We briefly recall the notions of bulk and hyperbolic genus of \strat{} \voa{}s \cite{Hoe03,Mor21}, which we discussed in detail in \cite{MR24a}.

\begin{defi}[Bulk Genus]\label{defi:bulkgenus}
Two \strat{} \voa{}s $V$ and $V'$ are in the same \emph{bulk genus} if their central charges $c(V)=c(V')$ are the same and their representation categories $\Rep(V)\cong\Rep(V')$ are equivalent as ribbon fusion categories.
\end{defi}

\begin{defi}[Hyperbolic Genus]\label{defi:hypgen}
Two \strat{} \voa{}s $V$ and $V'$ are in the same \emph{hyperbolic genus} if $V\otimes V_{\II_{1,1}}\cong V'\otimes V_{\II_{1,1}}$ as conformal vertex algebras, where $\II_{1,1}$ is the unique even, unimodular lattice of signature $(1,1)$.
\end{defi}

In \cite[\autoref*{MR1thm:hypcomm}]{MR24a} we give an alternative characterization of the \emph{hyperbolic genus} based on the notions of associated lattice and Heisenberg commutant. Based on this, we also show in \cite[\autoref*{MR1cor:hypbulk}]{MR24a} that the hyperbolic genus is a refinement of the bulk genus.

\medskip

Conformal extensions of \voa{}s are intimately related to commutative algebras and their condensations, as discussed in \autoref{sec:cat}. This will be a crucial ingredient for our definition of Witt equivalence of \voa{}s in \autoref{subsec:weakwitt}.

In general, given a \voa{} $V$ with a vertex tensor category $\Cat=\Rep(V)$ of $V$-modules in the sense of \cite{HLZ}, so that $\Cat$ is in particular a braided monoidal category, let $W$ be a conformal (i.e.\ with the same conformal vector) \voa{} extension of $V$. If $W$ is an object of $\Cat$, then $W$ defines the structure of a commutative algebra $A=A(W)$ in $\Cat$, which is haploid if $\Hom_\Cat(V,W)\cong\C$ \cite{HKL15}. In that case, the category of $W$-modules lying in $\Cat$ is a vertex tensor category isomorphic to $\Cat_A^\loc$ as a braided monoidal category \cite{CKM24}.

\medskip

In this text, we are interested in the case of \strat{} \voa{}s $V$. Then one can ask how, in the above situation, the regularity properties of $V$ imply the corresponding properties of $W$. An answer is given in Section~4.2 of \cite{CMSY25}:
\begin{prop}\label{prop:CMSY}
Suppose $V$ is a \strat{} \voa{} (so that $\Cat=\Rep(V)$ is a modular tensor category) and $W\supset V$ is a conformal \voa{} extension (which must be a finite direct sum of $V$-modules by the rationality of $V$ and hence an object of $\Cat$). Then, if $W$ is simple and $\N$-graded, it is already \strat{}. Moreover, $W$ defines a condensable algebra $A$ in $\Cat$ and the modular tensor category $\Rep(W)$ is given by the condensation $\Cat_A^\loc$.
\end{prop}
Because $V$ is not assumed to be positive, we need to demand explicitly that the \voa{} extension $W$ be $\N$-graded (which already implies CFT type here). In general, this cannot be seen from inside the representation category alone because the ribbon twist only depends on the $L_0$-weights of the $V$-modules modulo~$1$ (but see \cite[\autoref*{MR1prop:essentiallypositive}]{MR24a}).

Among other things, the above result strengthens \cite{HKL15} by removing the requirement that the dimension of the algebra $A$ in the modular tensor category $\Rep(V)$ be non-zero. This condition is automatically satisfied if $\Rep(V)$ is pseudo-unitary (for example, when $V$ is positive).

\medskip

Conversely, given a \strat{} \voa{} $V$, it is not difficult to see that any condensable algebra $A$ in $\Cat=\Rep(V)$ defines a simple \voa{} $W$ conformally extending $V$. Now, when $V$ is not positive, the \voa{} extension $W$ does not have to be of CFT-type, and so it may not be \strat{}. Hence, it does not necessarily follow directly from \cite{Hua08b} that $\Rep(W)$ is a modular tensor category.
\begin{prob}
Generalize the modular tensor category result of \cite{Hua08b} to \voa{}s that are \strat{} except that they are not of CFT-type.
\end{prob}
We note that \cite{HKL15} actually provides a way (with a small caveat) to extend the modular tensor category result from \cite{Hua08b} to the \voa{}s $W$ considered in the problem, as long as they contain a \strat{} conformal subalgebra $V$ (which is in practice often the case). We thank the referee for pointing this out. Indeed, it follows from \cite{HKL15} that $\Rep(W)\cong\Cat_A^\loc$ as categories. Then \cite{CKM24} states that $\Rep(W)$ is a braided tensor category in the correct vertex algebraic sense, and this braided tensor structure coincides with that of $\Cat_A^\loc$. Overall, one obtains that $\Rep(W)$ is a modular tensor category. However, as the proof of the modularity is not the one from \cite{Hua08b}, it is not obvious, e.g., if the modular $S$-matrix agrees with the vertex algebraic $S$-matrix (via the characters) up to a scalar.

For completeness, we state the above result in the setting where $V$ is \strat{} and positive (so that its representation category $\Cat=\Rep(V)$ is a pseudo-unitary modular tensor category), in which case it takes its nicest form. Then, any simple conformal extension is automatically of CFT-type. The main theorem of \cite{HKL15}, simplified using the above results, states:
\begin{prop}\label{prop:algext}
Let $V$ be a \strat{} and positive \voa{}. Then there is a bijection between:
\begin{enumerate}
\item conformal \voa{} extensions $W\supset V$, where $W$ is simple (and hence \strat{} and positive),
\item condensable algebras $A$ in $\Cat$.
\end{enumerate}
\end{prop}
We comment that in item~(1), since $V$ is positive, $W$ is $\N$-graded and hence \cite{CMSY25} already implies that $W$ is \strat{}. As $V$ is positive, so is $W$ since any irreducible $W$-module $M$ with $\Hom_\Cat(W,M)\neq\{0\}$ must be isomorphic to~$W$. We further remark that \cite{HKL15} consider in item~(2) rigid haploid commutative algebras $A$ in $\Cat$ with trivial twist, but by Proposition~3.11 in \cite{LW23} this is the same as condensable algebras. Indeed, the condition $\dim(A)\neq0$ is automatically satisfied since $\Cat$ is pseudo-unitary.

\medskip

We consider some examples of \strat{} \voa{}s. In \cite[\autoref*{MR1sec:voa}]{MR24a} we introduced lattice \voa{}s $V_L$ with pointed modular tensor categories $\Rep(V_L)=\Cat(L'/L)$, simple affine \voa{}s $L_\g(k,0)=\mathsf{X}_{n,k}$ with modular tensor categories $\Rep(\mathsf{X}_{n,k})=(X_n,k)$ and parafermion \voa{}s $K(\g,k)$.

The \voa{}s in these examples all satisfy the positivity condition, i.e.\ the smallest $L_0$-weights of all irreducible modules other than the \voa{} itself are positive. This implies that the corresponding modular tensor category is pseudo-unitary \cite{DJX13,DLN15}. We shall now discuss a family of examples of \strat{} \voa{}s that do not necessarily have pseudo-unitary representation categories.
\begin{ex}[Virasoro Minimal Models]\label{ex:minimal}
For any $c\in\C$, there is the universal Virasoro \voa{} $V_\mathrm{Vir}(c,0)$ of central charge $c$. Precisely when $c=c_{p,q}\coloneqq 1-6(p-q)^2/(pq)$ for coprime $p,q\geq 2$, does $V_\mathrm{Vir}(c,0)$ contain a maximal proper ideal and the corresponding simple quotient is the \emph{simple Virasoro \voa{}} (or \emph{Virasoro minimal model}) $L_\mathrm{Vir}(c_{p,q},0)$, which we denote by $L(c_{p,q})$ for short (\cite{FZ92,Wan93}, see also \cite{LL04}).

The \voa{} $L(c_{p,q})$ is \strat{} and the representation category is the modular tensor category $\Rep(L(c_{p,q}))\eqqcolon\mathcal{M}(c_{p,q})$ \cite{Wan93,DLM00}.

In general, $L(c_{p,q})$ may possess irreducible modules with negative $L_0$-weights. Indeed, the minimal $L_0$-weight of an irreducible $L(c_{p,q})$-module is $(1-(p-q)^2)/(4pq)$. This shows that the minimal models with only non-negative $L_0$-weights are exactly those where $q=p+1$ (assuming without loss of generality that $q>p$). In fact, one can verify that these \voa{}s all satisfy the positivity condition so that $\mathcal{M}(c_{p,p+1})$ is pseudo-unitary. The minimal models for $(p,q)=(m+2,m+3)$ with $m\in\N$ are called the \emph{discrete series} and denoted by $L(c_m)$ with $c_m\coloneqq c_{m+2,m+3}=1-6/((m+2)(m+3))$. Their representation category is denoted by $\mathcal{M}_m=\mathcal{M}(c_{m+2,m+3})$.
\end{ex}
One might want to conjecture that the modular tensor categories $\mathcal{M}(c_{p,q})$ are never pseudo-unitary away from the discrete series. The smallest example to look at is the minimal model for $(p,q)=(2,5)$, which is called the Yang--Lee minimal model and has central charge $c_{2,5}=-22/5$. In this case the categorical dimension of the only irreducible module other than the \voa{} itself is $(1-\sqrt{5})/2$, which is negative (see, e.g., \cite{EG17}). Hence, $\mathcal{M}(c_{2,5})$ is not pseudo-unitary.

\medskip

Before moving on, we record some conjectures in the theory of \strat{} \voa{}s that enter in several of our results. Recall that a \strat{} \voa{} $V$ is called \emph{holomorphic} if $\Rep(V)\cong\Vect$.

\begin{conj}[Moonshine Uniqueness \cite{FLM88}]\label{conj:moonshineuniqueness}
The moonshine module (or monster CFT) $V^\natural$ is the only \strat{}, holomorphic \voa{} of central charge $c=24$ with no weight-$1$ states, $V_1=\{0\}$.
\end{conj}

The following is discussed in greater detail in \cite[\autoref*{MR1subsec:reconstruction}]{MR24a}. For simplicity, we assume that the \voa{}s are positive.

\begin{conj}[Weak Reconstruction \cite{Ray23,MR24a}]\label{conj:weakreconstruction}
If there exists a positive, \strat{} \voa{} $V$ with $\Rep(V)\cong\mathcal{C}$, then there exists a positive, \strat{} \voa{} $W$ with $\Rep(W)\cong\overline{\mathcal{C}}$, where $\overline{\mathcal{C}}$ is the ribbon-reverse of $\mathcal{C}$.
\end{conj}

This conjecture is called ``weak reconstruction'' because it follows from the reconstruction conjecture, which asserts that every modular tensor category $\Cat$ is realized as $\Cat\cong\Rep(V)$ for some \strat{} \voa{} $V$. Weak reconstruction admits the following equivalent reformulation.

\begin{conj}\label{conj:weakreconstruction2}
Any positive, \strat{} \voa{} $V$ can be embedded primitively into a \strat{}, holomorphic \voa{}~$W$, in such a way that $\Com_W(V)$ is \strat{}.
\end{conj}
Here, primitive embedding means a conformal embedding $V\subset W$ such that $\Com_W(\Com_W(V))=W$.


\subsection{Physics}\label{subsec:physics}

Finally, we connect the mathematical concepts described above to physics, continuing the discussion in \cite[\autoref*{MR1subsec:physics}]{MR24a}.

\medskip

We recall the physical meaning of Witt equivalence of modular tensor categories, \autoref{defn:Wittequivalencecategories}. Remembering that a TQFT $(\mathcal{C},c)$ admits a topological boundary condition if and only if $c=0$ and $\mathcal{C}$ is of the form $Z(\mathcal{F})$ for some spherical fusion category $\mathcal{F}$, we learn that $\mathcal{C}$ and $\mathcal{D}$ are Witt equivalent if and only if the theory $(\mathcal{C}\boxtimes \overline{\mathcal{D}},0)$ obtained by stacking $(\mathcal{C},c)$ onto the orientation-reversal of $(\mathcal{D},c)$ has a topological boundary. Or, by ``unfolding'' along the boundary, we can reformulate the Witt equivalence of $\mathcal{C}$ and $\mathcal{D}$ as the existence of a topological interface which interpolates between the TQFTs $(\mathcal{C},c)$ and $(\mathcal{D},c)$. This is depicted in \autoref{fig:wittequivalence}.

\begin{figure}[ht]
\begin{center}
\begin{tikzpicture}
\draw[black, thick,dashed](1,-1)--(-5,-1);
\filldraw[black,fill=red!30,fill opacity=.7,  thick,draw=none] (-3,0) -- (1-3,1.5) -- (1-3,-1) -- (-3,-1-1.5) -- cycle;
\draw[black,  thick,dashed] (-6,0)--(0,0);
\draw[black,  thick,dashed] (0,-2.5) -- (-6,-1-1.5);
\draw[black, thick,dashed] (1,1.5)--(1-6,1.5);
\node[] at (-.5,-.5) {$(\mathcal{C},c)$};
\node[] at (-.5-4,-.5) {$(\mathcal{D},c)$};
\node[] at (-2.5,-.5) {$\mathcal{I}$};
\tikzstyle{s}=[draw,decorate,->]
\path[s] (1.2,-.75) to node[anchor=south] {fold} (2.2,-.75);
\draw[black, thick,dashed](1+6-.5,-1)--(-5+6+3-.5,-1);
\filldraw[black,fill=red!30,fill opacity=.7,  thick,draw=none] (-3+6-.5,0) -- (1-3+6-.5,1.5) -- (1-3+6-.5,-1) -- (-3+6-.5,-1-1.5) -- cycle;
\draw[black,  thick,dashed] (-6+6+3-.5,0)--(+6-.5,0);
\draw[black,  thick,dashed] (0+6-.5,-2.5) -- (-6+6+3-.5,-1-1.5);
\draw[black, thick,dashed] (1+6-.5,1.5)--(1-6+6+3-.5,1.5);
\node[] at (-.5+6-.5,-.5) {$(\mathcal{C},c)\boxtimes\overline{(\mathcal{D},c)}$};
\node[] at (-2.5+6-.5,-.5) {$\mathcal{B}$};
\end{tikzpicture}
\end{center}
\caption{A 2+1d depiction of Witt equivalence.}\label{fig:wittequivalence}
\end{figure}
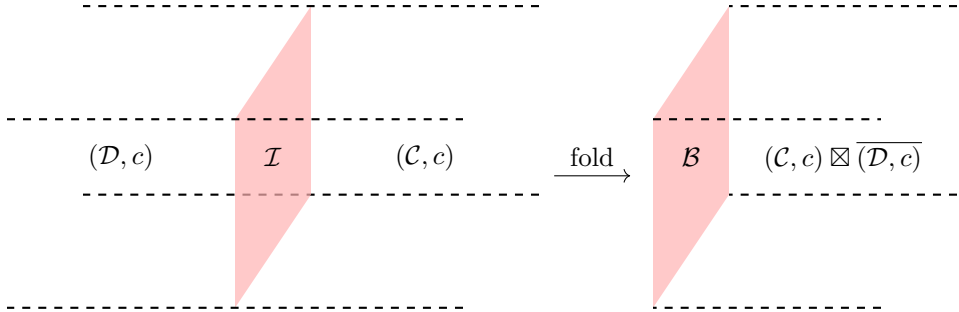

A sufficient but not necessary condition for $\mathcal{D}$ to be Witt equivalent to $\mathcal{C}$ is if $\mathcal{D}$ can be obtained from $\mathcal{C}$ by anyon condensation (or gauging a non-invertible 1-form symmetry) \cite{Kon14}. The data required to perform an anyon condensation is a condensable algebra $A$ of anyons, and we write $(\mathcal{D},c)=(\mathcal{C}_A^\loc,c)=(\mathcal{C},c)/A$ for the TQFT obtained by condensing the anyons of $A$ inside $(\mathcal{C},c)$. The reason this is sufficient to assert Witt equivalence is that one can obtain an interface between $(\mathcal{D},c)=(\mathcal{C},c)/A$ and $(\mathcal{C},c)$ by condensing (or gauging) $A$ in half of spacetime.

Recall that we may think of a rational chiral algebra as a gapless chiral boundary condition of a 3d TQFT. Many statements about chiral algebras become clearer when viewed from this $3$-dimensional perspective. For example, the fact that a condensable algebra $A$ of $\Rep(W)$ defines a conformal extension $W\subset V$ with $\Rep(V)\cong \Rep(W)/A$ can be seen by taking the topological interface that interpolates between $\Rep(W)$ and $\Rep(W)/A$ (obtained by condensing anyons in half of spacetime) and fusing it onto the gapless chiral boundary condition described by $W$ to obtain a new gapless chiral boundary condition of $\Rep(W)/A$. This is depicted in \autoref{fig:condensationextension}.

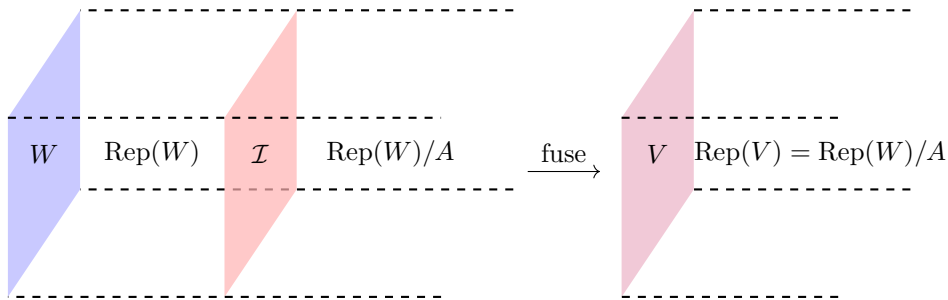
\begin{figure}[ht]
\begin{center}
\begin{tikzpicture}[scale=0.95]
\filldraw[black,fill=blue!30,fill opacity=.7, thick,draw=none] (-6,0) -- (-6,-1-1.5) -- (-5,-1) -- (1-6,1.5) -- cycle;
 \draw[black, thick,dashed](1,-1)--(-5,-1);
\filldraw[black,fill=red!30,fill opacity=.7, thick,draw=none] (-3,0) -- (1-3,1.5) -- (1-3,-1) -- (-3,-1-1.5) -- cycle;
\draw[black, thick,dashed] (-6,0)--(0,0);
\draw[black, thick,dashed] (0,-2.5) -- (-6,-1-1.5);
\draw[black, thick,dashed] (1,1.5)--(1-6,1.5);
\node[] at (-.5-.2,-.5) {$\Rep(W)/A$};
\node[] at (-.5-4+.5,-.5) {$\Rep(W)$};
\node[] at (-2.5,-.5) {$\mathcal{I}$};
\node[] at (-2.5-3,-.5) {$W$};
\tikzstyle{s}=[draw,decorate,->]
\path[s] (1.2,-.75) to node[anchor=south] {fuse} (2.2,-.75);
\draw[black, thick,dashed](1+6-.5,-1)--(-5+6+3-.5,-1);
\filldraw[black,fill=purple!30,fill opacity=.7,  thick,draw=none] (-3+6-.5,0) -- (1-3+6-.5,1.5) -- (1-3+6-.5,-1) -- (-3+6-.5,-1-1.5) -- cycle;
\draw[black, thick,dashed] (-6+6+3-.5,0)--(+6-.5,0);
\draw[black, thick,dashed] (0+6-.5,-2.5) -- (-6+6+3-.5,-1-1.5);
\draw[black, thick,dashed] (1+6-.5,1.5)--(1-6+6+3-.5,1.5);
\node[] at (-.5+6-.5+.2+0.05,-.5) {$\Rep(V)=\Rep(W)/A$};
\node[] at (-2.5+6-.5,-.5) {$V$};
\end{tikzpicture}
\end{center}
\caption{A 2+1d depiction of the relationship between condensable algebras and conformal extensions.}\label{fig:condensationextension}
\end{figure}

Recall from the introduction that we use the notation ${_V}\mathcal{H}_W$ to denote the 2d RCFT obtained by dimensionally reducing the picture in \autoref{fig:KS} along the interval direction. This is an RCFT with $V$ as part of its left-moving chiral algebra, $\overline{W}$ as part of its right-moving chiral algebra, and $\mathcal{H}$ the modular-invariant combination of representations of $V$ and $\overline{W}$ defining the $S^1$ Hilbert space of the theory, thought of as a Lagrangian algebra of $\Rep(V)\boxtimes\overline{\Rep(W)}$. We define the Verlinde lines $\operatorname{Ver}(\mathcal{H}\,/\,V,\!W)$ of ${_V}\mathcal{H}_W$ to be the (fusion) category of topological line operators of ${_V}\mathcal{H}_W$ which commute with the chiral algebras $V$ and $W$ \cite{Ver88}. Mathematically, $\operatorname{Ver}(\mathcal{H}\,/\,V,\!W)=(\Rep(V)\boxtimes\overline{\Rep(W)})_{\mathcal{H}}$, the category of $\mathcal{H}$-modules inside $\Rep(V)\boxtimes\overline{\Rep(W)}$; physically, $\operatorname{Ver}(\mathcal{H}\,/\,V,\!W)$ is simply the category of topological line operators supported on the topological surface defined by $\mathcal{H}$ in \autoref{fig:KS}. See, e.g., \cite{BT18,CLSWY19} for background on fusion category symmetries of 2d CFTs.


\section{Witt Equivalence}\label{sec:witt}

In this section, we extend the notion of Witt equivalence of even lattices, \autoref{defi:lattice_witt}, to the setting of \strat{} \voa{}s. Just as the notion of genus can be generalized in two inequivalent ways \cite{MR24a}, we also find two inequivalent generalizations of Witt equivalence: we refer to them as weak Witt equivalence and strong Witt equivalence, though we often drop the word ``weak'' below. As suggested by the names, strong Witt equivalence defines a strictly finer equivalence relation than weak Witt equivalence. We describe the basic properties of weak Witt equivalence in \autoref{subsec:weakwitt}, and those of strong Witt equivalence in \autoref{subsec:strongwitt}. We discuss familiar examples in \autoref{subsec:wittexamples}.


\subsection{Weak Witt Equivalence}\label{subsec:weakwitt}

We begin by considering the weaker variation of Witt equivalence. To motivate our definition, we rephrase Witt equivalence of even lattices, \autoref{defi:lattice_witt}, in vertex algebraic terms. Recall the definition of Witt equivalence of modular tensor categories, \autoref{defn:Wittequivalencecategories}, and also the definition of lattice vertex algebras or chiral free boson theories (see \cite{Bor86,FLM88}).

\begin{prop}\label{prop:wwittcomp}
Two positive-definite, even lattices $L$ and $M$ are Witt equivalent as lattices if and only if $\Rep(V_L)$ and $\Rep(V_M)$ are Witt equivalent as modular tensor categories and the central charges $c(V_L)=c(V_M)$ coincide.
\end{prop}
\begin{proof}
We only consider the forward direction. The logic of the reverse direction is nearly identical. Because $L$ and $M$ are assumed to be Witt equivalent, $c(V_L)=\rk(L)=\rk(M)=c(V_M)$. On the other hand, by the discussion around equation~\eqref{eqn:lattice_ext}, $L'/L$ and $M'/M$ having isometric subquotients means that $L$ and $M$ can be extended to even lattices $\tilde{L}$ and $\tilde{M}$ that have isometric discriminant forms. Equivalently, $V_L$ and $V_M$ can be conformally extended to lattice \voa{}s $V_{\tilde{L}}$ and $V_{\tilde{M}}$, respectively, which have ribbon equivalent representation categories, $\Rep(V_{\tilde{L}})\cong\Rep(V_{\tilde{M}})$. This means that there are condensable algebras $A$ and $B$ of $\Rep(V_L)$ and $\Rep(V_M)$, respectively, such that $\Rep(V_L)_A^\loc$ is ribbon equivalent to $\Rep(V_M)_B^\loc$. By item~(2) of \autoref{prop:equivalentformulationWitt}, it follows that $\Rep(V_L)$ and $\Rep(V_M)$ are Witt equivalent as modular tensor categories.
\end{proof}
This alternative characterization of Witt equivalence of even lattices admits a clear generalization to \strat{} \voa{}s.

\begin{defi}[Witt Equivalence]\label{defi:weakwittequivalence}
Two \strat{} \voa{}s $V$ and $V'$ are \emph{Witt equivalent} if they have the same central charge, $c(V)=c(V')$, and if $\Rep(V)$ and $\Rep(V')$ are Witt equivalent as modular tensor categories.
\end{defi}
This defines an equivalence relation. We denote the Witt class of a \strat{} \voa{} $V$ by $[V]_{\mathrm{W}}$. Often, we shall use the symbol $([\mathcal{C}],c)$ to represent the Witt class of \voa{}s $V$ with $c(V)=c$ and $\Rep(V)$ Witt equivalent to $\mathcal{C}$, i.e.\ $([\Rep(V)],c(V))=[V]_{\mathrm{W}}$.

\begin{rem}\label{rem:wittlat}
\autoref{prop:wwittcomp} states (in analogy to \cite[\autoref*{MR1prop:bulkcomp}]{MR24a} and \cite[\autoref*{MR1prop:hypercomp}]{MR24a} for the lattice genus) that the map from Witt classes of positive-definite, even lattices to Witt classes of \voa{}s,
\begin{equation*}
[L]\mapsto [V_L]_{\mathrm{W}},
\end{equation*}
is well-defined and injective. That is, our notion of (weak) Witt equivalence actually generalizes Witt equivalence of (positive-definite) even lattices. This is by design.
\end{rem}

\medskip

There is also the following physical characterization of Witt equivalence, which is in the same spirit as \cite[\autoref*{MR1defiph:bulkgenus}]{MR24a} for the bulk genus, and uses the physical interpretation of Witt equivalence of modular tensor categories in terms of 3d TQFTs.

\begin{defiph}[Witt Equivalence]\label{defiph:wittgenus}
Two rational chiral algebras $V$ and $W$ are Witt equivalent if their bulk 3d TQFTs $(\Rep(V),c(V))$ and $(\Rep(W),c(W))$, respectively, can be separated by a (not necessarily invertible) topological interface~$\mathcal{I}$, as depicted in \autoref{fig:WittEqVOA}.
\end{defiph}

\begin{figure}[ht]
\begin{center}
\begin{tikzpicture}[xscale=1.3]
\filldraw[black,fill=blue!30,fill opacity=.7,  thick,draw=none] (-6,0) -- (-6,-1-1.5) -- (-5,-1) -- (1-6,1.5) -- cycle;

\draw[black, thick,dashed](1,-1)--(-5,-1);
\filldraw[black,fill=blue!30,fill opacity=.7,  thick,draw=none] (-6+6,0) -- (-6+6,-1-1.5) -- (-5+6,-1) -- (1-6+6,1.5) -- cycle;
\filldraw[black,fill=red!30,fill opacity=.7,  thick,draw=none] (-3,0) -- (1-3,1.5) -- (1-3,-1) -- (-3,-1-1.5) -- cycle;
\draw[black,thick,dashed] (-6,0) -- (0,0);
\draw[black,thick,dashed] (0,-2.5) -- (-6,-1-1.5);
\draw[black,thick,dashed] (1,1.5) -- (1-6,1.5);
\node[] at (-.5-.4-.1,-.5) {$(\Rep(W),c(W))$};
\node[] at (-.5-4+.5,-.5) {$(\Rep(V),c(V))$};
\node[] at (-2.5,-.5) {$\mathcal{I}$};
\node[] at (-2.5-3,-.5) { $ V$};
\node[] at (-2.5-3+6,-.5) { $\overline{W}$};
\end{tikzpicture}
\end{center}
\caption{Witt equivalence of two rational chiral algebras $V$ and $W$.}
\label{fig:WittEqVOA}
\end{figure}
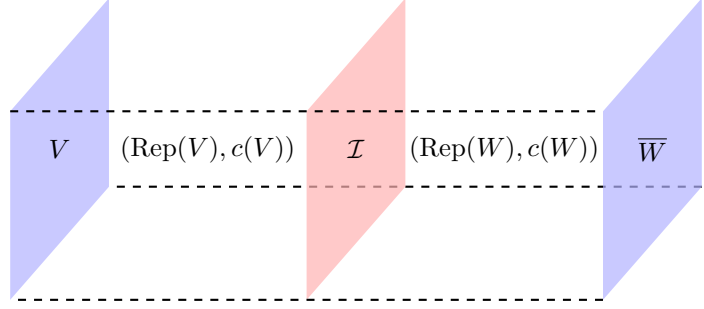

\medskip

We study some properties of Witt equivalence.
\begin{rem}\label{rem:bulkimplieswitt}
Evidently, two \strat{} \voa{}s $V$ and $V'$ in the same bulk genus, \autoref{defi:bulkgenus}, are Witt equivalent.

The converse, however, is not generally true: for example, the simple affine \voa{}s $\mathsf{E}_{8,1}$ and $\mathsf{D}_{8,1}$ are Witt equivalent since they both have central charge $c=8$ and
\begin{equation*}
\Rep(\mathsf{D}_{8,1})\boxtimes \overline{\Rep(\mathsf{E}_{8,1})}\cong Z(\Vect_{\Z_2})\boxtimes \overline{\Vect}\cong Z(\Vect_{\Z_2}),
\end{equation*}
but they are clearly not bulk equivalent.
\end{rem}

The following gives an alternative characterization of Witt equivalence of \strat{} \voa{}s, analogous to item~(2) of \autoref{prop:equivalentformulationWitt} for modular tensor categories.

\begin{prop}\label{prop:alternativecharacterizationWittgenus}
Let $V$ and $V'$ be two \strat{} \voa{}s. If there exist \strat{} \voa{}s $W$ and $W'$ that belong to the same bulk genus and contain $V$ and $V'$, respectively, as conformal subalgebras, then $V$ and $V'$ are Witt equivalent.

If $V$ and $V'$ are additionally positive, then the converse is true as well.
\end{prop}
\begin{proof}
Suppose that we are given $W$ and $W'$ satisfying the hypotheses of the proposition. Evidently, $c(V)=c(V')$. Furthermore, $W$ must be obtained via a condensable algebra $A$ in $\Rep(V)$ for which $\Rep(V)_A^\loc\cong \Rep(W)$, and similarly for $W'$ (see \autoref{prop:CMSY}). By item~(2) of \autoref{prop:equivalentformulationWitt}, $\Rep(V)$ is Witt equivalent to $\Rep(V')$ as a modular tensor category.

Conversely, if $V$ and $V'$ are Witt equivalent (and positive), then by item~(2) of \autoref{prop:equivalentformulationWitt}, there exist condensable algebras $A$ and $A'$ in $\Rep(V)$ and $\Rep(V')$, respectively, such that $\Rep(V)_A^\loc\cong \Rep(V')_{A'}^\loc$ as modular tensor categories. On the other hand, a condensable algebra $A$ in $\Rep(V)$ defines a conformal extension of $V$ to a \strat{} \voa{} $W$ with $\Rep(W)\cong\Rep(V)_A^\loc$; similarly, one obtains a  \strat{} \voa{} $W'$ containing $V'$ as a conformal subalgebra and having $\Rep(W')\cong\Rep(V')_{A'}^\loc$ (see \autoref{prop:algext}). Moreover, $c(W)=c(V)=c(V')=c(W')$.
\end{proof}
\begin{figure}[ht]
\centering
\begin{tikzcd}
[/tikz/execute at end picture={
\node (large) [rectangle, draw, fit=(A3) (A4)] {};
}]
|[alias=A3]| W & |[alias=A4]| W'&\text{bulk equiv.}\\
|[alias=A1]| V \arrow[u,hook,"\text{\,\,\,conf.\ ext.}" right] & |[alias=A2]| V'\arrow[u,hook]&\text{Witt equiv.}
\end{tikzcd}
\caption{Witt equivalence of \strat{} \voa{}s $V$ and $V'$.}
\label{fig:witt}
\end{figure}
We depict the above characterization of Witt equivalence in \autoref{fig:witt} (cf.\ \autoref{fig:wittlat}). In the special case of $W\cong W'=V'$ we obtain (cf.\ \autoref{ex:condwitt}, Lemma~7.12 in \cite{SY24}):
\begin{cor}\label{prop:subalgebraimpliesWitt}
Suppose $V$ and $V'$ are \strat{} \voa{}s. If $V$ is a conformal subalgebra of $V'$, then $V$ and $V'$ are Witt equivalent.
\end{cor}

\begin{rem}
In practice, deciding whether two modular tensor categories are Witt equivalent is a difficult problem. To this end, we note that there are relatively computable invariants of a Witt class of modular tensor categories, known as higher central charges \cite{KKOSS22}, which furnish necessary conditions for Witt equivalence. These can clearly be repurposed in the \voa{} context to decide whether two \voa{}s are Witt equivalent.
\end{rem}

Recall from \autoref{prop:anisorep} that every Witt class $[\mathcal{C}]$ of modular tensor categories admits a unique anisotropic representative $\mathcal{C}_{\mathrm{an}}$. Moreover, every modular tensor category $\mathcal{D}$ that is Witt equivalent to $\mathcal{C}_{\mathrm{an}}$ admits a (maximal) condensable algebra $A$ satisfying $\mathcal{C}_{\mathrm{an}}\cong \mathcal{D}_A^\loc$. The \voa{} translation of this is the following.
\begin{prop}\label{prop:Wittgenusclassification}
Every positive, \strat{} \voa{} $V$ in a Witt class $([\mathcal{C}],c)$ arises as a conformal subalgebra of a \strat{} \voa{} $W$ (which must then also be positive) in the bulk genus $(\mathcal{C}_{\mathrm{an}},c)$.
\end{prop}
Thus, the classification of positive, \strat{} \voa{}s $V$ in a Witt class $([\mathcal{C}],c)$ amounts to the enumeration of positive, \strat{} \voa{}s $W$ in the bulk genus $(\mathcal{C}_{\mathrm{an}},c)$ and all positive, \strat{} conformal subalgebras thereof.

If we drop the assumption that $V$ is positive in the above proposition, it is not clear that the conformal extension $W$ of $V$ corresponding to a maximal condensable algebra $A$ in $\Rep(V)=\mathcal{D}$ is again of CFT-type and thus \strat{}.

\medskip

We can also obtain a characterization of Witt equivalence in terms of 2d RCFTs, analogous to \cite[\autoref*{MR1prop:physicalcharacterizationbulkequivalence}]{MR24a}.

\begin{propph}\label{prop:nondiagonalRCFT}
Let $V$ and $V'$ be two positive, rational chiral algebras. Then $V$ and $V'$ are Witt equivalent if and only if there exists a (not necessarily diagonal) 2d RCFT with vanishing gravitational anomaly, with $V$ as a conformal subalgebra of its maximal left-moving chiral algebra, and with $\overline{V'}$ as a conformal subalgebra of its maximal right-moving chiral algebra.
\end{propph}
\begin{proof}
Assuming $V$ and $V'$ are Witt equivalent, one can use \autoref{prop:alternativecharacterizationWittgenus} to obtain  \strat{} \voa{}s $W$ and $W'$ with $c(W)=c(W')$ and $\Rep(W)\cong\Rep(W')$, and then use \cite[\autoref*{MR1prop:physicalcharacterizationbulkequivalence}]{MR24a} to obtain a 2d RCFT that satisfies the hypotheses of the proposition. Alternatively, one can simply form the sandwich in \autoref{fig:WittEqVOA} and dimensionally reduce the configuration to obtain a 2d RCFT that also satisfies the hypotheses of the proposition.

Going in the other direction, if such a 2d RCFT exists, then $c(V)=c(V')$, and its Hilbert space defines a Lagrangian algebra of $\Rep(V)\boxtimes\smash{\overline{\Rep(V')}}$, which in turn defines a (not necessarily invertible) topological domain wall between $\Rep(V)$ and $\Rep(V')$, so that $\Rep(V)$ is Witt equivalent to $\Rep(V')$.
\end{proof}

The following conjecture is a corollary of \cite[\autoref*{MR1conj:bulkfiniteness}]{MR24a} (the bulk genera finiteness) and the rank finiteness of modular tensor categories \cite{BNRW16}.
\begin{conj}\label{conj:countableinfinity}
Every Witt class of \strat{} \voa{}s is at most countably infinite.
\end{conj}
Of course, this is also inspired by the analogous fact that there are countably many lattices in each Witt class of even lattices.

\begin{ex}\label{ex:countableinfinity}
Witt classes containing infinitely many \voa{}s are quite typical. For example, let $V$ be a \strat{} \voa{} with a continuous automorphism group (this is the case for non-zero weight-$1$ space $V_1\neq\{0\}$). Pick a subgroup of $\Aut(V)$ isomorphic to $\mathrm{U}(1)$ and in it, for each $n\in\Ns$, the subgroup isomorphic to $\Z_n$. Then form the \fpvosa{}s $\smash{\{V^{\Z_n}\}_{n\in\Z_{>0}}}$. Each $V^{\Z_n}$ is \strat{} \cite{Miy15,CM16} and Witt equivalent to $V$ by \autoref{prop:subalgebraimpliesWitt}. But $V^{\Z_n}$ is not isomorphic to $V^{\Z_m}$ if $n\neq m$ because $\dim(\Rep(V^{\Z_n}))=n^2\dim(\Rep(V))\neq m^2\dim(\Rep(V))= \dim(\Rep(V^{\Z_m}))$. Thus, the $V^{\Z_n}$ furnish infinitely many, mutually non-isomorphic members of $[V]_{\mathrm{W}}$.

Each $V^{\Z_n}$ also belongs to a different bulk genus; so we have also shown that Witt classes usually decompose into infinitely many bulk genera.
\end{ex}

Finally, we extend the definition of Witt equivalence to full RCFTs.
\begin{defiph}[Witt Equivalence]\label{defi:WittRCFT}
Two RCFTs ${_V}\mathcal{H}_W$ and ${_{V'}}\mathcal{H}'_{W'}$ are \emph{Witt equivalent} if $V$ is Witt equivalent to $V'$, and $W$ is Witt equivalent to $W'$.
\end{defiph}

\begin{remph}\label{rem:invariantsofwitt}
By \autoref{prop:nondiagonalRCFT}, one only needs to check that $V$ is Witt equivalent to $V'$ and that $c(W)=c(W')$ to decide whether ${_V}\mathcal{H}_W$ and ${_{V'}}\mathcal{H}'_{W'}$ are Witt equivalent as RCFTs. In particular, a Witt class of RCFTs can be labeled by a triple $([\mathcal{C}],c_L,c_R)$ consisting of the left- and right-moving central charges $c_L$ and $c_R$, and the Witt class $[\mathcal{C}]$ of the representation category of the left-moving (or right-moving) chiral algebra of any representative RCFT of the class.
\end{remph}


\subsection{Strong Witt Equivalence}\label{subsec:strongwitt}

We describe a second definition of Witt equivalence of \voa{}s. Our definition of strong Witt equivalence can be compared to our alternative characterization of the hyperbolic genus in \cite[\autoref*{MR1defi:commgenus}]{MR24a}. Heuristically speaking, in that context, one applies the lattice-theoretic notion of genus at the level of the associated lattice (see \autoref{sec:voa}) and asks that the Heisenberg commutant be preserved. We can apply the same reasoning to produce a strong version of Witt equivalence of \voa{}s.

\begin{defi}[Strong Witt Equivalence]\label{defi:strongwittequivalence}
Two \strat{} \voa{}s are \emph{strongly Witt equivalent} if their Heisenberg commutants are isomorphic and their associated lattices are Witt equivalent as lattices.
\end{defi}
We can immediately justify our use of the word ``strong''.

\begin{prop}\label{prop:weakstrongwitt}
If two \strat{} \voa{}s are strongly Witt equivalent, then they are Witt equivalent.
\end{prop}
\begin{proof}
Consider two \strat{} \voa{}s $V_1$ and $V_2$ that are strongly Witt equivalent. Let $C_i$ and $L_i$ be the Heisenberg commutant and associated lattice, respectively, of $V_i$. By assumption, $C_1$ is isomorphic to $C_2$, and $L_1$ is Witt equivalent to $L_2$, and hence $c(V_1)=c(C_1)+\rk(L_1)=c(C_2)+\rk(L_2)=c(V_2)$. Further, $\Rep(C_1)$ is ribbon equivalent to $\Rep(C_2)$, and $\Rep(V_{L_1})$ and $\Rep(V_{L_2})$ are Witt equivalent as modular tensor categories.

By Theorem~5.5 in \cite{CKM22} (see also Lemma~2.16 in \cite{Lin17}), $\Rep(C_i\otimes V_{L_i})$ is the Deligne product $\Rep(C_i)\boxtimes\Rep(V_{L_i})$. It follows that $\Rep(C_1\otimes V_{L_1})$ is Witt equivalent to $\Rep(C_2\otimes V_{L_2})$ as a modular tensor category, and further that  $C_1\otimes V_{L_1}$ is Witt equivalent to $C_2\otimes V_{L_2}$ as a \voa{}. On the other hand, $V_i$ is a conformal extension of $C_i\otimes V_{L_i}$, so \autoref{prop:subalgebraimpliesWitt} can be used to deduce that $C_i\otimes V_{L_i}$ is Witt equivalent to $V_i$. It follows by transitivity that $V_1$ is Witt equivalent to $V_2$.
\end{proof}

We can make \autoref{prop:weakstrongwitt} more precise as follows.
\begin{prop}\label{prop:weakstrongwitt2}
Two \strat{} \voa{}s are strongly Witt equivalent if and only if they are Witt equivalent and their Heisenberg commutants are isomorphic.
\end{prop}
\begin{proof}
The forward direction follows from \autoref{prop:weakstrongwitt}.

In order to prove the reverse direction, suppose $V$ and $V'$ are Witt equivalent and both have the Heisenberg commutant $C$. Denote by $L$ and $M$ the associated lattices of $V$ and $V'$, respectively. Because $V$ and $V'$ are assumed to be Witt equivalent, $\rk(L)=c(V)-c(C)=c(V')-c(C)=\rk(M)$. We proceed to show that $\mathcal{C}(L'/L)$ and $\mathcal{C}(M'/M)$ are Witt equivalent as modular tensor categories, from which it follows that $L$ and $M$ are Witt equivalent as lattices, and further that $V$ and $V'$ are strongly Witt equivalent as \voa{}s.

There are the dual pairs
\begin{equation*}
C\otimes V_L\subset V\quad\text{and}\quad C\otimes V_M\subset V'
\end{equation*}
and, like in the proof of \cite[\autoref*{MR1cor:unpointed}]{MR24a}, we obtain
\begin{align*}
\Rep(V_L)&\cong(\Rep(V)\boxtimes\overline{\Rep(C)})^\loc_{A},\\
\Rep(V_M)&\cong(\Rep(V')\boxtimes\overline{\Rep(C)})^\loc_{B}
\end{align*}
for some condensable algebras $A$ and $B$ \cite{FFRS06}. Here, we recall that $\overline{\Rep(C)}$ denotes the ribbon-reverse of the modular tensor category $\Rep(C)$, as defined, e.g., in \cite[\autoref*{MR1sec:cat}]{MR24a}.

Now, since $\Rep(V)$ and $\Rep(V')$ are Witt equivalent modular tensor categories by assumption, so are $\Rep(V)\boxtimes\smash{\overline{\Rep(C)}}$ and $\Rep(V')\boxtimes\smash{\overline{\Rep(C)}}$. The same can be said about their condensations by \autoref{ex:condwitt}. So, $\Rep(V_L)\cong\mathcal{C}(L'/L)$ and $\Rep(V_M)\cong\mathcal{C}(M'/M)$ are Witt equivalent, which by \autoref{prop:wwittcomp} means that the associated lattices $L$ and $M$ are Witt equivalent. Hence, $V$ and $V'$ are in the same strong Witt class.
\end{proof}

Recall that we denoted the Witt class of a \strat{} \voa{} $V$ by $([\mathcal{C}],c)$ where $c$ is the central charge of the \voa{}s in the class and $[\mathcal{C}]$ is the Witt class of the representation category of any representative.

It now follows from \autoref{prop:weakstrongwitt2} that we can label a strong Witt class of \voa{}s by a tuple $(C,[\mathcal{C}],c)$, where additionally $C$ is the Heisenberg commutant of any representative (up to isomorphism).

\begin{rem}\label{rem:strongwittlat}
\autoref{prop:weakstrongwitt2} implies that the map from Witt classes of positive-definite, even lattices to strong Witt classes of \voa{}s defined by mapping the Witt class $[L]$ of a lattice $L$ to the strong Witt class of the lattice \voa{} $V_L$ is well-defined and injective. That is, the notion of strong Witt equivalence (like weak Witt equivalence, see \autoref{rem:wittlat}) generalizes Witt equivalence of (positive-definite) even lattices, as is intended.
\end{rem}
Together with \autoref{rem:wittlat} this also means, not surprisingly, that for lattice \voa{}s, weak and strong Witt equivalence coincide.

\medskip

We can also immediately see that strong Witt equivalence is coarser than hyperbolic equivalence (recalled in \autoref{defi:hypgen}):
\begin{prop}\label{prop:hypstrongwitt}
If two \strat{} \voa{}s are in the same hyperbolic genus, then they are strongly Witt equivalent.
\end{prop}
\begin{proof}
Consider two \strat{} \voa{}s $V_1$ and $V_2$ in the same hyperbolic genus. By the alternative characterization in \cite[\autoref*{MR1defi:commgenus}]{MR24a}, it follows that the Heisenberg commutants of $V_1$ and $V_2$ are isomorphic. Moreover, their associated lattices reside in the same genus, which in particular implies that their associated lattices are Witt equivalent. It follows that $V_1$ and $V_2$ are strongly Witt equivalent.
\end{proof}

\begin{rem}
There are no implication arrows between bulk equivalence and strong Witt equivalence.

For example, the moonshine module $V^\natural$ and the Leech lattice \voa{} $V_{\Lambda}$ are in the same bulk genus because they both have $c(V^\natural)=c(V_{\Lambda})=24$ and $\Rep(V^\natural)\cong \Rep(V_{\Lambda})\cong \Vect$. However, they are not strongly Witt equivalent because their Heisenberg commutants are non-isomorphic: the Heisenberg commutant of $V_\Lambda$ is trivial but the Heisenberg commutant of $V^\natural$ is $V^\natural$ itself.

On the other hand, the two affine \voa{}s $\mathsf{D}_{8,1}$ and $\mathsf{E}_{8,1}$ are strongly Witt equivalent because they both have trivial Heisenberg commutant and their associated lattices are Witt equivalent (indeed, the $D_8$ lattice embeds into the $E_8$ lattice). However, they are straightforwardly seen to have different representation categories because, e.g., $\rk(\Rep(\mathsf{E}_{8,1}))=1\neq 4 =\rk(\Rep(\mathsf{D}_{8,1}))$.
\end{rem}

We conclude this section by extending strong Witt equivalence to full RCFTs in the obvious way (cf.\ \autoref{defi:WittRCFT}).
\begin{defiph}[Strong Witt Equivalence]
Two RCFTs ${_V}\mathcal{H}_W$ and ${_{V'}}\mathcal{H}'_{W'}$ are strongly Witt equivalent if $V$ is strongly Witt equivalent to $V'$ and $W$ is strongly Witt equivalent to $W'$.
\end{defiph}


\subsection{Examples}\label{subsec:wittexamples}

We now turn to examples. Classifying isomorphism classes of \voa{}s in a given Witt class or strong Witt class is of course generally an even more difficult problem than classifying \voa{}s in a given bulk or hyperbolic genus. For example, as we saw in \autoref{ex:countableinfinity}, a Witt class of \voa{}s generally splits into infinitely many bulk genera, and similarly a strong Witt class may split into infinitely many hyperbolic genera. However, there are some familiar exceptions for central charge at most~$1$, where one can say quite a bit.

\medskip

First, we briefly discuss the case of \voa{}s with Witt trivial representation categories.

\begin{ex}[Witt Classes {$([\Vect],c)$}]\label{ex:trivialwitt}
By \autoref{prop:Wittgenusclassification}, each positive, \strat{} \voa{} in a Witt class $([\Vect],c)$ can be conformally extended to a holomorphic \voa{}. In particular, the central charge must be $c\in 8\N$ \cite{Zhu96}. In fact, if $c\in 8\N$, then the positive \voa{}s in the Witt class $([\Vect],c)$ consist exactly of all \strat{}, holomorphic \voa{}s in the bulk genus $(\Vect,c)$ and all their positive, \strat{} conformal subalgebras.
\end{ex}
It would be interesting to know if there exist any \strat{} \voa{}s that are not positive but condensable to $\Vect$. We can (almost) give a negative answer in some cases \cite{Moe18}. Recall that a large subclass of the modular tensor categories in the trivial Witt class $[\Vect]$ are the twisted Drinfeld doubles $\mathcal{D}_\omega(G)=Z(\Vect_G^\omega)$ for finite groups $G$ and $3$-cocycles $\omega\in H^3(G,\C^\times)$. Their Lagrangian subalgebras were described in \cite{DS17}.
\begin{prop}\label{prop:vectpos}
Let $V$ be a \strat{} \voa{} with modular tensor category $\Rep(V)\cong\mathcal{D}_\omega(G)$. If $V$ admits a conformal extension to a \strat{}, holomorphic \voa{} $W$, i.e.\ with $\Rep(W)\cong\Vect$, then all $V$-modules are non-negatively graded.
\end{prop}
\begin{proof}
By assumption, there is a subgroup $G<\Aut(W)$, corresponding to the Lagrangian algebra describing the extension $V\subset W$, such that $V$ is the \fpvosa{} $V=W^G$ (cf.\ \cite{Kir02,DNR25}). Moreover, all the irreducible $V$-modules appear as submodules of the $g$-twisted $W$-modules for $g\in G$ \cite{DRX17}. Finally, in \cite{Moe18} we show that all $g$-twisted $W$-modules are non-negatively graded.
\end{proof}

By \autoref{prop:weakstrongwitt2}, in order to describe the decomposition of $([\Vect],c)$ into strong Witt classes, we need to simply add the Heisenberg commutant as a further invariant. It would be interesting to study this in detail for $c=24$. As a first step in this direction, in \cite{HM23} we describe the possible Heisenberg commutants of index-$2$ conformal subalgebras of the \strat{}, holomorphic, $c=24$ \voa{}s. In fact, we only do this for roughly half of them, namely the ones that correspond to \emph{non-anomalous} orbifolds of order~$2$. For these, modulo \autoref{conj:moonshineuniqueness}, there are exactly $18$ Heisenberg commutants up to isomorphism.

\medskip

We now come to two examples that deal with \strat{} \voa{}s of central charge at most~$1$.

\begin{ex}[Virasoro Minimal Models]\label{ex:discretewitt}
We revisit the minimal models in the discrete series from \autoref{ex:minimal}, i.e.\ the positive, \strat{} \voa{}s $L(c_m)$ of central charge $c_{m}=1-6/((m+2)(m+3))<1$ and with (pseudo-unitary) representation category $\mathcal{M}_m=\Rep(L(c_m))$ for $m\in\N$. In the following, we describe all the \strat{}, \emph{positive} \voa{}s that reside in the Witt classes $([\mathcal{M}_m],c_m)$, as well as how those \voa{}s split into strong Witt classes.

The strong Witt classes turn out to be trivial. Indeed, we note that any positive \voa{} $V$ with central charge $c<1$ necessarily has trivial weight-$1$ space, $V_1=\{0\}$, by the results of \cite{DM04b}. Hence, such a $V$ is equal to its own Heisenberg commutant, and so by \autoref{prop:weakstrongwitt2}, checking strong Witt equivalence of positive \voa{}s with $c<1$ is the same as checking if they are isomorphic.

Next, we note that the \vosa{} generated by the conformal vector (or stress tensor) of a positive, \strat{} \voa{} is simple \cite{DZ08}. In particular, it follows that if $V$ is a positive, \strat{} \voa{} of central charge $c=c_m$ for some $m\in\N$, then it is a conformal extension of $L(c_m)$, and in particular is Witt equivalent to $L(c_m)$ by \autoref{prop:subalgebraimpliesWitt}. Thus, the problem of computing the positive \voa{}s in the Witt classes $([\mathcal{M}_m],c_m)$ is the same as the problem of determining the conformal extensions of $L(c_m)$. These are classified in \cite{DL15}.

For completeness, we write out the (positive members of the) Witt classes in detail. We recall from \cite{Wan93} that the irreducible modules of $L(c_m)$ are given by $L(c_m,h^{(m)}_{r,s})$, where
\begin{equation*}
h^{(m)}_{r,s} = \frac{(r(m+3)-s(m+2))^2-1}{4(m+2)(m+3)}
\end{equation*}
for integers $1\leq r\leq m+1$ and $1\leq s\leq m+2$. Here, $h^{(m)}_{r,s}$ denotes the smallest $L_0$-weight of the irreducible module. In this parametrization, each irreducible module of $L(c_m)$ appears exactly twice since $\smash{L(c_m,h^{(m)}_{r,s})\cong L(c_m,h^{(m)}_{m+2-r,m+3-s})}$.

Furthermore, the special module $L(c_m,h^{(m)}_{m+1,1})$ is always a simple-current module of $L(c_m)$, and $\smash{h^{(m)}_{m+1,1}}=m(m+1)/4$ is an integer if and only if $m\equiv 0,3\pmod 4$. Then, by \autoref{prop:algext}, there is a unique \voa{} structure on
\begin{equation*}
W(c_m) \coloneqq L(c_m)\oplus L(c_m,h^{(m)}_{m+1,1})\quad\text{for }m\equiv 0,3\pmod 4.
\end{equation*}
There are also the four exceptional extensions
{\allowdisplaybreaks
\begin{align*}
E(c_9) &\coloneqq L(c_9)\oplus L(c_9,h^{(9)}_{1,7}),\\
E(c_{10}) &\coloneqq L(c_{10})\oplus L(c_{10},h^{(10)}_{7,1}),\\
E(c_{27}) &\coloneqq L(c_{27})\oplus L(c_{27},h^{(27)}_{1,11})\oplus L(c_{27},h^{(27)}_{1,19})\oplus L(c_{27},h^{(27)}_{1,29}),\\
E(c_{28})&\coloneqq L(c_{28})\oplus L(c_{28},h^{(28)}_{11,1}) \oplus L(c_{28},h^{(28)}_{19,1}) \oplus L(c_{28},h^{(28)}_{29,1}).
\end{align*}
}%
In terms of these ingredients, the subclass $([\mathcal{M}_m],c_m)_+$ of positive \voa{}s in the Witt class $([\mathcal{M}_m],c_m)$ is, up to isomorphism,
\begin{equation*}
([\mathcal{M}_m],c_m)_+ =
\begin{cases}
\{L(c_m)=\C\vac\},       &m=0,\\
\{L(c_m)\},              &m\equiv1,2\pmod4\text{ and }m\neq9,10,\\
\{L(c_m),W(c_m)\},       &m\equiv0,3\pmod4\text{ and }m\neq0,27,28,\\
\{ L(c_m),E(c_m)\},      &m=9,10,\\
\{L(c_m),W(c_m),E(c_m)\},&m=27,28
\end{cases}
\end{equation*}
for all $m\in\N$.
\end{ex}

\begin{rem}
We have seen that the Witt class $([\mathcal{M}_m],c_m)$ contains (at least) all positive, \strat{} \voa{}s of central charge $c_m$. On the other hand, there are \strat{} \voa{}s of central charge $c_m$ that are not positive. For instance, $V\coloneqq L(c_{3,5})^{\otimes2}\otimes V_L$ for any positive-definite, even lattice $L$ of rank~$2$ has central charge $c=2c_{3,5}+2=2(-3/5)+2=4/5=c_3$.

Nonetheless, we do not believe that such examples interfere with the Witt classes of the discrete series. That is, we conjecture: \emph{all \voa{}s in the Witt classes $([\mathcal{M}_m],c_m)$ for $m\in\N$ are positive} and hence given by the \voa{}s in \autoref{ex:discretewitt} (cf.\ the discussion preceding \autoref{prop:vectpos}). This conjecture also implies the conjecture: \emph{the bulk genus $(\mathcal{M}_m,c_m)$ of $L(c_m)$ contains only $L(c_m)$}. The latter conjecture boils down to believing that the pseudo-unitary modular tensor category $\mathcal{M}_m=\Rep(L(c_m))$ cannot be realized by a \strat{} \voa{} not satisfying the positivity condition. The absence of positivity-violating \voa{}s $V$ with $\Rep(V)=\mathcal{M}_m$ can be proven for many $m$ using \cite[\autoref*{MR1prop:essentiallypositive}]{MR24a}, e.g., for $m=1,2$, though we are not able to produce a proof that works for all $m$ uniformly.
\end{rem}

The previous example was concerned with the Witt classes of \strat{} \voa{}s of central charge $c<1$. Another case where a complete classification of all \strat{} \voa{}s is conceivable is that of central charge $c=1$, at least once again under the assumption of positivity. For larger central charges, a classification of all \strat{} \voa{}s becomes quite intractable, and we have to content ourselves, at least for the time being, with studying the \voa{}s within a given bulk genus, as we did in \cite[\autoref*{MR1subsec:bulkex}]{MR24a}.

\begin{ex}[\VOA{}s with $c=1$]\label{ex:cc1}
We recall the following folklore conjecture (see, e.g., \cite{Gin88,Kir89,DVVV89}): \emph{every positive, \strat{} \voa{} of central charge $c=1$ is a rank-$1$ lattice \voa{} or a \fpvosa{} (under a finite group of automorphisms) thereof}. (See \cite{Xu05} for a result in the language of conformal nets under a certain ``spectrum condition''.)

This conjecture is false if we drop the positivity condition. For example, take the \strat{} \voa{} $V\coloneqq L(c_{3,5})^{\otimes5}\otimes V_L$ for any positive-def\-i\-nite, even lattice $L$ of rank~$4$. It has central charge $c=5c_{3,5}+4=5(-3/5)+4=1$ but is certainly not isomorphic to a \voa{} of the conjectured form because its associated lattice $L$ has rank greater than $1$.

\smallskip

As the positive-definite, even lattices of rank~$1$ are easily enumerated and the automorphism groups of lattice \voa{}s are well-understood \cite{FLM88,Bor92,DN99,HM22}, we can derive a complete list of all \strat{} \voa{}s appearing in the conjecture. We shall then organize them into Witt classes.

All positive-definite, even lattices of rank~$1$ are isometric to $\sqrt{m}A_1=\sqrt{2m}\Z$ for $m\in\Ns$, i.e.\ they are scaled versions of the root lattice $A_1=\sqrt{2}\Z$. The isometry group is $\Aut(\sqrt{m}A_1)=\langle-\id\rangle$ in each case. It will be convenient to distinguish the cases $m=1$ and $m>1$ in the following.

For $m>1$, the \voa{} $\smash{V_{\sqrt{m}A_1}}$ has a $1$-dimensional, abelian weight-$1$ Lie algebra (isometric to $\sqrt{m}A_1\otimes_\Z\C$). This implies that any \fpvosa{} of $\smash{V_{\sqrt{m}A_1}}$ under a finite-order group of automorphisms is either again a lattice \voa{} $\smash{V_{\sqrt{k}A_1}}$ for some finite-index sublattice of $\sqrt{m}A_1$, which is again of the form $\sqrt{k}A_1$ for some $k\geq m$ with $k/m$ a perfect square, or the \fpvosa{} $\smash{V_{\sqrt{k}A_1}^+}$ under (some unimportant choice of lift of) the $(-\id)$-involution on $\sqrt{k}A_1$. Both $\smash{V_{\sqrt{k}A_1}}$ and $\smash{V_{\sqrt{k}A_1}^+}$ are \strat{} and positive (see below).

We come to the case $m=1$. In that case, the lattice \voa{} $V_{A_1}\cong\mathsf{A}_{1,1}$ is also the simple affine \voa{} for $\sl_2$ at level~$1$. After all, $A_1$ is the root lattice of $\sl_2$. Any lift of the $(-\id)$-involution on the lattice $A_1$ is an inner automorphism of the \voa{} $V_{A_1}$. Indeed, $V_{A_1}$ is strongly generated by the standard generators $\{e,f,h\}$ of the weight\nobreakdash-$1$ Lie algebra $\sl_2$, and the outer automorphism group of $\sl_2$ is trivial. Hence, the automorphism group of $V_{A_1}$ is given by the inner automorphism group of $\sl_2$, which is $\mathrm{PSL}(2,\C)$ acting by conjugation.

Now, the \emph{finite} subgroups of $\mathrm{PSL}(2,\C)=\mathrm{SL}(2,\C)/\{\pm\id\}$ or equivalently of $\mathrm{SU}(2)/\{\pm\id\}\cong\mathrm{SO}(3)$, up to conjugation, famously correspond bijectively to the Dynkin diagrams of ADE-type (via the Kleinian or Du Val singularities). The \fpvosa{}s of $V_{A_1}$ under subgroups of $\Aut(V_{A_1})$ corresponding to Dynkin types $A$ (cyclic subgroups) and $D$ (dihedral groups) are of the form $\smash{V_{\sqrt{k}A_1}}$ and $\smash{V_{\sqrt{k}A_1}^+}$, respectively, for $k$ a perfect square, as already discussed for $m>1$. The remaining three subgroups $G(E_n)<\Aut(V_{A_1})$ of type $E_6$, $E_7$ and $E_8$ are exceptional cases called the tetrahedral, octahedral and icosahedral group, abstractly isomorphic to $A_4$, $S_4$ and $A_5$, respectively. We denote the corresponding \fpvosa{}s by $V(E_n)\coloneqq\smash{V_{A_1}^{G(E_n)}}$. Except for $G(E_8)\cong A_5$, all these groups are solvable so that we know that the \fpvosa{}s are \strat{} by \cite{Miy15,CM16}, with the possible exception of $V(E_8)$. Of course, one is inclined to believe that also $V(E_8)$ is \strat{}. Moreover, these \fpvosa{}s are all positive (provided they are \strat{}), as their irreducible modules are contained in twisted modules for $V_{A_1}$ \cite{DRX17}, which have positive weight grading \cite{DL96}.

\smallskip

Summarizing the above considerations, we arrive at the following, conjecturally complete, list of positive, \strat{}, $c=1$ \voa{}s:
\begin{enumerate}
\item rank-$1$ lattice \voa{}s, i.e.\ $V_{\sqrt{m}A_1}$ for $m\in\Ns$,
\item \fpvosa{}s $\smash{V_{\sqrt{m}A_1}^+}$, $m\in\Ns$, of the above under (a lift of) the $(-\id)$-involution (charge conjugation orbifold),
\item one of three exceptional \fpvosa{}s $V(E_6)$, $V(E_7)$ or $V(E_8)$ of $V_{A_1}$.
\end{enumerate}

\smallskip

Still assuming the above conjecture, one can then ask how these \voa{}s are organized into Witt classes. To this end, we first note that any positive integer $m$ can be uniquely factorized as $m=k\ell^2$, where $k$ is squarefree. It follows from the conformal extension $\smash{V_{\sqrt{k}\ell A_1}\subset V_{\sqrt{k}A_1}}$ already discussed above that the \voa{}s $V_{\sqrt{k}\ell A_1}$ all belong to the same Witt class if one fixes squarefree $k\in\Ns$ and varies over $\ell\in\Ns$.

On the other hand, we can show that the lattice \voa{}s $V_{\sqrt{k}A_1}$ all belong to different Witt classes as one varies over squarefree integers $k\in\Ns$. To see this, recall that two inequivalent modular tensor categories belong to different Witt classes if they are both anisotropic since there is only one anisotropic modular tensor category in every Witt class, up to ribbon equivalence (see \autoref{prop:anisorep}). Thus, since the categories $\smash{\Rep(V_{\sqrt{k}A_1})}\cong\mathcal{C}((\sqrt{k}A_1)'/(\sqrt{k}A_1))$ are ribbon inequivalent for different values of squarefree $k$ (they have different ranks, for example), it will follow that the \voa{}s $\smash{V_{\sqrt{k}A_1}}$ belong to different Witt classes if we prove that each category $\Rep(V_{\sqrt{k}A_1})$ is anisotropic.

To demonstrate anisotropy, it suffices to argue that $\smash{V_{\sqrt{k}A_1}}$ has no irreducible module with integral $L_0$-weights (besides the vacuum module). Indeed, we observe that the smallest $L_0$-weights of the irreducible modules of $V_{\sqrt{k}A_1}$ are $h_\lambda\equiv\lambda^2/4k\pmod{1}$ for $\lambda=0,\dots,2k-1$, and therefore $\smash{\Rep(V_{\sqrt{k}A_1})}$ is anisotropic provided there is no $\lambda\in\{1,\dots,2k-1\}$ such that $\lambda^2\equiv0\pmod{4k}$. The absence of such a $\lambda$ follows from the fact that $k$ is squarefree.

The non-lattice \voa{}s from above all conformally embed into at least one lattice \voa{} by construction, and so they naturally fall into the Witt classes we have already discussed. In particular, it follows straightforwardly from our discussion that the Witt classes at $c=1$ that contain at least one positive, \strat{} \voa{} are
\begin{equation*}
\bigl(\bigl[\mathcal{C}((\sqrt{k}A_1)'/(\sqrt{k}A_1))\bigr],1\bigr),\quad k\in\Ns\text{ squarefree},
\end{equation*}
and the corresponding subclasses of positive, \strat{} \voa{}s are, up to isomorphism,
\begin{equation}\label{eqn:witt1}
\begin{split}
&\bigl(\bigl[\mathcal{C}((\sqrt{k}A_1)'/(\sqrt{k}A_1))\bigr],1\bigr)_+\\[-2mm]
&=\{V_{\sqrt{k}\ell A_1} \}_{\ell\in\Ns}\cup\{V_{\sqrt{k}\ell A_1}^+\}_{\ell\in\Ns}\cup
\begin{cases}
\emptyset, & k\neq 1,\\
\{V(E_6),V(E_7),V(E_8)\}, & k=1,
\end{cases}
\end{split}
\end{equation}
all still under the assumption of the above conjecture and provided that $V(E_8)$ is \strat{}.

\smallskip

On the other hand, using the same logic we deployed in the analysis of the Virasoro minimal models, the \voa{}s $\smash{V^+_{\sqrt{m}A_1}}$ for $m\in\Ns$ and $V(E_n)$ for $n=6,7,8$ each sit in their own strong Witt class. Two lattice \voa{}s $\smash{V_{\sqrt{k}\ell A_1}}$ and $\smash{V_{\sqrt{k'}\ell'A_1}}$, where $k$ and $k'$ are again squarefree, belong to the same strong Witt class if and only if $k=k'$.
\end{ex}

\begin{rem}
By Theorem~3.15 of \cite{CGH19}, equation~\eqref{eqn:witt1} fully captures all of the \emph{unitary} \voa{}s which are Witt equivalent to a rank-$1$ lattice \voa{}. Combining this with the mild hypothesis that all strongly rational conformal subalgebras of a unitary strongly rational \voa{} are also unitary, equation~\eqref{eqn:witt1} fully exhausts all of the \emph{positive} \voa{}s which are Witt equivalent to a rank-$1$ lattice \voa{}.
\end{rem}

\begin{exph}[RCFTs with $(c_L,c_R)=(1,1)$]\label{ex:rcftsc=1wittclasses}
In the previous example, we analyzed the \strat{}, $c=1$ chiral algebras. In this example, we study how the full unitary RCFTs with $(c_L,c_R)=(1,1)$ organize into Witt classes. Assuming the conjectural classification of positive, \strat{}, $c=1$ \voa{}s, one finds that every unitary, $(c_L,c_R)=(1,1)$ RCFT is either
\begin{enumerate}
\item a compact free boson CFT $\mathcal{T}^{\mathrm{circ}}_R$ of radius $R=\sqrt{2p/q}$,
\item a charge conjugation orbifold $\smash{\mathcal{T}^{\mathrm{orb}}_R\coloneqq\mathcal{T}^{\mathrm{circ}}_R\big/\Z_2^{\mathrm{C}}}$ of a compact free boson CFT of radius $\smash{R=\sqrt{2p/q}}$, or
\item one of three exceptional theories $\mathcal{T}_{E_6}$, $\mathcal{T}_{E_7}$ or $\mathcal{T}_{E_8}$, which are constructed as the canonical diagonal RCFTs built on the chiral algebras $V(E_n)$.
\end{enumerate}
Above, $p,q\in\Ns$ are coprime, and we work in conventions where $\smash{\mathcal{T}^{\mathrm{circ}}_{\sqrt{2}}}$ corresponds to the $\mathrm{SU}(2)_1$ Wess--Zumino--Witten model. Then the circle and orbifold branches meet at $\smash{\mathcal{T}^{\mathrm{circ}}_{2\sqrt{2}}\cong\mathcal{T}^{\mathrm{orb}}_{\sqrt{2}}}$. T-duality, which asserts the equivalences $\smash{\mathcal{T}_R^{\mathrm{circ}}\cong\mathcal{T}_{2/R}^{\mathrm{circ}}}$ and $\smash{\mathcal{T}^{\mathrm{orb}}_R\cong\mathcal{T}^{\mathrm{orb}}_{2/R}}$, imposes additional identifications on the collection of unitary $(c_L,c_R)=(1,1)$ RCFTs.

Let $R=\smash{\sqrt{2p/q}}$. Note that $\smash{\mathcal{T}^{\mathrm{circ}}_R}$ contains the \voa{} $\smash{V_{\sqrt{2pq}\Z}}$ as part of its left- and right-moving chiral algebras, and likewise $\smash{\mathcal{T}^{\mathrm{orb}}_R}$ contains the \voa{} $\smash{V_{\sqrt{2pq}\Z}^+}$ as part of its left- and right-moving chiral algebras. An elementary calculation then shows that $\smash{\mathcal{T}^{\mathrm{x}}_R}$ is Witt equivalent to $\smash{\mathcal{T}^{\mathrm{y}}_{R'}}$ if and only if $R/R'$ is a rational number, where $\mathrm{x},\mathrm{y}\in\{\mathrm{circ},\mathrm{orb}\}$. Finally, the exceptional RCFTs $\mathcal{T}_{E_n}$ are all Witt equivalent to $\mathcal{T}^{\mathrm{circ}}_{\sqrt{2}}$.
\end{exph}


\section{Topological Manipulations}\label{sec:orb}

In this section, we explain what it means for two \strat{} \voa{}s (or RCFTs) to be related by topological manipulations. We shall formalize this concept by defining the mathematical notions of \emph{orbifold equivalence} and \emph{inner orbifold equivalence}, which are inspired by rational equivalence for even lattices, \autoref{defi:lattice_rat}, and discuss the relations of these notions to orbifolding in physics. (The notion of orbifold equivalence was first introduced in \cite{Hoe03}, though it was not given this name.) We also introduce a closely related physical notion of \emph{interface equivalence} and its relationship to orbifolding. Then, in \autoref{subsec:relationshiptogenwitt}, we shall explore the interplay of these definitions with weak and strong Witt equivalence.


\subsection{Orbifold Equivalence}\label{subsec:orb}

Before providing the definition of orbifold equivalence, let us set the stage by reviewing how a physicist would perform an orbifold of a chiral CFT (or holomorphic \voa{})~$V$.

Let $\mathcal{F}$ be a fusion category of topological line operators of a theory $V$. The inequivalent ways of ``gauging'' the symmetry $\mathcal{F}$ are in bijection with any of the following three mathematical structures \cite{BT18}:
\begin{enumerate*}
\item indecomposable module categories of $\mathcal{F}$,
\item Morita equivalence classes of gaugeable algebras of $\mathcal{F}$ or
\item Lagrangian algebras of $Z(\mathcal{F})$.
\end{enumerate*}
The idea is as follows. Gauging proceeds in two steps:
\begin{enumerate}
\item Passing to $V^{\mathcal{F}}$, the conformal subalgebra of operators in $V$ that commute with the lines in $\mathcal{F}$, which has $\Rep(V^{\mathcal{F}})\cong Z(\mathcal{F})$.\footnote{This is a generalization of the result of \cite{Kir02} to non-invertible symmetries.}
\item ``Adding in twisted sectors'', which corresponds to conformally extending from $V^{\mathcal{F}}$ to a new holomorphic \voa{} $V'$.
\end{enumerate}
The different ways of performing the conformal extension of the second step are labeled precisely by module categories of $\mathcal{F}$ (or gaugeable algebras of $\mathcal{F}$ or Lagrangian algebras of $Z(\mathcal{F})$, equivalently). If $V$ and $V'$ are related by such a procedure, we say that they are related by performing a generalized orbifold and write $V'\cong V/\mathcal{M}$ if $V'$ is obtained by gauging a module category $\mathcal{M}$ of $\mathcal{F}$.

The description given above suggests the following equivalent characterization of what it means to be related by an orbifold:

\begin{propph}\label{lem:orbifoldholo}
Let $V$ and $V'$ be two \strat{}, holomorphic \voa{}s. Then $V'$ can be obtained from $V$ by performing a generalized orbifold if and only if there is a \strat{} \voa{} $W$ that conformally embeds into both $V$ and $V'$.
\end{propph}
\begin{proof}
The forward direction follows immediately from the discussion above by taking $W=V^{\mathcal{F}}$.

On the other hand, assume that there is a common \strat{} conformal subalgebra $W$ inside both $V$ and $V'$. By the symmetry-subalgebra duality (see \autoref{subsec:symmetrysubalgebraduality}), we may form the category $\operatorname{Ver}(V/W)$ of topological lines of $V$ that commute with the subalgebra $W$. By construction, $W=V^{\operatorname{Ver}(V/W)}$, and so it follows that $V$ and $V'$ are related by gauging a module category of $\operatorname{Ver}(V/W)$.
\end{proof}

Although a physicist would likely accept the arguments given above, they would not parse for a mathematician. So we separate out the following special case of \autoref{lem:orbifoldholo}, which can be made rigorous with currently available tools.

\medskip

First, we recall the (group-like) orbifold construction for holomorphic \voa{}s. Suppose that $V$ is a \strat{}, holomorphic \voa{} and that $G<\Aut(V)$ is a finite group of automorphisms acting (faithfully) on $V$. We assume, as is widely believed to be the case, that $V^G$ is also \strat{}. If $G$ is solvable (e.g., abelian), then this is a theorem \cite{Miy15,CM16} (see also \cite{McR21b,McR26}). We also assume that $V^G$ satisfies the positivity condition, i.e.\ that all irreducible $V^G$-modules except for $V^G$ have only positive $L_0$-weights (cf.\ \autoref{prop:vectpos}). Then the representation category of $V^G$ is a twisted Drinfeld double $\Rep(V^G)\cong\mathcal{D}_\omega(G)=Z(\Vect_G^\omega)$ for some $3$-cocycle $\omega\in H^3(G,\C^\times)$. Indeed, this was conjectured in \cite{DVVV89,DPR90}, proved for cyclic $G$ in \cite{Moe16,EMS20a} and in full generality in \cite{DNR25} (see also \cite{Kir02}). The (holomorphic) orbifold construction then takes $V^G$ and conformally extends it again to a \strat{}, holomorphic \voa{} $V'$,
\begin{equation*}
V\supset V^G\subset V'.
\end{equation*}
These extensions are controlled by Lagrangian condensable algebras in $\mathcal{D}_\omega(G)$, which are described in Corollary~2 of \cite{EG22}. We may and do assume that the holomorphic extension $V'$ of $V^G$ cannot already be obtained as an extension of $V^H$ for some proper subgroup $H<G$ (like the original \voa{} $V$, for which we can take $H$ to be trivial). Then, the results of \cite{EG22} tell us that such a holomorphic extension $V'$ of $V^G$ only exists if $\omega\in H^3(G,\C^\times)$ is trivial, or if $G$ is \emph{non-anomalous} in physics language. Moreover, such an extension is not generally unique, but rather depends on a choice of $\psi\in H^2(G,\C^\times)$, called \emph{discrete torsion} by physicists. We say that $V'$ is an \emph{orbifold (construction)} of $V$ associated with $G$ and with discrete torsion $\psi$.

We now state the following characterization of the orbifold construction in the special case of abelian groups:
\begin{prop}\label{prop:orbifoldholoabelian}
Let $V$ and $V'$ be two \strat{}, holomorphic \voa{}s. Then $V'$ can be obtained from $V$ by performing an orbifold construction by an abelian group (possibly with non-trivial discrete torsion) if and only if there is a positive, \strat{} \voa{} $W$ such that $\Rep(W)$ is pointed and $W$ conformally embeds into $V$ and $V'$.
\end{prop}
In particular, this means that there is an inverse orbifold construction by an abelian subgroup $G'<\Aut(V')$, necessarily isomorphic to $G$, i.e.\ $V$ is a conformal extension of $\smash{V'^{G'}\subset V'}$.
\begin{proof}
If $V'$ is obtained from $V$ by an orbifold construction for an abelian (and in particular solvable) subgroup $G<\Aut(V)$, then $W\coloneqq V^G$ is \strat{} and, by assumption, positive. Moreover, $\Rep(V^G)\cong\mathcal{D}_\omega(G)$ is pointed when $G$ is abelian and $\omega$ trivial (see, e.g., Corollary~4.3 in \cite{MN18}). Also, by definition, $W$ conformally embeds into $V$ and $V'$. This proves the forward direction.

For the reverse direction, let $W$ be a positive, \strat{} \voa{} with $\Rep(W)$ pointed, i.e.\ $\Rep(W)\cong\mathcal{C}(D)$ for some metric group $D$. Suppose further that $W$ conformally extends to the holomorphic \voa{}s $V$ and $V'$. Then, $V$ is a simple-current extension of $W$ corresponding to an isotropic subgroup $H<D$ with $H^\bot=H$. By the theory of such extensions (see, e.g., \cite{Yam04}), $G\coloneqq\hat{H}=\Hom(H,\C^\times)$ is naturally a subgroup of the automorphism group $\Aut(V)$ and $V^G=W$. Hence, by definition, $V'$ is an orbifold construction of $V$ for the group $G$, or for a proper subgroup of $G$.
\end{proof}

The utility of \autoref{lem:orbifoldholo} (or \autoref{prop:orbifoldholoabelian}) is that it allows us to generalize the notion of being ``orbifold-related'' to (in general non-holomorphic) \strat{} \voa{}s in a rigorous manner.

\begin{defi}[Orbifold Equivalence]\label{defi:orbequiv}
Two \strat{} \voa{}s $V$ and $V'$ are \emph{orbifold equivalent} if there is a sequence of \strat{} \voa{}s $V=V_1,V_2,\dots,V_n=V'$ such that $V_i$ and $V_{i+1}$ share a common (up to isomorphism) \strat{} conformal subalgebra $W_i$, as depicted in \autoref{fig:orbequiv}.
\end{defi}
\begin{figure}[ht]
\centering
\begin{tikzcd}[column sep=small]
V& & V_2&\dots & V_{n-1} && V'& \text{Orb.\ equiv.}\\
& W_1\arrow[hookrightarrow]{ul}{\text{ext.}}\arrow[hookrightarrow]{ur}[swap]{\text{ext.}} && \dots\arrow[hookrightarrow]{ul}{\text{ext.}}\arrow[hookrightarrow]{ur}[swap]{\text{ext.}} && W_{n-1}\arrow[hookrightarrow]{ul}{\text{ext.}}\arrow[hookrightarrow]{ur}[swap]{\text{ext.}}
\end{tikzcd}
\caption{Orbifold equivalence of \strat{} \voa{}s $V$ and $V'$.}
\label{fig:orbequiv}
\end{figure}

\begin{rem}\label{rem:hyporb}
From \cite{Rie25}, it is known that $W_i$ can always be obtained as the fixed-point subalgebra with respect to a hypergroup action on $V_i$. (Though it has not yet been established in the literature how to intrinsically define a notion of hypergroup action on \voa{}s that does not make reference to a conformal subalgebra. See, however, \cite{GR26} for a physics perspective on this issue.) Thus, one can think of \autoref{defi:orbequiv} as a generalization of \autoref{prop:orbifoldholoabelian} that captures what it means for two \strat{} \voa{}s to be related by hypergroup orbifolds (as opposed to orbifolds by finite abelian groups).
\end{rem}
The formulation of orbifold equivalence in \autoref{defi:orbequiv} suggests that it is a vertex algebraic generalization of the notion of rational equivalence of even lattices, \autoref{defi:lattice_rat} and \autoref{fig:ratlat}.

Indeed, in analogy to \autoref{rem:wittlat} and \cite[\autoref*{MR1prop:bulkcomp}]{MR24a}, one can prove that the map from rational equivalence classes of positive-definite, even lattices to orbifold classes of \voa{}s defined by mapping the class $[L]$ of a lattice $L$ to the orbifold class of the lattice \voa{} $V_L$ is well-defined and injective:
\begin{prop}\label{prop:orblatvoa}
Two positive-definite, even lattices $L$ and $M$ are rationally equivalent if and only if $V_L$ and $V_M$ are orbifold equivalent \voa{}s.
\end{prop}
That is, our notion of orbifold equivalence actually generalizes rational equivalence of (positive-definite) even lattices. Once again, this is by design.
\begin{proof}
By assumption, $L$ and $M$ have a common full-rank sublattice $K$. This implies that the lattice \voa{} $V_K$ is a conformal \vosa{} of $V_L$ and $V_M$, and hence that the latter are orbifold equivalent.

Conversely, suppose that the lattice \voa{}s $V_L$ and $V_M$ are orbifold equivalent. Then, by \autoref{prop:subalgebraimpliesWitt} (or \autoref{prop:orbimplieswitt} below), they are Witt equivalent. But we proved in \autoref{prop:wwittcomp} that this implies that the even lattices $L$ and $M$ are Witt equivalent, which by \autoref{prop:lat_ratwit} is the same as rational equivalence.
\end{proof}

\begin{rem}\label{rem:nontrans}
Naively, one might think that one could say that $V$ and $V'$ are orbifold equivalent if they share a \strat{} conformal subalgebra $W$, corresponding to taking $n=2$ in \autoref{defi:orbequiv}. However this appears not to be an equivalence relation, the failure of transitivity being the issue. This is in contrast to the situation for even lattices (see \autoref{rem:trans}).

We consider the following (conjectural) counterexample, which was pointed out to us by Gerald Höhn, see \cite{Hoe03} (and also the discussion after Theorem~5.4 of \cite{BDG23}). We recall from \autoref{ex:cc1} the conjecture that all \strat{} \voa{}s of central charge $c=1$ satisfying the positivity condition are rank-$1$ lattice \voa{}s or \fpvosa{}s thereof under finite groups of automorphisms.

Now, we consider the \fpvosa{}s $V=\smash{\mathsf{A}_{1,1}^{G(E_8)}}$ and $V'=\smash{\mathsf{A}_{1,1}^{G(E_7)}}$ of $\mathsf{A}_{1,1}=V_{A_1}$, where the octahedral group $G(E_7)$ and the icosahedral group $G(E_8)$ are the finite subgroups of $\mathrm{PSL}(2,\C)=\Aut(\mathsf{A}_{1,1})$ corresponding to $E_7$ and $E_8$, respectively, in the ADE classification (see \autoref{ex:cc1}). These both trivially share a common \strat{} conformal subalgebra with $\mathsf{A}_{1,1}$, but they cannot share a common \strat{} conformal subalgebra $W$ with each other, disproving the transitivity.

Indeed, if the above conjecture holds, and if such a $W$ existed, it would again have to be of the form $\mathsf{A}_{1,1}^G$ for some finite subgroup $G<\mathrm{PSL}(2,\C)$. But there is no \emph{finite} subgroup of $\mathrm{PSL}(2,\C)$ containing both $G(E_7)$ and $G(E_8)$. Hence, $W$ cannot be \strat{}.
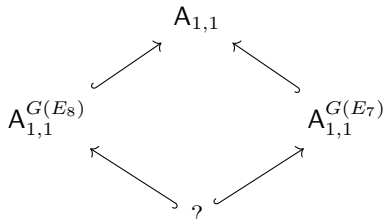
\begin{equation*}
\begin{tikzcd}
&\mathsf{A}_{1,1}\\
\mathsf{A}_{1,1}^{G(E_8)}\arrow[hookrightarrow]{ur}&&\mathsf{A}_{1,1}^{G(E_7)}\arrow[hookrightarrow]{ul}\\
& ?\arrow[hookrightarrow]{ul}\arrow[hookrightarrow]{ur}
\end{tikzcd}
\end{equation*}
Certainly, the intersection $V\cap V'=\mathsf{A}_{1,1}^{\langle G(E_7),G(E_8)\rangle}$, after choosing concrete realizations of $G(E_7)$ and $G(E_8)$ inside $\mathrm{PSL}(2,\C)$, is a conformal subalgebra of $V$ and $V'$, but it is most likely not \strat{}; it certainly is not if we believe the conjecture.

To establish this counterexample, we do not need the full strength of the aforementioned conjecture, but it suffices to have classified all \strat{} conformal subalgebras of $\mathsf{A}_{1,1}$.
\end{rem}

\smallskip

For full RCFTs, we define orbifold equivalence in terms of the more standard notion of gauging generalized global symmetries.\footnote{See, e.g., \cite{BT18} for the relevant background on gauging finite symmetries of 2d CFTs.}

\begin{defiph}[Orbifold Equivalence]
Two RCFTs $\mathcal{T}$ and $\mathcal{T}'$ are orbifold equivalent if there is a sequence of theories $\mathcal{T}=\mathcal{T}_1,\mathcal{T}_2,\dots,\mathcal{T}_n=\mathcal{T}'$, fusion categories $\mathcal{F}_i$ acting on $\mathcal{T}_i$, and module categories $\mathcal{M}_i$ of $\mathcal{F}_i$ such that $\mathcal{T}_i/\mathcal{M}_i\cong \mathcal{T}_{i+1}$.
\end{defiph}

\begin{remph}\label{rem:nottransRCFT}
Again, one is tempted to declare that one can always take $n=2$ in the above definition of orbifold equivalence, but this does not appear to define a transitive relation. This may seem surprising to a physicist who is familiar with orbifold groupoids \cite{GK21}.\footnote{We thank Justin Kulp, Sahand Seifnashri and Yifan Wang for useful related discussions.} Indeed, the entire point of the orbifold groupoid picture is to describe how to compose two orbifold operations so that they can be performed ``in one shot'',
\begin{equation}\label{eqn:oneshotorbifold}
\mathcal{T}/\mathcal{M}_1/\mathcal{M}_2 \cong \mathcal{T}/\mathcal{M}_3.
\end{equation}
To resolve this tension, recall that when orbifolding a symmetry $\mathcal{F}$ using a module category $\mathcal{M}_1$, the orbifolded theory $\mathcal{T}/\mathcal{M}_1$ is guaranteed to have a ``dual'' symmetry, which mathematically is described by the category of $\mathcal{F}$-module functors from $\mathcal{M}_1$ to itself, $\operatorname{Fun}_{\mathcal{F}}(\mathcal{M}_1,\mathcal{M}_1)^{\mathrm{op}}$.\footnote{If $\mathcal{C}$ is a fusion category, then $\mathcal{C}^{\mathrm{op}}$ denotes the \emph{opposite} fusion category, with the tensor product reversed; see, e.g., Definition~2.1.5 in \cite{EGNO15}.} In the orbifold groupoid, once one has orbifolded a module category $\mathcal{M}_1$ of $\mathcal{F}$, the orbifolds one performs after that are only allowed to involve module categories $\mathcal{M}_2$ of this dual category $\operatorname{Fun}_{\mathcal{F}}(\mathcal{M}_1,\mathcal{M}_1)^{\mathrm{op}}$. Under these conditions, the orbifolds by $\mathcal{M}_1$ and $\mathcal{M}_2$ can be composed and performed in one shot by a module category $\mathcal{M}_3$ of $\mathcal{F}$. But more generally, one could choose $\mathcal{M}_2$ such that it is a module category of some symmetry that is not contained in $\operatorname{Fun}_{\mathcal{F}}(\mathcal{M}_1,\mathcal{M}_1)^{\mathrm{op}}$, and in this situation it is not clear that it is always possible to find a \emph{finite} symmetry of $\mathcal{T}$, and a corresponding module category $\mathcal{M}_3$ such that equation~\eqref{eqn:oneshotorbifold} holds.

As an example, consider the \voa{} $V=\smash{\mathsf{A}_{1,1}^{G(E_8)}}$ described in \autoref{rem:nontrans}, and let $\mathcal{T}$ be the diagonal RCFT constructed from $V$. It is believed that any finite-rank fusion category symmetry acting on $\mathcal{T}$ can be embedded into its (finite-rank) category of Verlinde lines, $\operatorname{Ver}(\mathcal{T}\,/\,V,\!V)$. Indeed, suppose there is some fusion category $\mathcal{F}$ that does not commute with the chiral algebra $V$. Then one can form the subalgebra $V^{\mathcal{F}}$ of operators in $V$ that commute with the lines in $\mathcal{F}$. Such a subalgebra $V^{\mathcal{F}}$ would ostensibly define a new positive, \strat{}, $c=1$ \voa{}, contradicting the conjectured classification in \autoref{ex:cc1}.

Now, since one can go from the $\mathrm{SU}(2)_1$ Wess--Zumino--Witten model to $\mathcal{T}$ by gauging a $G(E_8)$ symmetry, one can go in the reverse direction by gauging a $\Rep(G(E_8))$ symmetry. From the $\mathrm{SU}(2)_1$ theory we may perform orbifolds to get to any of the theories $\smash{\mathcal{T}^{\mathrm{circ}}_{\sqrt{2}p/q}}$ described in \autoref{ex:rcftsc=1wittclasses}. Therefore, $\mathcal{T}$ is orbifold equivalent to all of the theories $\smash{\mathcal{T}^{\mathrm{circ}}_{\sqrt{2}p/q}}$. However, as we argued in the previous paragraph, there are only finitely many finite-rank fusion categories acting on $\mathcal{T}$, and hence only finitely many ``one-shot'' orbifolds one can perform on $\mathcal{T}$. So, one must genuinely iteratively perform several orbifolds starting from $\mathcal{T}$ in order to reach most of the theories $\mathcal{T}^{\mathrm{circ}}_{\sqrt{2}p/q}$.

We stress that the theory $\mathcal{T}$ still has infinitely many topological line operators, e.g., because it is obtained by orbifolding the $\mathrm{SU}(2)_1$ Wess--Zumino--Witten model, which has a continuous symmetry group. Most of these lines, however, do not participate in a finite-rank fusion category. Presumably, if one insists on going from $\mathcal{T}$ to $\smash{\mathcal{T}^{\mathrm{circ}}_{\sqrt{2}p/q}}$ via a one-shot orbifold, one must gauge algebra objects of finite quantum dimension inside the full, \emph{infinite-rank} category of topological line operators of $\mathcal{T}$.

We suspect that there should be a general notion of gauging algebra objects of finite quantum dimension inside categories of topological line operators of infinite rank. This is probably what is needed to make one-shot orbifolding an equivalence relation. We leave a more detailed investigation of this possibility to future work.
\end{remph}


\subsection{Interface Equivalence}\label{subsec:int}

Recall that 3d TQFTs are related by generalized orbifolding if and only if they can be separated by a topological interface \cite{Mul24}. Inspired by this result, we introduce the following definition for \strat{} \voa{}s.

\begin{defiph}[Interface Equivalence]
Two \strat{} \voa{}s $V$ and $W$ are \emph{interface equivalent} if there exists a topological interface $\mathcal{I}$ separating the 3d TQFTs $(\Rep(V),c(V))$ and $(\Rep(W),c(W))$ that can terminate on a topological line interface $I$ of finite quantum dimension between the gapless chiral boundaries defined by $V$ and $W$, as in \autoref{fig:interfaceequivalence}.
\end{defiph}

\begin{figure}[ht]
\begin{center}
\begin{tikzpicture}
\filldraw[fill=blue!30,fill opacity=.7,thick,draw=none] (0,-2.5) -- (-6,-1-1.5) -- (-5,-1) -- (1,-1) -- cycle;
\filldraw[fill=red!30,fill opacity=.7,thick,draw=none] (-3,0) -- (1-3,1.5) -- (1-3,-1) -- (-3,-1-1.5) -- cycle;
\draw[very thick,purple] (1-3,-1) -- (-3,-1-1.5);
\node[] at (-.5-.4,-.5) {$\Rep(W)$};
\node[] at (-.5-4+.5,-.5) {$\Rep(V)$};
\node[] at (-2.5,-.5) {$\mathcal{I}$};
\node[] at (-2.5-3+1.2,-1.75) {$V$};
\node[] at (-2.5+3-1.2,-1.75) {$W$};
\node[] at (-2.5+3-1.2-1.4,-1.75) {$I$};
\end{tikzpicture}
\end{center}
\caption{Pictorial representation of interface equivalence.}\label{fig:interfaceequivalence}
\end{figure}
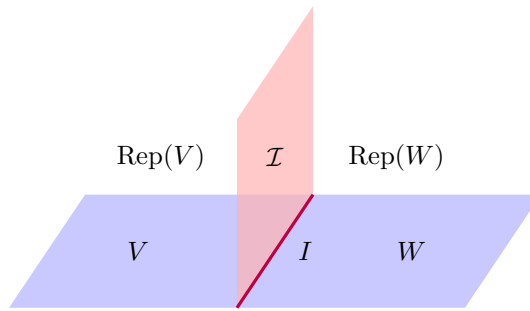

The following will be useful in later sections.
\begin{propph}\label{prop:technicallemma}
If $V$ and $W$ are interface equivalent, then any topological interface $\mathcal{I}'$ between $\Rep(V)$ and $\Rep(W)$ can terminate on a topological line interface $I'$ of finite quantum dimension between the gapless chiral boundaries defined by $V$ and $W$.
\end{propph}
\begin{proof}
We content ourselves with sketching the physical intuition. First, suppose that a 2d QFT $\mathcal{T}$ admits a fusion category symmetry $\mathcal{F}$, and a boundary condition~$B$. By fusing the lines in $\mathcal{F}$ onto the boundary condition $B$, one generates an ``$\mathcal{F}$-multiplet''  $\mathcal{B}$ of boundary conditions (mathematically described by an $\mathcal{F}$-module category, see, e.g., \cite{CRSS23} for a recent review) that transform into one another under the action of $\mathcal{F}$. It is known that if one performs an orbifold of $\mathcal{T}$ by a module category $\mathcal{M}$, then the orbifolded theory $\mathcal{T}/\mathcal{M}$ will possess a category of boundary conditions that is dual to $\mathcal{B}$ (see, e.g., \cite{HLS21,CRZ24} for more details). Thus, the existence of a boundary condition for $\mathcal{T}$ guarantees the existence of a boundary condition for any of its orbifolds $\mathcal{T}/\mathcal{M}$.

We can apply the same way of thinking to $2$-dimensional interfaces, as opposed to 2d QFTs. Indeed, a topological interface $\mathcal{I}$ between the 3d TQFTs $\Rep(V)$ and $\Rep(W)$ always supports a fusion category $\mathcal{F}$ of topological line operators on its world-volume, and it is always possible to reach any other topological interface $\mathcal{I}'$ by performing a 1-gauging of $\mathcal{F}$ \cite{RSS23}. Using the same kinds of arguments as in the previous paragraph, the existence of the line interface $I$ between the $V$ and $W$ boundaries on which $\mathcal{I}$ terminates guarantees, for any other topological interface~$\mathcal{I}'$, the existence of a line interface $I'$ between the $V$ and $W$ boundaries on which $\mathcal{I}'$ can terminate.
\end{proof}

\begin{propph}\label{prop:interfaceorb}
If two \strat{} \voa{}s $V$ and $V'$ are orbifold equivalent, then they are interface equivalent.
\end{propph}
\begin{proof}
We first show that if $V$ and $V'$ share a \strat{} conformal subalgebra $W$, then they are interface equivalent. To see this, we start with the TQFT $(\Rep(W),c(W))$ in the presence of its gapless chiral boundary defined by $W$. Recall that $V$ defines a condensable algebra $A(V)$ in $\Rep(W)$, and that performing anyon condensation of $(\Rep(W),c(W))$ using $A(V)$ in the presence of the $W$ boundary produces the TQFT $(\Rep(V),c(V))$ in the presence of the $V$ boundary. If we perform this anyon condensation in half of spacetime, then we obtain an interface between $\Rep(V)$ and $\Rep(W)$, ending on a line interface between the $V$ and $W$ boundaries. Similar statements apply replacing $V$ with $V'$. If we fuse these two (line) interfaces together, then we see that $V$ and $V'$ are interface equivalent. See \autoref{fig:interfaceorb}.

Now, if two \strat{} \voa{}s $V$ and $V'$ are orbifold equivalent, then there is a sequence $V=V_1,\dots,V_n=V'$ of \strat{} \voa{}s such that $V_i$ shares a \strat{} conformal subalgebra $W_i$ with $V_{i+1}$. We may apply the arguments of the previous paragraph to each pair of $V_i$ and $V_{i+1}$ to produce interfaces $(\mathcal{I}_i,I_i)$ between them and then fuse these all together to produce an interface $(\mathcal{I},I)$ between $V$ and $V'$.
\end{proof}

\begin{figure}[ht]
\begin{center}
{\footnotesize
\begin{tikzpicture}[scale=.57]
\filldraw[fill=blue!30,fill opacity=.7,draw=none] (0,-2.5) -- (-3,-1-1.5) -- (-2,-1) -- (1,-1) -- cycle;
\node[] at (-.5-.4,-.5) {$\Rep(W)$};
\node[] at (-2.5+3-1.2-.4,-1.75) {$W$};
\node[] at (2.5,-.5) {$\xrightarrow{\text{condense}}$};
\end{tikzpicture}\hspace{-.1in}
\begin{tikzpicture}[scale=.57]
\filldraw[fill=blue!30,fill opacity=.7,draw=none] (0,-2.5) -- (-6,-1-1.5) -- (-5,-1) -- (1,-1) -- cycle;
\filldraw[fill=red!30,opacity=.7,draw=none] (-3-1,0) -- (1-3-1,1.5) -- (1-3-1,-1) -- (-3-1,-1-1.5) -- cycle;
\node[] at (-2.5+3-1.2-1.75,-.5) {$\Rep(W)$};
\filldraw[fill=yellow!30,,opacity=.7,draw=none] (-3-1+2,0) -- (1-3-1+2,1.5) -- (1-3-1+2,-1) -- (-3-1+2,-1-1.5) -- cycle;
\draw[very thick,red] (1-3+1,-1) -- (-3+1,-1-1.5);
\draw[very thick,blue] (1-3-1,-1) -- (-3-1,-1-1.5);
\node[] at (-.5-.4+1,-.5) {$\Rep(V')$};
\node[] at (-.95-4,-.5) {$\Rep(V)$};
\node[] at (-2.5-3+1.2-.4,-1.75) {$V$};
\node[] at (-2.5+3-1.2+.4,-1.75) {$V'$};
\node[] at (-2.5+3-1.2-2,-1.75) {$W$};
\end{tikzpicture}
\begin{tikzpicture}[scale=.57]
\node[] at (-6,-.5) {$\xrightarrow{\text{fuse}}$};
\filldraw[fill=blue!30,fill opacity=.7,draw=none] (0,-2.5) -- (-6,-1-1.5) -- (-5,-1) -- (1,-1) -- cycle;
\filldraw[fill=orange!30,  fill opacity=.7,draw=none] (-3,0) -- (1-3,1.5) -- (1-3,-1) -- (-3,-1-1.5) -- cycle;
\draw[very thick,purple] (1-3,-1) -- (-3,-1-1.5);
\node[] at (-.5-.4,-.5) {$\Rep(V')$};
\node[] at (-.5-4+.5,-.5) {$\Rep(V)$};
\node[] at (-2.5-3+1.2,-1.75) {$V$};
\node[] at (-2.5+3-1.2,-1.75) {$V'$};
\end{tikzpicture}
}
\end{center}
\caption{Depiction of the proof of \autoref{prop:interfaceorb}.}\label{fig:interfaceorb}
\end{figure}

\begin{remph}
It is tempting to assert, in the spirit of \cite{DLWW24}, that the converse of \autoref{prop:interfaceorb} holds as well. While this certainly seems conceivable to us, it is far from obvious. We will offer further intuition below, in the more familiar (to physicists) context of full RCFTs.
\end{remph}

\smallskip

The notion of interface equivalence can be extended to RCFTs as follows.

\begin{defiph}[Interface Equivalence]
Two RCFTs are interface equivalent if they can be separated by a topological line interface of finite quantum dimension.
\end{defiph}

Again, by using the trick of ``performing an orbifold in half of spacetime'', we can derive the following analog of \autoref{prop:interfaceorb}.

\begin{propph}\label{prop:interfaceorbRCFT}
If two RCFTs are orbifold equivalent, then they are interface equivalent.
\end{propph}

\begin{remph}
One might come away from \cite{DLWW24} with the impression that the converse to \autoref{prop:interfaceorbRCFT} holds. Indeed, op.\ cit.\ shows that if $I$ is a line interface between two RCFTs $\mathcal{T}$ and $\mathcal{T}'$, then the fusion of $I$ with its orientation reversal, $A\coloneqq I^\ast\otimes I$, defines a (non-simple) topological line operator inside $\mathcal{T}$, and in fact an algebra object of finite quantum dimension that can be gauged to reach~$\mathcal{T}'$. However, in this argument, no reference is made to the ambient category in which $A$ sits. Even though $A$ decomposes into finitely many simple line operators by virtue of its finite quantum dimension, it is not clear that the category generated by the simple lines appearing in $A$ has \emph{finite rank}.

Indeed, returning to the example discussed at the end of \autoref{rem:nottransRCFT}, if one composes the half-space gauging line interface between $\mathcal{T}$ and $\mathrm{SU}(2)_1$ with the half-space gauging line interface between $\mathrm{SU}(2)_1$ and $\smash{\mathcal{T}^{\mathrm{circ}}_{\sqrt{2}p/q}}$, then one obtains an interface $I$ between $\mathcal{T}$ and $\smash{\mathcal{T}^{\mathrm{circ}}_{\sqrt{2}p/q}}$. Even though $I^\ast\otimes I$ must contain finitely many simple topological lines, evidently the category these simple lines generate must have infinite rank for most choices of $p$ and $q$.

This leads us back to the idea that orbifold equivalence and interface equivalence should become the same if, in the definition of orbifold equivalence, one is allowed to gauge algebra objects of finite quantum dimension inside categories of topological line operators that have infinite rank.
\end{remph}

We have defined closely related notions of orbifold equivalence and interface equivalence for both \strat{} \voa{}s and full RCFTs. The following results describe the relationships between the definitions for \voa{}s and those for CFTs.

\begin{propph}\label{prop:orbrelationshipchiralfull}
The RCFTs ${_V}\mathcal{H}_W$ and ${_{V'}}\mathcal{H}'_{W'}$ are orbifold equivalent if and only if $V$ is orbifold equivalent to $V'$ and $W$ is orbifold equivalent to $W'$ as \strat{} \voa{}s.
\end{propph}
\begin{proof}
Suppose that ${_V}\mathcal{H}_W$ is orbifold equivalent to ${_{V'}}\mathcal{H}'_{W'}$. Then there is a sequence of theories
\begin{equation*}
{_V}\mathcal{H}_W={_{V_1}}(\mathcal{H}_1)_{W_1},\dots,{_{V_n}}(\mathcal{H}_n)_{W_n}={_{V'}}\mathcal{H}'_{W'}
\end{equation*}
that are pairwise related via orbifolds by symmetry categories $\mathcal{F}_i$. We assume without loss of generality that the $V_i$ and $W_i$ are taken to be the \emph{maximal} left- and right-moving chiral algebras. In that case, the action of $\mathcal{F}_i$ on ${_{V_i}}(\mathcal{H}_i)_{W_i}$ restricts to an action on the corresponding chiral algebras, and we can therefore define $X_i=V_i^{\mathcal{F}_i}$ and $Y_i=W_i^{\mathcal{F}_i}$. By the nature of the orbifold construction, $X_i$ embeds conformally into $V_{i}$ and $V_{i+1}$, and similarly $Y_i$ embeds conformally into $W_i$ and $W_{i+1}$. Thus, $V$ is orbifold equivalent to $V'$, and $W$ is orbifold equivalent to $W'$.

In the other direction, suppose that there is a sequence $V=V_1,\dots,V_n=V'$ and interpolating subalgebras $X_1,\dots,X_{n-1}$ such that $X_i\subset V_i,V_{i+1}$. Similarly, assume there are \voa{}s $W=W_1,\dots, W_n=W'$ and $Y_1,\dots,Y_{n-1}$ such that $Y_i\subset W_i,W_{i+1}$. As $\mathcal{H}$ defines a Lagrangian algebra in $\Rep(V)\boxtimes \overline{\Rep(W)}$, it follows that $V$ and $W$ have Witt equivalent representation categories. Also, the \voa{}s in the set $\{V_i\}\cup \{X_i\}$ are all Witt equivalent to one another by \autoref{prop:subalgebraimpliesWitt}, as are those in $\{W_i\}\cup \{Y_i\}$. In particular, $V_i$ and $W_i$ have Witt equivalent representation categories, and hence one can always find a Lagrangian algebra $\mathcal{H}_i$ in $\Rep(V_i)\boxtimes \overline{\Rep(W_i)}$. Take $\mathcal{H}_1=\mathcal{H}$ and $\mathcal{H}_n=\mathcal{H}'$. In this way, we obtain a sequence of RCFTs ${_V}\mathcal{H}_W={_{V_1}}(\mathcal{H}_1)_{W_1},\dots,{_{V_n}}(\mathcal{H}_n)_{W_n}={_{V'}}\mathcal{H}'_{W'}$.

To finish the proof, we must simply establish that ${_{V_{i+1}}}(\mathcal{H}_{i+1})_{W_{i+1}}$ can be obtained from ${_{V_{i}}}(\mathcal{H}_{i})_{W_{i}}$ via an orbifold. To see this, note that because of the conformal embeddings $X_i\subset V_i,V_{i+1}$ and $Y_i\subset W_{i},W_{i+1}$ we may replace the triples ${_{V_{i}}}(\mathcal{H}_{i})_{W_{i}}$ and ${_{V_{i+1}}}(\mathcal{H}_{i+1})_{W_{i+1}}$ with ${_{X_i}}(\mathcal{H}_i)_{Y_i}$ and ${_{X_i}}(\mathcal{H}_{i+1})_{Y_i}$, respectively. In particular, this shows that ${_{V_{i+1}}}(\mathcal{H}_{i+1})_{W_{i+1}}$ and ${_{V_{i}}}(\mathcal{H}_{i})_{W_{i}}$ can be thought of as RCFTs which are both built on top of the \emph{same} chiral algebras. It is known that one may always pass between any two RCFTs with the same chiral algebras by orbifolding a module category of the category of Verlinde lines, $\operatorname{Ver}(\mathcal{H}_i\,/\,X_i,\!Y_i)$; see, e.g., Section~5.2 of~\cite{BT18}.
\end{proof}

\begin{propph}\label{prop:rcftvoainterface}
Two RCFTs ${_V}\mathcal{H}_W$ and ${_{V'}}\mathcal{H}'_{W'}$ are interface equivalent if and only if $V$ is interface equivalent to $V'$ and $W$ is interface equivalent to $W'$ as \strat{} \voa{}s.
\end{propph}

\begin{proof}
Suppose that ${_V}\mathcal{H}_W$ and ${_{V'}}\mathcal{H}'_{W'}$ are interface equivalent. By inflating to three dimensions, in the spirit of Kapustin and Saulina \cite{KS11}, we find a picture as on the left of \autoref{fig:KSinterface}. From this picture, one straightforwardly sees that $V$ is interface equivalent to $V'$ and $W$ is interface equivalent to $W'$.

For the reverse direction, we note that it is always possible to find a topological line junction sitting at the intersection of any number of topological interfaces between 3d TQFTs. (This is a generalization of the statement that topological surface operators in a 3d TQFT always support topological line operators on their boundary.) Thus, if there is an interface equivalence $(\mathcal{I}_1,I_1)$ between $V$ and $V'$ and an interface equivalence $(\mathcal{I}_2,I_2)$ between $W$ and $W'$, then we can always find a topological line junction sitting at the four-way intersection of the interfaces $\mathcal{H}$, $\mathcal{I}_1$, $\mathcal{H}'$ and $\mathcal{I}_2$. Thus, we can always form the picture on the left of \autoref{fig:KSinterface}, and upon dimensional reduction along the interval direction we obtain a topological line interface between ${_V}\mathcal{H}_{W}$ and ${_{V'}}\mathcal{H}'_{W'}$.
\end{proof}


\subsection{Inner Orbifold Equivalence}
We shall also discuss the following specialization of orbifold equivalence. For this notion, transitivity (cf.\ \autoref{rem:nontrans}) does indeed hold because it is essentially inherited from the corresponding property for even lattices (see \autoref{rem:trans}):
\begin{defi}[Inner Orbifold Equivalence]\label{defi:innerorbequ}
Two \strat{} \voa{}s $V$ and $V'$ are \emph{inner orbifold equivalent} if there is a sequence of \strat{} \voa{}s $V=V_1,V_2,\dots,V_n=V'$ such that $V_i$ and $V_{i+1}$ share a \strat{} conformal subalgebra $W_i$, and furthermore such that the concrete realization of $W_i$ inside $V_i$ (and similarly for $V_{i+1}$) can be chosen such that $W_i$ and $V_i$ have the same Heisenberg commutant.
\end{defi}

\begin{defiph}[Inner Orbifold Equivalence]
Two RCFTs ${_V}\mathcal{H}_W$ and ${_{V'}}\mathcal{H}'_{W'}$ are \emph{inner orbifold equivalent} if $V$ is inner orbifold equivalent to $V'$ and $W$ is inner orbifold equivalent to $W'$.
\end{defiph}

We shall see in \autoref{rem:innerorbtrans} that we can equivalently demand the (formally stronger) condition that $V$ and $V'$ be connected ``in one step'', i.e.\ that there is a conformal subalgebra $W$ of both $V$ and $V'$ with the same Heisenberg commutant as $V$ and $V'$.
\begin{rem}\label{rem:innerorbtoorb}
It is immediate from the definition that inner orbifold equivalence implies orbifold equivalence, both for \strat{} \voa{}s and for RCFTs.
\end{rem}
We have chosen the label ``inner'' because, in some special cases, inner orbifold equivalence reduces to a (group-like) orbifold construction under inner automorphisms (see \autoref{prop:innerneighchar} and \autoref{rem:nameinner}).

Finally, it follows directly from the definition that inner orbifold equivalence (like orbifold equivalence, see \autoref{prop:orblatvoa}) reduces to rational equivalence of lattices if we consider lattice \voa{}s. Of course, this is by design.

This also means that for lattice \voa{}s, the notions of orbifold equivalence and inner orbifold equivalence coincide (like weak and strong Witt equivalence, see \autoref{rem:wittlat} and \autoref{rem:strongwittlat}).


\section{Relation Between Witt Equivalence and Topological Manipulations}\label{subsec:relationshiptogenwitt}

We now return to the question posed in the introduction: \emph{when are two \strat{} \voa{}s, or two RCFTs, related by topological manipulations?} In \autoref{subsec:necessary}, we argue that Witt equivalence of the two theories is necessary, and we conjecture that it is also sufficient. In \autoref{subsec:partialconverse}, we prove a partial result in the direction of this conjecture, and in \autoref{subsec:exampleswittorb} we give a few examples. In \autoref{subsec:classificationrelation}, we describe a kind of ``structure theorem'' for the classification of RCFTs that follows from this conjecture. In  \autoref{subsec:deformationrelation} we describe the relationship of these ideas to properties of deformation classes of QFTs.


\subsection{Necessary Condition}\label{subsec:necessary}

The following result is a straightforward consequence of the machinery we have developed thus far.
\begin{thm}\label{prop:orbimplieswitt}
If two \strat{} \voa{}s $V$ and $V'$ are orbifold equivalent, then they are Witt equivalent.
\end{thm}
\begin{proof}
By the definition of orbifold equivalence, it suffices to prove that if two \strat{} \voa{}s $V$ and $V'$ share a common \strat{} subalgebra $W$, then they are Witt equivalent. This follows immediately from \autoref{prop:subalgebraimpliesWitt} and the fact that Witt equivalence is an equivalence relation.
\end{proof}

\begin{propph}\label{prop:interfaceimplieswitt}
If two \strat{} \voa{}s $V$ and $V'$ are interface equivalent, then they are Witt equivalent.
\end{propph}
\begin{proof}
If $V$ and $V'$ are interface equivalent, then in particular their corresponding TQFTs $(\Rep(V),c(V))$ and $(\Rep(V'),c(V'))$ can be separated by a topological interface, which is physically what it means for $V$ and $V'$ to be Witt equivalent.
\end{proof}

We extend these two assertions to full RCFTs.
\begin{thmph}\label{thm:orbimplieswittRCFT}
If two RCFTs ${_V}\mathcal{H}_W$ and ${_{V'}}\mathcal{H}'_{W'}$ are orbifold equivalent, then they are Witt equivalent.
\end{thmph}
\begin{proof}
By \autoref{prop:orbrelationshipchiralfull}, $V$ is orbifold equivalent to $V'$, and $W$ is orbifold equivalent to $W'$. Then, by \autoref{prop:orbimplieswitt}, $V$ and $V'$ are Witt equivalent as \strat{} \voa{}s, and so are $W$ and $W'$. This is the definition of ${_V}\mathcal{H}_W$ and ${_{V'}}\mathcal{H}'_{W'}$ being Witt equivalent as RCFTs.
\end{proof}

\begin{thmph}\label{thm:interfaceimplieswittRCFT}
If two RCFTs ${_V}\mathcal{H}_W$ and ${_{V'}}\mathcal{H}'_{W'}$ are interface equivalent, then they are Witt equivalent.
\end{thmph}
\begin{proof}
By \autoref{prop:rcftvoainterface}, $V$ is interface equivalent to $V'$ and $W$ is interface equivalent to $W'$. By \autoref{prop:interfaceimplieswitt}, $V$ is Witt equivalent to $V'$ and $W$ is Witt equivalent to $W'$, which is by definition what it means for ${_V}\mathcal{H}_W$ to be Witt equivalent to ${_{V'}}\mathcal{H}'_{W'}$.
\end{proof}

Thus, we see that Witt inequivalence is a kind of obstruction to being orbifold or interface equivalent. In particular, by \autoref{rem:invariantsofwitt}, one has a criterion for testing whether two RCFTs may be related by a sequence of generalized orbifolds, or separated by a topological line interface: one computes the triple $([\Rep(V)],c_L,c_R)=([\Rep(W)],c_L,c_R)$ corresponding to ${_V}\mathcal{H}_W$ and the triple $([\Rep(V')],c_L',c_R')=([\Rep(W')],c_L',c_R')$ corresponding to ${_{V'}}\mathcal{H}'_{W'}$ and simply checks whether or not they are equal. (Here, e.g., $[\Rep(V)]$ is the Witt class of the modular tensor category $\Rep(V)$.) If the triples are equal, then ${_V}\mathcal{H}_W$ and ${_{V'}}\mathcal{H}'_{W'}$ might be related by an orbifold; if they are not equal, then the theories are definitively not related by an orbifold.

\begin{ex}[Orbifolding \VOA{}s with $c=1$]\label{ex:cc1wittorbifold}
We recall from \autoref{ex:cc1} the conjectural classification of the positive, \strat{} \voa{}s of central charge $c=1$ and how these \voa{}s fall into Witt classes. \autoref{prop:orbimplieswitt} states that for two of these \voa{}s to be orbifold equivalent, they must be in the same Witt class. And, indeed, this is apparent from equation~\eqref{eqn:witt1}.

Conversely, we can read off that all positive, \strat{} \voa{}s in the same Witt class are actually orbifold equivalent. Of course, this depends on the classification conjecture for $c=1$.
\end{ex}

In the above example, we argued that Witt and orbifold equivalence coincide (conjecturally) for positive, \strat{}, $c=1$ \voa{}s. The same can be said about the corresponding RCFTs:
\begin{exph}[Orbifolding RCFTs with $(c_L,c_R)=(1,1)$]\label{exph:orbifoldingc=1}
We recall from \autoref{ex:rcftsc=1wittclasses} that $\mathcal{T}^{\mathrm{x}}_R$ is Witt equivalent to $\mathcal{T}^{\mathrm{y}}_{R'}$ if and only if $R/R'\in\mathbb{Q}$, where $\mathrm{x},\mathrm{y}\in\{\mathrm{circ},\mathrm{orb}\}$. Thus, \autoref{thm:orbimplieswittRCFT} says that one can pass from $\mathcal{T}^{\mathrm{x}}_R$ to $\mathcal{T}^{\mathrm{y}}_{R'}$ via a sequence of orbifolds only if $R/R'\in \mathbb{Q}$. For example, the $\mathrm{SU}(2)_1$ Wess--Zumino--Witten model (corresponding to $\smash{\mathcal{T}^{\mathrm{circ}}_{\sqrt{2}}}$) and the bosonized Dirac fermion (corresponding to $\mathcal{T}_2^{\mathrm{circ}}$) cannot be related by a sequence of orbifolds, in spite of the fact that they are both RCFTs with the same central charges.

Conversely, it is known that one may indeed move between any two of the known $(c_L,c_R)=(1,1)$ unitary RCFTs via a sequence of orbifolds, provided they belong to the same Witt class. For example, one may go from $\smash{\mathcal{T}_{\sqrt{2}}^{\mathrm{circ}}}$ to $\smash{\mathcal{T}_{\sqrt{2}p/q}^{\mathrm{circ}}}$ by gauging a $\Z_p\times \Z_q$ subgroup of the $\mathrm{U}(1)_{\mathrm{mom}}\times \mathrm{U}(1)_{\mathrm{win}}$ momentum and winding symmetries.
\end{exph}

The previous examples might embolden one to conjecture that converses to \autoref{prop:orbimplieswitt} and \autoref{prop:interfaceimplieswitt} (and \autoref{thm:orbimplieswittRCFT} and \autoref{thm:interfaceimplieswittRCFT} in the context of RCFTs) hold. We shall discuss this in detail in the next two sections.


\subsection{Converse Conjectures}

Recall from \autoref{prop:orblatvoa} that orbifold equivalence generalizes rational equivalence of even lattices. Moreover, in the context of lattices, in addition to rational equivalence implying Witt equivalence, also the converse holds (see \autoref{prop:lat_ratwit}). Together with \autoref{prop:wwittcomp}, we obtain:
\begin{prop}
Let $V_L$ and $V_M$ be lattice \voa{}s. Then they are Witt equivalent if and only if they are orbifold equivalent.
\end{prop}
It is an interesting open question to determine to what extent the same is true for \strat{} \voa{}s in general, i.e.\ to what extent the converses to \autoref{prop:orbimplieswitt}, \autoref{prop:interfaceimplieswitt}, \autoref{thm:orbimplieswittRCFT} and \autoref{thm:interfaceimplieswittRCFT} hold. For simplicity, we assume that the \voa{}s are positive.

\begin{conj}\label{conj:wittorb}
If two positive, \strat{} \voa{}s belong to the same Witt class, then they are orbifold equivalent.
\end{conj}
\begin{conjph}\label{conj:wittinterface}
If two positive, \strat{} \voa{}s belong to the same Witt class, then they are interface equivalent.
\end{conjph}
\begin{conjph}\label{conj:wittorbRCFT}
If two RCFTs with positive chiral algebras belong to the same Witt class, then they are orbifold equivalent.
\end{conjph}
\begin{conjph}\label{conj:wittinterfaceRCFT}
If two RCFTs with positive chiral algebras belong to the same Witt class, then they are interface equivalent.
\end{conjph}

These conjectures are not all independent. By following various arrows established so far, we obtain the following diagram:
\begin{equation}\label{eqn:conjectureequivalences}
\begin{tikzcd}
\text{\autoref{conj:wittorb}}\arrow[r, Leftrightarrow]\arrow[d, Rightarrow]
&\text{\autoref{conj:wittorbRCFT}}\arrow[d,Rightarrow]\\
\text{\autoref{conj:wittinterface}}\arrow[r, Leftrightarrow]
&\text{\autoref{conj:wittinterfaceRCFT}}
\end{tikzcd}
\end{equation}

\begin{remph}
If \autoref{conj:wittorb} is true, then orbifold equivalence is the same as interface equivalence (because they are both the same as Witt equivalence).
\end{remph}

\smallskip

We shall not establish these conjectures in this work. This is probably a very difficult problem, which, as we shall see in \autoref{rem:fakemoonshine} below, is likely at least as difficult as proving the moonshine uniqueness, \autoref{conj:moonshineuniqueness}.

However, we can achieve a partial simplification of the problem. Working with an entire Witt class of \voa{}s is psychologically daunting. For instance,  \autoref{ex:countableinfinity} illustrates that there are often infinitely many \voa{}s in a Witt class. We therefore note that \autoref{conj:wittorb} can be recast in the following formally weaker form, which requires one to consider just bulk genera, which are thought to contain only finitely many \voa{}s.
\begin{conj}\label{conj:anisobulkorb}
If two positive, \strat{} \voa{}s belong to the same bulk genus and their representation categories are anisotropic, then they are orbifold equivalent.
\end{conj}

\begin{prop}\label{prop:wittbulkorb}
\autoref{conj:wittorb} and \autoref{conj:anisobulkorb} are equivalent.
\end{prop}
\begin{proof}
\autoref{conj:wittorb} obviously implies \autoref{conj:anisobulkorb} because bulk equivalence implies Witt equivalence.

On the other hand, assume \autoref{conj:anisobulkorb} and consider two positive, \strat{} \voa{}s $V$ and $V'$ belonging to the same Witt class. By \autoref{prop:anisorep}, there is a unique anisotropic representative $\mathcal{C}_{\mathrm{an}}$ of the Witt class $[\Rep(V)]=[\Rep(V')]$. Now, by \autoref{prop:algext}, $V$ and $V'$ can be conformally embedded into \voa{}s $W$ and $W'$, respectively, both belonging to the bulk genus $(\mathcal{C}_{\mathrm{an}},c(V))=(\mathcal{C}_{\mathrm{an}},c(V'))$. \autoref{conj:anisobulkorb} then implies that $W$ is orbifold equivalent to $W'$. Then, because $V$ is orbifold equivalent to $W$ and $V'$ to $W'$, it follows that $V$ is orbifold equivalent to $V'$, and hence that \autoref{conj:wittorb} is true.
\end{proof}

As a special illustrative case, \autoref{prop:wittbulkorb} shows that, in order to prove that any two positive, \strat{} \voa{}s with the same central charge and \emph{Witt trivial}\footnote{Recall from \autoref{sec:cat} that $\mathcal{C}$ being Witt trivial, $[\mathcal{C}]=[\Vect]$, means that it is of the form $Z(\mathcal{F})$ for some spherical fusion category $\mathcal{F}$, or equivalently that $\mathcal{C}$ is condensable to $\Vect$.} representation categories are orbifold equivalent, it is actually sufficient to prove that any two \strat{}, holomorphic \voa{}s of the same central charge are orbifold equivalent. A variant of the latter statement was conjectured to be true in Section~4.2 of \cite{Joh21}. We briefly review the evidence in its favor.
\begin{ex}[Generalized Deep Holes]\label{ex:deepholes}
Recall the classification results of \strat{}, holomorphic \voa{}s surveyed in \cite[\autoref*{MR1ex:schellekens}]{MR24a} and discussed in more detail in \cite[\autoref*{MR1sec:holo}]{MR24a}.

For $c=0,8,16$ the bulk genus $(\Vect,c)$ consists only of lattice \voa{}s. These are then all (inner) orbifold equivalent because, e.g., the $2$-neigh\-bor\-hood graph of the corresponding positive-definite, even, unimodular lattices is connected and lattice neighborhood can clearly be lifted to neighborhood of lattice \voa{}s. See \autoref{ex:c=16orbifolds} and \autoref{sec:neighborhood} below for more details.

For $c=24$, the situation is more difficult because most \voa{}s in the bulk genus $(\Vect,24)$ are not lattice \voa{}s. Nonetheless, the main result of \cite{ELMS21,MS23} establishes that, assuming \autoref{conj:moonshineuniqueness} on the uniqueness of the moonshine module, all \strat{}, holomorphic \voa{}s of central charge~$24$, Schellekens' list \cite{Sch93}, can be obtained as cyclic orbifold constructions of the Leech lattice \voa{} $V_\Lambda$ associated with \emph{generalized deep holes} and are thus all orbifold equivalent to one another. (See also \cite{HM23,DSW23} for other, related approaches.)
\end{ex}
We record this important observation in the following two propositions:
\begin{prop}\label{prop:witttrivialorbifoldingVOAs}
Assuming \autoref{conj:moonshineuniqueness} on the uniqueness of the moonshine module, \autoref{conj:wittorb} is true for the infinitely many positive, \strat{} \voa{}s belonging to the Witt classes $([\Vect],c)$ with $c\leq 24$.
\end{prop}

\begin{propph}\label{prop:witttrivialorbifolding}
Assuming \autoref{conj:moonshineuniqueness} on the uniqueness of the moonshine module,  \autoref{conj:wittorbRCFT} is true for the infinitely many positive RCFTs belonging to the Witt classes $([\Vect],c_L,c_R)$ with $c_L,c_R\leq 24$.
\end{propph}

\begin{rem}\label{rem:fakemoonshine}
Any ``fake'' moonshine module, i.e.\ any \strat{}, holomorphic \voa{} of central charge~$24$ with weight-$1$ space $V_1=\{0\}$ and $V\ncong V^\natural$, would likely yield a counterexample to \autoref{conj:wittorb} and \autoref{conj:anisobulkorb}.

Indeed, it is not difficult to show that if such a ``fake'' moonshine module $V$ existed, it could not be orbifold equivalent under iterated \emph{cyclic} orbifold constructions (see \autoref{sec:neighborhood}) to any of the $71$ Schellekens \voa{}s. (And likely, this argument can be extended to the more general notion of orbifold equivalence given in \autoref{defi:orbequiv}.)

It suffices to show that any cyclic orbifold construction of $V$ is again a ``fake'' moonshine module. To see this, we first note that because $V_1=\{0\}$, any cyclic orbifold construction applied to $V$ must yield a \voa{} with an abelian weight-$1$ Lie algebra, which only leaves the Leech lattice \voa{} $V_\Lambda$, the moonshine module $V^\natural$ or a ``fake'' moonshine module.\footnote{The reason is that a reductive Lie algebra can only have a zero fixed-point Lie subalgebra (under a cyclic group) if it is abelian \cite{Kac90}.} On the other hand, all the cyclic orbifold constructions starting from the moonshine module $V^\natural$ were determined in \cite{PPV16} and they either give $V^\natural$ or $V_\Lambda$. Similarly, it was shown in \cite{Car18} that any cyclic orbifold construction of $V_\Lambda$ with vanishing weight-$1$ space must already be $V^\natural$. This proves the claim.
\end{rem}


\subsection{Partial Converse}\label{subsec:partialconverse}

Establishing \autoref{conj:wittorb} (or \autoref{conj:anisobulkorb}) in full generality appears to be a very difficult problem. For example, as we just saw in \autoref{rem:fakemoonshine}, as a very narrow special case, one would likely obtain the uniqueness of the moonshine module, \autoref{conj:moonshineuniqueness}. Instead, we prove the following partial result, which shows that at least \emph{strongly} Witt equivalent \voa{}s are necessarily orbifold equivalent, and in fact \emph{inner} orbifold equivalent (recalling the notion from \autoref{defi:innerorbequ}).

\begin{thm}\label{thm:rationalinnerorb}
Two \strat{} \voa{}s $V$ and $V'$ are inner orbifold equivalent if and only if they are strongly Witt equivalent.
\end{thm}
\begin{proof}
Suppose that $V$ and $V'$ are inner orbifold equivalent. In each intermediate step, since $W_i$ and $V_i$ have the same Heisenberg commutant $C$, we can choose a Cartan subalgebra of $W_i$ that is also a Cartan subalgebra of $V_i$. Indeed, any Cartan subalgebra of $W_i$ can be extended to a Cartan subalgebra of $V_i$, but the condition on the Heisenberg commutants implies that these Cartan subalgebras have the same dimension.

Hence, $V_i$ has the dual pair $C\otimes V_L$ and $W_i$ the dual pair $C\otimes V_K$ where $K$ is a full-rank sublattice of $L$ because $V_K=\Com_{W_i}(C)\subset\Com_{V_i}(C)=V_L$. Applying the same reasoning to $W_i$ and $V_{i+1}$, we arrive at the conclusion that $V_i$ and $V_{i+1}$ have rationally equivalent associated lattices. The last property is transitive (see \autoref{rem:trans}), resulting in the same statement for $V$ and $V'$. Hence, $V$ and $V'$ are strongly Witt equivalent.

To show the converse direction, suppose that $V$ contains the dual pair $C\otimes V_L$ and $V'$ the dual pair $C\otimes V_M$ (strictly speaking, up to isomorphism) where the associated lattices $L$ and $M$ are rationally equivalent and hence have a common finite-index sublattice $K$. Hence, both $V$ and $V'$ are conformal extensions of $W\coloneqq C\otimes V_K$, and so they are orbifold equivalent. Since $C$ is the Heisenberg commutant of $W$, the proof is complete.
\end{proof}

\begin{rem}\label{rem:innerorbtrans}
The proof also shows that while in general two orbifold equivalent \voa{}s cannot be connected ``in one step'' (see \autoref{rem:nontrans}), this is the case for inner orbifold equivalence (cf.\ \autoref{rem:trans}).
\end{rem}

\begin{thmph}\label{thmph:rcftinnerorb=stronglywitt}
Two RCFTs are inner orbifold equivalent if and only if they are strongly Witt equivalent.
\end{thmph}

As a corollary to \autoref{thm:rationalinnerorb}, we obtain:
\begin{cor}\label{cor:orbifoldrelation}
If $V$ and $V'$ are two \strat{} \voa{}s in the same hyperbolic genus, then they are inner orbifold equivalent.

In the special case where $V$ and $V'$ are further holomorphic, one can be obtained from the other by an orbifold construction for an abelian group of automorphisms, possibly with discrete torsion.
\end{cor}
\begin{proof}
The first assertion immediately follows from \autoref{thm:rationalinnerorb} and the fact that hyperbolic equivalence implies strong Witt equivalence, \autoref{prop:hypstrongwitt}.

Now, because $V$ and $V'$ are inner orbifold equivalent, they are both conformal extensions of $C\otimes V_K$, where $C$ is the Heisenberg commutant of both $V$ and $V'$ and $K$ is some even lattice. In the special case that $V$ and $V'$ are further holomorphic, $C$ has a pointed representation category (see, e.g., Proposition~4.4 in \cite{HM23}), and thus so does $C\otimes V_K$. It then follows from \autoref{prop:orbifoldholoabelian} that $V$ and $V'$ are related by orbifolding an abelian group, possibly with discrete torsion.
\end{proof}

\begin{rem}\label{rem:orbhypgen}
A fortiori, \autoref{cor:orbifoldrelation} establishes that hyperbolic equivalence implies orbifold equivalence. In general, the converse of that assertion fails: famously, the moonshine module $V^\natural$ is a $\Z_2$-orbifold construction of the Leech lattice \voa{} $V_{\Lambda}$ \cite{FLM88}, but $V^\natural$ and $V_{\Lambda}$ certainly belong to different hyperbolic genera, as discussed in \cite[\autoref*{MR1ex:schellekens2}]{MR24a} and in more detail in \cite[\autoref*{MR1sec:holhyp}]{MR24a}.

In fact, as we reviewed in \autoref{ex:deepholes}, any of the $71$ \strat{}, holomorphic \voa{}s of central charge $c=24$ (that are not ``fake'' moonshine modules) can be obtained as cyclic orbifold constructions of $V_{\Lambda}$, in general of order greater than~$2$, but they fall into a total of $12$ hyperbolic genera.

Note that even inner orbifold equivalence is in general weaker than hyperbolic equivalence. In \autoref{sec:neighborhood}, in the case of holomorphic \voa{}s $V$ and $V'$, we make precise what we have to additionally require in order to find a notion of orbifold equivalence (namely inner $p$-neigh\-bor\-hood) that is equivalent to hyperbolic equivalence, at least conjecturally.
\end{rem}


\subsection{Examples}\label{subsec:exampleswittorb}

We now discuss a few examples.
\begin{ex}[Holomorphic \VOA{}s of $c=16$]\label{ex:c=16orbifolds}
The two \strat{}, holomorphic \voa{}s of central charge~$16$ are the lattice \voa{}s $\smash{V_{E_8^2}=\mathsf{E}_{8,1}^{\otimes 2}}$ and $\smash{V_{D_{16}^+}}$, where $\smash{D_{16}^+}$ is an even (unimodular) extension of index~$2$ of the root lattice $D_{16}$. Both are lattice \voa{}s and thus have trivial Heisenberg commutants. It then follows from \cite[\autoref*{MR1cor:holhyp}]{MR24a} that they belong to the same hyperbolic genus.

\autoref{cor:orbifoldrelation} states that it must be possible to go from one \voa{} to the other by performing an inner orbifold. Indeed, note that there are the following conformal extensions of index~$2$
\begin{equation*}
V_{E_8^2}\supset V_{(D_8^2)^+}\subset V_{D_{16}^+},
\end{equation*}
where $\smash{(D_8^2)^+}$ is an even, index-$2$ extension of the root lattice $D_8^2$. In other words, $\smash{V_{E_8^2}}$ and $\smash{V_{D_{16}^+}}$ are $\Z_2$-orbifold constructions of one another. We should point out that in this example, we are simply reproducing the fact that two lattices in the same genus are rationally equivalent. We point to \cite{KLT86} for an early appearance in the physics literature, and to \cite{HM23} for more details.
\end{ex}

In \autoref{sec:neighborhood}, we discuss in detail a special case of orbifold equivalence, namely $n$-neighborhood or being related by a $\Z_n$-orbifold construction. The above is an example of this. In \autoref{subsec:neighborexamples}, we consider the holomorphic \voa{}s of central charge $c=24$.

\begin{exph}[Narain Theories and Deformed Wess--Zumino--Witten Models]\label{exph:generalizednarain}
Consider for simplicity the diagonal $\mathrm{SU}(N+1)_k$ Wess--Zumino--Witten RCFT, $k,N\in\Ns$, though our comments are more general. The chiral algebra of this theory is the affine Kac--Moody algebra $\mathsf{A}_{N,k}$. The Heisenberg commutant or interacting sector of $\mathsf{A}_{N,k}$ is known as a parafermion chiral algebra \cite{DR17} (see also \cite[\autoref*{MR1ex:para}]{MR24a}) and is sometimes written as $C=K(\mathfrak{sl}_{N+1},k)$; the associated lattice is the $A_N$ root lattice with its norms scaled by a factor of $k$.

As reviewed in \cite{DHJ21}, these theories admit $J\bar{J}$ (or current-current) exactly marginal deformations; one might call the resulting conformal manifold the ``generalized Narain moduli space'' of $\mathrm{SU}(N+1)_k$. Such current-current deformations are known to leave the Heisenberg commutant (or interacting sector) unchanged because they only touch the lattice (free sector). In particular, it follows from \autoref{prop:weakstrongwitt2} that if $\mathcal{T}$ and $\mathcal{T}'$ are two RCFTs that are connected to $\mathrm{SU}(N+1)_k$ by current-current deformations, then $\mathcal{T}$ is strongly Witt equivalent to $\mathcal{T}'$ if and only if $\mathcal{T}$ is Witt equivalent to $\mathcal{T}'$.

Then, \autoref{thmph:rcftinnerorb=stronglywitt} asserts that two rational theories on the generalized Narain moduli space of $\mathrm{SU}(N+1)_k$ can be related by orbifolds if and only if they are Witt equivalent. This is a (partial) generalization of our observations in \autoref{exph:orbifoldingc=1}, which pertained to the special case $N=k=1$.
\end{exph}

\begin{ex}[Witt Classes with Anisotropic Representatives of Low Rank]
Let $\mathcal{C}_{\mathrm{an}}$ be an anisotropic unitary modular tensor category with $\rk(\mathcal{C}_{\mathrm{an}})\leq 4$ \cite{RSW09} other than $\smash{\overline{(A_1,5)}_{1/2}}$ (see op.\ cit.\ for the definition of the latter). Further, let $c\in\Q$ be the unique rational number satisfying $0<c\leq 8$ and $e^{2\pi i c/8}=e^{2\pi i c(\mathcal{C}_{\mathrm{an}})/8}$. Here, we recall from \cite[equation~(\ref*{MR1eqn:chiralcentralcharge})]{MR24a} the definition of the chiral central charge $c(\mathcal{C}_{\mathrm{an}})\in \mathbb{Q}/8\mathbb{Z}$ of $\mathcal{C}_{\mathrm{an}}$.

Then, by the results of \cite{Ray23}, there is a unique positive, \strat{} \voa{} $V$ in the bulk genus $(\mathcal{C}_{\mathrm{an}},c)$.\footnote{By the results of \cite{MR24a}, there are in fact no non-positive, \strat{} \voa{}s in the bulk genus $(\mathcal{C}_{\mathrm{an}},c)$.} It follows, using the logic of \autoref{prop:wittbulkorb}, that any two positive, \strat{} \voa{}s in the (infinitely large) Witt class $([\mathcal{C}_{\mathrm{an}}],c)$ of $V$ are orbifold equivalent. Indeed, by \autoref{prop:Wittgenusclassification}, any positive \voa{} that is Witt equivalent to $V$ can be conformally extended to $V$, and hence is orbifold equivalent to $V.$

Using \autoref{prop:orbrelationshipchiralfull}, this result can be extended to full RCFTs. In particular, one finds that the entire Witt class $([\mathcal{C}_{\mathrm{an}}],c_L,c_R)=([\mathcal{C}_{\mathrm{an}}],c,c)$ of the canonical diagonal RCFT built on top of $V$ is connected by orbifolds. For example, every RCFT that is Witt equivalent to the $G_{2,1}$ Wess--Zumino--Witten model can be reached by performing a sequence of orbifolds on the latter.
\end{ex}


\section{Classification and Deformation Classes of CFTs}

In this brief section, we highlight some physical consequences of the conjectures articulated in the previous section.

\subsection{Relationship to RCFT Classification}\label{subsec:classificationrelation}

First, we show that \autoref{conj:wittorb} on the relation between Witt and orbifold equivalence leads to a kind of structure theorem on RCFTs ${_V}\mathcal{H}_W$ for which $V$ and $W$ are positive.

One's first instinct might  be to conjecture that orbifolding is a powerful enough tool to be able to connect any pair of RCFTs with the same central charges. However, \autoref{prop:orbimplieswitt} and \autoref{thm:orbimplieswittRCFT} imply that orbifolding can never move between different Witt classes, so this instinct is too naive. On the other hand, the operation of taking a coset (or commutant) of, say, a \strat{} \voa{} $V$ \emph{can} change its Witt class, both by changing the central charge $c(V)$ and by changing the categorical Witt class $[\Rep(V)]$. Therefore, one might hope that it is possible to use orbifolding and cosets in combination to navigate through the entirety of theory space.

We now make this intuition precise, first for \voa{}s and then for full RCFTs.

\begin{prop}
Assuming \autoref{conj:weakreconstruction} or \autoref{conj:weakreconstruction2} (weak reconstruction) and that \autoref{conj:wittorb} holds for \strat{} \voa{}s~$V'$ with $[\Rep(V')]=[\Vect]$, it follows that any positive, \strat{} \voa{} $V$ can be obtained by starting with $\mathsf{E}_{8,1}^{\otimes n}\cong V_{E_8^n}$ for some sufficiently  large $n\in\mathbb{Z}_{\geq 0}$, performing a sequence of orbifolds and then taking a commutant.
\end{prop}
\begin{proof}
Moving backwards, it follows from \autoref{conj:weakreconstruction2} that $V$ can be embedded into a holomorphic \voa{} $U$ and obtained via the coset construction as $V=\Com_U(W)$, where $W=\Com_U(V)$. On the other hand, $U$ and $\smash{\mathsf{E}_{8,1}^{\otimes n}}$ are Witt equivalent for $n\coloneqq c(U)/8\in\N$ and satisfy $[\Rep(U)]=[\Rep(\mathsf{E}_{8,1}^{\otimes n})]=[\Vect]$; by assumption, they can be connected by a sequence of orbifolds.
\end{proof}

\begin{rem}
The proof above reveals that we do not require the full might of \autoref{conj:wittorb}. We only need to demand that it hold for \voa{}s~$V'$ with $[\Rep(V')]=[\Vect]$. But \autoref{prop:witttrivialorbifoldingVOAs} shows that this is true if $c(V')\leq 24$, at least if we accept the uniqueness of the moonshine module, \autoref{conj:moonshineuniqueness}.
\end{rem}

\smallskip

In the extension to full RCFTs, we shall use the terminology ``quotienting decoupled degrees of freedom'' to refer to the operation $\mathcal{T}_1\otimes \mathcal{T}_2\to \mathcal{T}_1$ of passing from a trivially decoupled tensor product $\mathcal{T}_1\otimes\mathcal{T}_2$ of RCFTs to one of the tensor factors, e.g., $\mathcal{T}_1$. Then the following  holds.
\begin{thmph}\label{thmph:classification}
Assuming \autoref{conj:weakreconstruction} or \autoref{conj:weakreconstruction2} (weak reconstruction) and that \autoref{conj:wittorbRCFT} holds for RCFTs ${_{V'}}\mathcal{H}'_{W'}$ with $\Rep(V')$ being Witt trivial, it follows that any positive RCFT ${_V}\mathcal{H}_W$ can be obtained by starting with $\smash{\mathsf{E}_{8,1}^{\otimes n}\otimes\overline{\mathsf{E}}_{8,1}^{\otimes m}}$ for some large enough $n,m\in\N$, performing a sequence of orbifolds and then quotienting decoupled degrees of freedom.
\end{thmph}
\begin{proof}
From \autoref{conj:weakreconstruction2}, it follows that $V$ and $W$ can be embedded into holomorphic \voa{}s (or chiral CFTs) $A$ and $B$, respectively. Define $V'=\Com_A(V)$ and $W'=\Com_B(W)$, which are both \strat{} \voa{}s. Since $\Rep(V)$ is Witt equivalent to $\Rep(W)$ and since $\Rep(V')\cong \overline{\Rep(V)}$ and $\Rep(W')\cong \overline{\Rep(W)}$, it follows that $\Rep(V')$ is Witt equivalent to $\Rep(W')$. In particular, one can form an RCFT of the form ${_{V'}}\mathcal{H}'_{W'}$. The tensor product ${_V}\mathcal{H}_W\otimes {_{V'}}\mathcal{H}'_{W'}$ can be thought of as an RCFT with respect to the chiral algebras $V\otimes V'$ and $W\otimes W'$. That is, there is some $\mathcal{H}''$ such that
\begin{equation*}
{_V}\mathcal{H}_W\otimes {_{V'}}\mathcal{H}'_{W'}\cong {_{V\otimes V'}}\mathcal{H}''_{W\otimes W'}.
\end{equation*}
Let $n=c(A)/8=c(V\otimes V')/8$ and $m=c(B)/8=c(W\otimes W')/8$ be in $\N$. Then $\smash{\mathsf{E}_{8,1}^{\otimes n}\otimes \overline{\mathsf{E}}_{8,1}^{\otimes m}}$ is Witt equivalent to ${_{V\otimes V'}}\mathcal{H}''_{W\otimes W'}$. Hence, by \autoref{conj:wittorbRCFT}, there is a sequence of orbifolds that takes $\smash{\mathsf{E}_{8,1}^{\otimes n}\otimes \overline{\mathsf{E}}_{8,1}^{\otimes m}}$ to ${_V}\mathcal{H}_W\otimes {_{V'}}\mathcal{H}'_{W'}\cong {_{V\otimes V'}}\mathcal{H}''_{W\otimes W'}$. Then, the target RCFT ${_V}\mathcal{H}_W$ is obtained by quotienting the decoupled degrees of freedom corresponding to ${_{V'}}\mathcal{H}'_{W'}$.
\end{proof}


\subsection{Relationship to Deformation Classes of QFTs}\label{subsec:deformationrelation}

There is a folklore slogan \cite{Sei19} that \emph{anomalies are the only obstruction to connecting QFTs by continuous deformations}. We described in the introduction how a narrow special case of this idea leads to the expectation that any 2d CFT with vanishing gravitational anomaly should admit a boundary condition. The following conjecture is a slightly stronger version of this expectation that applies to RCFTs.
\begin{conjph}\label{conj:RCFTboundaries}
Every RCFT with vanishing gravitational anomaly admits a boundary condition with finite $g$-function.
\end{conjph}

Recall that the $g$-function  of a boundary condition $B$ is defined as the overlap of the corresponding boundary state $|B\rangle$ with the vacuum $|0\rangle$, i.e.\ $g=\langle B|0\rangle$ \cite{AL91}. The main result of this brief subsection is to point out the following:
\begin{thmph}\label{thmph:recharacterizationmainconj}
\autoref{conj:wittinterface}, \autoref{conj:wittinterfaceRCFT} and \autoref{conj:RCFTboundaries} are all equivalent to one another.
\end{thmph}
Informally, the statement that the gravitational anomaly is the only obstruction to an RCFT admitting a boundary condition is the same as the statement that Witt inequivalence is the only obstruction to two RCFTs being related by topological manipulations.

\begin{proof}
It suffices to show the equivalence of \autoref{conj:wittinterface} and \autoref{conj:RCFTboundaries} because the equivalence of \autoref{conj:wittinterface} and \autoref{conj:wittinterfaceRCFT} was already established in equation~\eqref{eqn:conjectureequivalences}.

Assume \autoref{conj:wittinterface}, and suppose we are given an RCFT ${_V}\mathcal{H}_W$ for which we aim to construct a boundary. We learn from \autoref{prop:technicallemma} that there is a simple topological line interface $I$ of finite quantum dimension between the gapless chiral boundaries $V$ and $W$ on which $\mathcal{H}$ terminates, as on the right of \autoref{fig:foldingboundary}. By folding this picture as if closing a book, one obtains a boundary condition $B$ with finite $g$-function of the RCFT ${_V}\mathcal{H}_W$. Thus, \autoref{conj:RCFTboundaries} holds.

Conversely, suppose that any RCFT admits a boundary condition with finite $g$-function, and suppose $V$ and $W$ are two Witt equivalent \strat{} \voa{}s. Choose any topological interface $\mathcal{H}$ between $\Rep(V)$ and $\Rep(W)$ and form the RCFT ${_V}\mathcal{H}_W$. By assumption, this RCFT has a boundary condition. One can inflate this to a $3$-dimensional Kapustin--Saulina picture as on the right of \autoref{fig:foldingboundary} to obtain a topological line interface $I$ between $V$ and $W$ on which $\mathcal{H}$ terminates. This means that $V$ and $W$ are interface equivalent, and hence that \autoref{conj:wittinterface} is true.
\end{proof}

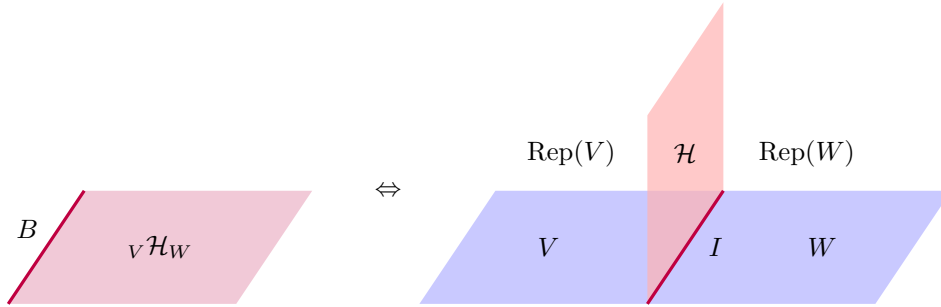
\begin{figure}[ht]
\begin{center}
\begin{tikzpicture}
\filldraw[fill=purple!30,thick,fill opacity=.7,draw=none] (0,-2.5) -- (-3,-1-1.5) -- (-2,-1) -- (1,-1) -- cycle;
\draw[very thick,purple] (-3,-1-1.5)--(-2,-1);
\node[] at (-1,-1.75) {${_V}\mathcal{H}_W$};
\node[] at (2,-1) {$\Leftrightarrow$};
\node[] at (-2.75,-1-.5) {$B$};
\end{tikzpicture}
\begin{tikzpicture}
\filldraw[fill=blue!30,fill opacity=.7,thick,draw=none] (0,-2.5) -- (-6,-1-1.5) -- (-5,-1) -- (1,-1) -- cycle;
\filldraw[fill=red!30,fill opacity=.7,thick,draw=none] (-3,0) -- (1-3,1.5) -- (1-3,-1) -- (-3,-1-1.5) -- cycle;
\draw[very thick,purple] (1-3,-1) -- (-3,-1-1.5);
\node[] at (-.5-.4,-.5) {$\Rep(W)$};
\node[] at (-.5-4+.5,-.5) {$\Rep(V)$};
\node[] at (-2.5,-.5) {$\mathcal{H}$};
\node[] at (-2.5-3+1.2,-1.75) {$V$};
\node[] at (-2.5+3-1.2,-1.75) {$W$};
\node[] at (-2.5+3-1.2-1.4,-1.75) {$I$};
\end{tikzpicture}
\end{center}
\caption{A $3$-dimensional depiction of a boundary condition $B$ of an RCFT ${_V}\mathcal{H}_W$.}\label{fig:foldingboundary}
\end{figure}


\section{Witt Equivalence and Generalized Symmetries}\label{sec:wittgensym}

In this section, we discuss the relationship between Witt equivalence and generalized global symmetries. In \autoref{subsec:symmetrysubalgebraduality}, we relate the problem of classifying \voa{}s within a weak Witt class to the problem of understanding symmetries. In \autoref{subsec:galois}, we refine this idea further by describing how a generalization of quantum Galois theory enters into the picture. In \autoref{subsec:classifyingsymmetries}, we describe an approach to classifying generalized global symmetries in RCFTs. We conclude with examples in \autoref{subsec:symwittex}. In particular, we give a physical argument that all the finite symmetries of the $\mathrm{SU}(2)_1$ Wess--Zumino--Witten model are invertible, and we also find two Fibonacci symmetries of the moonshine module.


\subsection{Symmetry-Subalgebra Duality}\label{subsec:symmetrysubalgebraduality}

Recall that every Witt class $[\mathcal{C}]$ of modular tensor categories contains a unique anisotropic representative $\mathcal{C}_{\mathrm{an}}$. The significance of this privileged representative $\mathcal{C}_{\mathrm{an}}$ for the theory of \voa{}s is highlighted in \autoref{prop:Wittgenusclassification}, which asserts that every positive, \strat{} \voa{} in the Witt class $([\mathcal{C}],c)$ arises as a conformal subalgebra of a \voa{} in the bulk genus $(\mathcal{C}_{\mathrm{an}},c)$.

In this subsection, we provide a re-interpretation of \autoref{prop:Wittgenusclassification} that invokes the language of generalized global symmetries. The connection goes through a concept called symmetry-subalgebra duality (see, e.g., \cite{Bis17,Ray23,JW19,JW20b,JW20a,KLWZZ20,CW23,CJW25,Rie25,DNRX25}), which can be derived using ideas from symmetry TQFTs.

For simplicity, in this work, we restrict our attention to the case that $\mathcal{C}_{\mathrm{an}}=\Vect$, though everything we say can be extended to more general modular tensor categories with suitable modifications. In particular, let $V$ be a \strat{}, holomorphic \voa{}. We claim that any \strat{} conformal subalgebra $W$ defines a symmetry of $V$, which we denote as $\operatorname{Ver}(V/W)$ for ``Verlinde''. The reason for this name is that, as we will see, just as the Verlinde lines \cite{Ver88} of an RCFT are the symmetries of the theory that commute with the chiral algebras, the category $\operatorname{Ver}(V/W)$ by definition consists of those topological lines of $V$ that commute with the subalgebra $W$.

\begin{figure}[ht]
\begin{center}
\begin{tikzpicture}
\filldraw[black,fill=blue!30,fill opacity=.7,  thick,draw=none] (-6,0) -- (-6,-1-1.5) -- (-5,-1) -- (1-6,1.5)    -- cycle;
 \draw[black, thick,dashed](1-3,-1)--(-5,-1);
\filldraw[black,fill=red!30,fill opacity=.7,  thick,draw=none] (-3,0) -- (1-3,1.5) -- (1-3,-1) -- (-3,-1-1.5) -- cycle;
\draw[black,  thick,dashed] (-6,0)--(-3,0);
\draw[black,  thick,dashed] (-3,-2.5) -- (-6,-1-1.5);
\draw[black, thick,dashed] (1-3,1.5)--(1-6,1.5);
\node[] at (-.5-4+.5,-.5) {$\Rep(W)$};
\node[] at (-3,-2.8) {$L(V)$};
\node[] at (-2.5-3.5,-2.8) {$W$};
\tikzstyle{s}=[draw,decorate,<->]
\path[s] (-1.5,-.75) to (2,-.75);
\filldraw[black,fill=purple!30,fill opacity=.7,  thick,draw=none] (-3+6-.5,0) -- (1-3+6-.5,1.5) -- (1-3+6-.5,-1) -- (-3+6-.5,-1-1.5) -- cycle;
\node[] at (-2.5+6-.5-.5,-2.8) {$V$};
\node at (-2.5+6-.5,-.2) [circle,fill,inner sep=1.5pt]{};
\node at (-2.5+6-.5,.1) {$\mathcal{O}(x)$};
\node at (-2.5+6-.5-8.5,-.2) [circle,fill,inner sep=1.5pt]{};
\node at (-2.5+6-.5-8.5,-.5) {$\mathcal{O}(x)$};
\draw[thick,teal] (-2.5+6-.5-.5,-.5)--(-2.5+6-.5-.5+.9,-.5-.7);
\draw[thick,teal] (-2.5+6-.5-.5-5.5,-.5)--(-2.5+6-.5-.5+.9-5.5,-.5-.7);
\node[] at (-2.5+6-.5,-.5-.25-.4){$\mathcal{L}$};
\node[] at (-2.5+6-.5-5.5,-.5-.25+.15){$\mathcal{L}$};
\end{tikzpicture}
\end{center}
\caption{A local operator $\mathcal{O}(x)\in W$ commutes with a topological line $\mathcal{L}\in \operatorname{Ver}(V/W)$.}\label{fig:symTFT}
\end{figure}
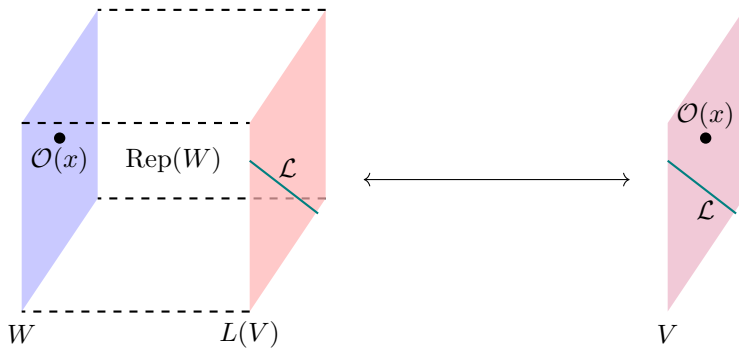

Indeed, note that $V$ defines a Lagrangian algebra $L(V)$ in $\Rep(W)$, and therefore a topological boundary condition in the TQFT $(\Rep(W),c(W))$.\footnote{More pedantically, since the 3d TQFT has non-vanishing central charge, what we call a ``topological boundary condition'' should be more properly thought of as a topological interface to the invertible TQFT with the same central charge.} Consider placing $(\Rep(W),c(W))$ on an interval, with its $W$ boundary imposed on one end and the topological boundary condition $L(V)$ imposed on the other. The topological boundary condition supports a fusion category worth of topological line operators $\operatorname{Ver}(V/W)\coloneqq \Rep(W)_{L(V)}$, which is the category of $L(V)$-modules inside the modular tensor category $\Rep(W)$. On the other hand, by dimensionally reducing along the interval direction, we recover the original theory $V$, and immediately see that $\operatorname{Ver}(V/W)$ can be thought of as furnishing a symmetry of $V$. Moreover, as advertised, this $3$-dimensional perspective reveals that the local operators in $W$ are transparent to the lines in $\operatorname{Ver}(V/W)$, by virtue of the fact that the latter are supported on the topological boundary while the former live on the $W$ boundary (which is often called the ``physical'' boundary in the literature on symmetry TQFTs). See \autoref{fig:symTFT}.

\begin{remph}
We note in passing that $\operatorname{Ver}(V/W)$ not only abstractly captures a symmetry category acting on $V$ but further concretely encodes the twisted sectors of this symmetry. Indeed, each simple object of $\operatorname{Ver}(V/W)$ can be thought of either as a topological line operator $\mathcal{L}$ \emph{or} as a module $M$ of the conformal subalgebra $W$: the latter defines the space of operators that can live at the endpoint of the former, i.e.\ $V_{\mathcal{L}}\cong M$ as $W$-modules, where $V_{\mathcal{L}}$ is the ``$\mathcal{L}$-twisted sector''.
\end{remph}

There is also a clear way to associate a conformal subalgebra with a symmetry. Namely, given a fusion category $\mathcal{F}$ acting on $V$ by symmetries, one can obtain the conformal subalgebra $V^{\mathcal{F}}$ consisting of the operators in $V$ that commute with the lines in $\mathcal{F}$. The bulk TQFT on which $V^{\mathcal{F}}$ occurs as a boundary is $(Z(\mathcal{F}),c(V))$, where $Z(\mathcal{F})$ is the Drinfeld center of $\mathcal{F}$.

The main claim of symmetry-subalgebra duality is that these maps are naturally inverse to one another and hence define a duality,
\begin{equation*}
V^{\operatorname{Ver}(V/W)}=W,\quad\operatorname{Ver}(V/V^{\mathcal{F}})=\mathcal{F}.
\end{equation*}
In other words, every \strat{} conformal subalgebra of a \strat{}, holomorphic \voa{} defines a symmetry, and conversely every symmetry defines a \strat{} conformal subalgebra. See \cite{Bis17} for related ideas in the context of conformal nets.

Symmetry-subalgebra duality then allows us to re-interpret \autoref{prop:Wittgenusclassification} as follows:
\begin{propph}\label{prop:Wittgenusfromsymmetries}
Every positive, \strat{} \voa{} in the Witt class $([\Vect],c)$ arises as $V^{\mathcal{F}}$ for some \strat{}, holomorphic \voa{} $V$ with central charge $c$ and some fusion category $\mathcal{F}$ acting on $V$ by symmetries.
\end{propph}
Thus, the classification of positive, \strat{} \voa{}s in the Witt class $([\Vect],c)$ amounts to the enumeration of \strat{}, holomorphic \voa{}s of central charge $c$, as well as their finite (generally non-invertible) global symmetries.

Finally, we note that \autoref{prop:Wittgenusfromsymmetries} can be generalized to Witt classes $([\mathcal{C}],c)$ with $[\mathcal{C}]$ different from $[\Vect]$, and also made mathematically rigorous, using the results of \cite{Rie25}. To wit, one has the following:
\begin{prop}\label{prop:hypwitt}
Let $\mathcal{C}_{\mathrm{an}}$ be a completely anisotropic, pseudo-unitary modular tensor category. Every positive, \strat{} \voa{} $W$ in the Witt class $([\mathcal{C}_{\mathrm{an}}],c)$ conformally embeds into a \strat{} \voa{} $V$ in the bulk genus $(\mathcal{C}_{\mathrm{an}},c)$, and can be recovered as the fixed points $W=V^K$ of an action of a hypergroup $K$ on $V$ by $W$-module maps.
\end{prop}

As we have alluded to (see \autoref{rem:hyporb}), the limitation of \autoref{prop:hypwitt} is that there is not yet an intrinsic definition of what it means for a hypergroup to act on a \voa{} that does not make reference to a conformal subalgebra. Therefore, \autoref{prop:hypwitt} cannot yet be turned into a mathematically rigorous classification program. More on the physical interpretation of hypergroup actions on vertex operator algebras will be presented in \cite{GR26}.


\subsection{Quantum Galois Theory}\label{subsec:galois}

One can actually say more about the structure of Witt classes of the form $([\Vect],c)$. Not only do all positive \voa{}s in such Witt classes conformally embed into a holomorphic \voa{} (see \autoref{prop:Wittgenusclassification}), but further, the Hasse diagrams formed by these embeddings are controlled by the symmetries of the holomorphic \voa{}s.

Suppose one fixes a \strat{}, holomorphic \voa{} $V$ and a positive, \strat{} conformal subalgebra $W\subset V$. The intermediate \strat{} \voa{}s $U$ satisfying $W\subset U\subset V$ are in bijection with condensable algebras $A=A(U)$ in $\Rep(W)$ that are contained in the maximal condensable algebra corresponding to the conformal extension $V$.

On the other hand, symmetry-subalgebra duality provides a fusion category symmetry $\mathcal{F}$ of $V$ satisfying $\Rep(W)=Z(\mathcal{F})$. (A mathematical shadow of this is an action of the hypergroup associated with $\mathcal{F}$ by $W$-module maps on $V$ \cite{Rie25}.) By Theorem~4.10 of \cite{DMNO13}, the condensable algebras of $Z(\mathcal{F})$, and thus the simple (and hence strongly rational) intermediate \voa{}s $U$ by \autoref{prop:algext}, are in bijection with fusion subcategories $\mathcal{G}\subset\mathcal{F}$. The explicit map can be described via symmetry-subalgebra duality (and made mathematically rigorous in terms of hypergroups using \cite{Rie25}) as
\begin{equation}\label{eqn:galoisconnection}
U\mapsto \operatorname{Ver}(V/U),\quad\mathcal{G}\mapsto V^{\mathcal{G}}.
\end{equation}

Moreover, following \cite{DJX13}, let us define the degree of a conformal extension of $U$ to a holomorphic \voa{} $V$ as
\begin{equation*}
[V:U]\coloneqq\lim_{\tau\to 0}\frac{\ch_V(\tau)}{\ch_U(\tau)}.
\end{equation*}
By equation~(4.1) of \cite{DJX13}, we can equivalently define $[V:U]=\operatorname{FPdim}(L(V))$, the Frobenius--Perron dimension of the Lagrangian algebra $L(V)$ in $\Rep(U)$ defined by the embedding $U\subset V$. Then by Theorem~4.10 of \cite{DMNO13}, we find a relationship between the degree of the extension $U\subset V$ and the dimension of the category of symmetries of $V$ preserved by $U$, namely
\begin{equation*}
[V:U]=\operatorname{FPdim}(\operatorname{Ver}(V/U)).
\end{equation*}

Putting this all together, we find the following, which can be thought of as a kind of ``non-invertible'' generalization of the quantum Galois theory put forward in \cite{DM97,DLM96,HMT99,DJX13}:
\begin{prop}[Quantum Galois Theory for Chiral CFTs]\label{prop:galois}
Let $V$ be a \strat{}, holomorphic \voa{} and $W$ a positive, \strat{} \vosa{} of $V$. Equation~\eqref{eqn:galoisconnection} defines an antitone Galois connection between simple \vosa{}s $U$ of $V$ containing $W$, and fusion subcategories $\mathcal{G}$ of $\operatorname{Ver}(V/W)$. Moreover, if $U$ is associated with $\mathcal{G}$ via this correspondence, then $[V:U]=\operatorname{FPdim}(\mathcal{G})$.
\end{prop}

In \cite{DNRX25,GR26}, this Galois correspondence is generalized from holomorphic \voa{}s $V$ to arbitrary strongly rational \voa{}s.


\subsection{Classifying Generalized Symmetries of RCFTs}\label{subsec:classifyingsymmetries}

Before moving on, we explain how information about the classification of RCFTs within a Witt class can often be translated into information about the symmetries of the individual members of the Witt class.

Consider an RCFT ${_V}\mathcal{H}_W$ that admits an action of a fusion category $\mathcal{F}$. Further assume that the chiral algebras $V$ and $W$ are taken to be maximal. The action of $\mathcal{F}$ on the full Hilbert space $\mathcal{H}$ then restricts to an action of $\mathcal{F}$ on the left-moving chiral algebra $V$, and also to an action on the right-moving chiral algebra $W$. In particular, we may form the subalgebra $X\coloneqq V^{\mathcal{F}}\subset V$ of operators in $V$ that commute with the topological lines in $\mathcal{F}$, and also the subalgebra $Y\coloneqq W^{\mathcal{F}}\subset W$. These subalgebras are expected to be \strat{}.

The total Hilbert space $\mathcal{H}$ of the CFT, which defines a Lagrangian algebra of $\Rep(V)\boxtimes\smash{\overline{\Rep(W)}}$, can also be thought of as a Lagrangian algebra over $\Rep(X)\boxtimes\overline{\Rep(Y)}$.\footnote{Indeed, because $X$ and $Y$ are conformal subalgebras of $V$ and $W$, respectively, it follows that $\Rep(V)\boxtimes\smash{\overline{\Rep(W)}}$ can be obtained via anyon condensation from $\Rep(V^{\mathcal{F}})\boxtimes\smash{\overline{\Rep(W^{\mathcal{F}})}}$. By Proposition~2.3.2 of \cite{Dav10}, a condensable algebra in a modular tensor category obtained by anyon condensation can be thought of as a condensable algebra in the modular tensor category before anyon condensation. In particular, $\mathcal{H}$ may be thought of as a Lagrangian algebra in $\Rep(V^{\mathcal{F}})\boxtimes\smash{\overline{\Rep(W^{\mathcal{F}})}}$.} More physically, we may think of the RCFT  ${_V}\mathcal{H}_W$ as being built on top of the smaller chiral algebras $X$ and $Y$, i.e.\ ${_V}\mathcal{H}_W\cong {_X}\mathcal{H}_Y$.

Recall that the category of Verlinde lines of an RCFT is by definition the category of topological line operators of the theory that commute with the left- and right-moving chiral algebras. In the present situation, the category of Verlinde lines of ${_X}\mathcal{H}_Y$ is
\begin{equation*}
\operatorname{Ver}(\mathcal{H}\,/\,X,\!Y)\coloneqq\big(\Rep(X)\boxtimes\overline{\Rep(Y)}\big)_{\mathcal{H}},
\end{equation*}
i.e.\ the category of $\mathcal{H}$-modules inside $\Rep(X)\boxtimes\overline{\Rep(Y)}$. By construction, the symmetry category $\mathcal{F}$ commutes with the chiral algebras $X=V^{\mathcal{F}}$ and $Y=W^{\mathcal{F}}$, and therefore, it must embed as a subcategory of $\operatorname{Ver}(\mathcal{H}\,/\,X,\!Y)$.

Now, suppose we know the Witt classes of $V$ and $W$. In particular, this means that we have a list of \strat{} conformal subalgebras of $V$ and $W$; therefore, we have a menu of putative options for what the fixed point subalgebras $X=V^{\mathcal{F}}$ and $Y=W^{\mathcal{F}}$ could be, and hence, a menu of Verlinde categories into which an arbitrary $\mathcal{F}$ must be able to embed. Thus, we have the following expectation:
\begin{propph}\label{prop:rcftsymmetries}
Let ${_V}\mathcal{H}_W$ be an RCFT. Any fusion category $\mathcal{F}$ that acts on ${_V}\mathcal{H}_W$ by symmetries occurs as a fusion subcategory of a Verlinde category of the form $\operatorname{Ver}(\mathcal{H}\,/\,X,\!Y)$ for some \strat{} conformal subalgebras $X\subset V$ and $Y\subset W$.
\end{propph}
We will see in \autoref{exph:SU(2)1symmetries} (see also \cite{GR26}) that this kind of reasoning can be very effectively applied to $c=1$ (unitary) RCFTs to classify their global symmetries.

Finally, we note that, because a \strat{} \voa{} with $c\geq 1$ is often expected to have infinitely many \strat{} conformal subalgebras, \autoref{prop:rcftsymmetries} sets the complementary expectation that an RCFT with $c_L,c_R\geq 1$ should often have infinitely many global symmetries.


\subsection{Examples}\label{subsec:symwittex}

We next provide some examples illustrating the connection between Witt classes and global symmetries.

\begin{ex}[A Simple Example of Symmetry-Subalgebra Duality]
In this example, we describe the duality between a $\mathbb{Z}_2$ symmetry and a conformal subalgebra of $\mathsf{E}_{8,1}$.

Consider the current algebra $\mathsf{D}_{8,1}$. By virtue of the fact that $D_8=\mathfrak{so}_{16}$ is a Lie subalgebra of $E_8$, there is a conformal embedding of $\mathsf{D}_{8,1}$ into $\mathsf{E}_{8,1}$. One can associate a symmetry with this subalgebra as follows. By restriction, $\mathsf{E}_{8,1}$ decomposes into irreducible modules of $\mathsf{D}_{8,1}$ as
\begin{equation}\label{eqn:E8decomp}
\mathsf{E}_{8,1}\cong \mathsf{D}_{8,1}\oplus \mathsf{D}_{8,1}(s),
\end{equation}
where $\mathsf{D}_{8,1}(s)$ is the spinor representation of $\mathsf{D}_{8,1}$. This decomposition naturally defines a $\Z_2$ symmetry of $\mathsf{E}_{8,1}$, which acts as $+1$ on states in $\mathsf{D}_{8,1}$ and as $-1$ on states in $\mathsf{D}_{8,1}(s)$, and we claim that $\operatorname{Ver}(\mathsf{E}_{8,1}/\mathsf{D}_{8,1})\cong \Vect_{\mathbb{Z}_2}$ describes precisely this $\Z_2$ symmetry.

In the other direction, one may associate a conformal subalgebra with this $\mathbb{Z}_2$ symmetry by observing that the fixed-point subalgebra $\mathsf{E}_{8,1}^{\Z_2}$ is isomorphic to $\mathsf{D}_{8,1}$.

From the $3$-dimensional perspective afforded by \autoref{fig:symTFT}, $\mathsf{D}_{8,1}$ has the toric code (i.e.\ $\Z_2$-gauge theory) as its bulk TQFT. Labeling the anyons of the toric code as $\mathbf{1},e,m,f$, the following identification with modules of $\mathsf{D}_{8,1}$ holds,
\begin{equation*}
\mathbf{1}\leftrightarrow\mathsf{D}_{8,1}, \ \ e\leftrightarrow\mathsf{D}_{8,1}(s), \ \ m \leftrightarrow \mathsf{D}_{8,1}(c), \ \ f\leftrightarrow \mathsf{D}_{8,1}(v),
\end{equation*}
where $\mathsf{D}_{8,1}(c)$ and $\mathsf{D}_{8,1}(v)$ are the conjugate spinor and vector representations, respectively. The decomposition in equation~\eqref{eqn:E8decomp} can thus be interpreted as defining the Lagrangian algebra $L(\mathsf{E}_{8,1})=\mathbf{1}\oplus e$ of the toric code. This Lagrangian algebra corresponds to the ``electric'' boundary condition of the toric code, which is known to support (non-anomalous) $\Z_2$-symmetry lines on its worldvolume; this furnishes another way to see that $\operatorname{Ver}(\mathsf{E}_{8,1}/\mathsf{D}_{8,1}) = \Vect_{\Z_2}$. If one compactifies the toric code on an interval with the topological boundary condition defined by $\mathbf{1}\oplus e$ on one end and the $\mathsf{D}_{8,1}$ chiral gapless boundary condition on the other, then one obtains the theory $\mathsf{E}_{8,1}$.
\end{ex}

\begin{ex}[Generalized Symmetries of Chiral $\mathsf{E}_{8,1}$]
Generalizing the previous example, consider the problem of classifying all positive, \strat{} \voa{}s in the Witt class $([\Vect],8)$. This is equivalent to classifying positive, \strat{} \voa{}s in all bulk genera of the form $(Z(\mathcal{F}),8)$, where $\mathcal{F}$ is a spherical fusion category. By the logic of \autoref{prop:Wittgenusfromsymmetries} (or \autoref{prop:hypwitt}), noting that $\mathsf{E}_{8,1}$ is the unique \strat{}, holomorphic \voa{} with central charge $c=8$, this problem in turn is essentially equivalent to classifying all actions of spherical fusion categories $\mathcal{F}$ on $\mathsf{E}_{8,1}$ and enumerating the corresponding fixed-point subalgebras.

It is natural to organize this problem by the rank of $\mathcal{F}$. For example, it is solved when $\rk(\mathcal{F})=2$ \cite{Ray23}, where it is known that the positive\footnote{In fact, by the results of \cite{MR24a}, there are no non-positive \voa{}s in these bulk genera.} \voa{}s in the corresponding bulk genera are
\begin{equation*}
(Z(\Vect_{\Z_2}),8) = \{ \mathsf{D}_{8,1} \},\;(Z(\Vect_{\Z_2}^\omega),8) = \{ \mathsf{A}_{1,1}\mathsf{E}_{7,1}\},\;(Z(\operatorname{Fib}),8) = \{\mathsf{G}_{2,1}\mathsf{F}_{4,1}\}.
\end{equation*}
Many more examples are known; for example, see \cite{BKN22} for work on the case where $\mathcal{F}$ is a Tambara--Yamagami category \cite{TY98}.

There has been ongoing work to identify a \strat{}, $c=8$ \voa{} $V$ that has the Drinfeld center $Z(\operatorname{Hg})$ of the Haagerup fusion category $\operatorname{Hg}$ as its representation category, i.e.\ to show that the symbol $(\mathcal{C},c)=(Z(\operatorname{Hg}),8)$ corresponds to a valid bulk genus, and to construct a \voa{} representative of this genus. For example, it is known that if such a theory exists, then its character vector $\smash{\ch_{V(\alpha)}(\tau)=\tr_{V(\alpha)}q^{L_0-1/3}}$ must be one of two options, calculated in \cite{EG11}. From \autoref{prop:Wittgenusfromsymmetries} (or \autoref{prop:hypwitt}), we learn that the existence of such a \voa{} $V$ would imply that $\mathsf{E}_{8,1}$ has a $\operatorname{Hg}$ symmetry satisfying $\smash{\mathsf{E}_{8,1}^{\operatorname{Hg}}}\cong V$.
\end{ex}

\begin{ex}[Fibonacci Symmetry of the Monster CFT and Quantum Galois Theory]\label{exph:fibmonster}
As another example, we use symmetry-subalgebra duality and quantum Galois theory to study non-invertible symmetries of the moonshine module $V^\natural$ (or monster CFT). This chiral CFT is celebrated for its invertible symmetries, which furnish the simple sporadic monster group, $\smash{\Aut(V^\natural)}\cong \mathbb{M}$. However, the full category of topological line defects contains (infinitely) many non-invertible symmetries as well. See \cite{LS21,Vol24} for examples of Tambara--Yamagami categories \cite{TY98} acting on $V^\natural$.

Here, we investigate actions of the unitary Fibonacci category $\operatorname{Fib}$ on $V^\natural$.\footnote{While this work was being completed, the paper \cite{FH24} appeared on the arXiv, which has overlapping content.} By symmetry-subalgebra duality, if such a symmetry exists, the operators of $V^\natural$ that are transparent to the Fibonacci line participate in a positive, \strat{} conformal subalgebra $W\coloneqq \smash{(V^\natural)^{\operatorname{Fib}}}$ with $\Rep(W)=Z(\operatorname{Fib})\cong (G_2,1)\boxtimes (F_4,1)$. Conversely, we may prove the existence of a Fibonacci action by finding such a conformal subalgebra $W$. (At least the latter statement is mathematically rigorous using \cite{Rie25}.)  It follows from the results of \cite{Ray23} (see Section~E.4.18) that there are only two possibilities for the character vector $\Ch_W$ of such a putative \voa{} $W$ if it is to satisfy positivity (and hence, only two possibilities for the Fibonacci twisted sector partition functions of $V^\natural$):
\begin{equation}\label{eqn:firstoption}
\Ch_W(\tau) =
\begin{pmatrix*}[l]
q^{-1}+59045 q+6092960 q^2+241395800 q^3 +\dots \\
8671 q^{3/5}+1804672 q^{8/5}+97010251 q^{13/5} +\dots \\
q^{-3/5}+1568 q^{2/5}+672918 q^{7/5}+45506688 q^{12/5}+\dots \\
137839 q+15400800 q^2+622904170 q^3+\dots \\
\end{pmatrix*}
\end{equation}
or
\begin{equation}\label{eqn:secondoption}
\Ch_W(\tau) =
\begin{pmatrix*}[l]
q^{-1}+55920 q+5961710 q^2+239095800 q^3 +\dots \\
2 q^{-2/5}+7299 q^{3/5}+1731552 q^{8/5}+95550116 q^{13/5}+\dots \\
1938 q^{2/5}+693294 q^{7/5}+45905184 q^{12/5} +\dots \\
140964 q+15532050 q^2+625204170 q^3+14648672160 q^4+\dots \\
\end{pmatrix*}.
\end{equation}
We shall show that both options are realized in the monster CFT, and hence that there are at least two Fibonacci topological line defects in $V^\natural$ that cannot be conjugated into each other using invertible symmetries.

\begin{prop}
There are at least two positive, strongly rational conformal subalgebras $W$ of the moonshine module $V^\natural$ with $\Rep(W)\cong Z(\operatorname{Fib})$, up to isomorphism. One subalgebra realizes the character vector in \eqref{eqn:firstoption}, the other the one in \eqref{eqn:secondoption}.
\end{prop}

It would be interesting to prove that these two are the only two conformal subalgebras of $V^\natural$ with $\Rep(W)\cong Z(\operatorname{Fib})$ (equivalently, the only two Fibonacci symmetries of $V^\natural$) up to conjugation.

\begin{proof}
\emph{Option 1}. First, let us describe how to construct a Fibonacci action of the moonshine module corresponding to \eqref{eqn:firstoption}. We mention that this \voa{} $W\subset V^\natural$ with $\Rep(W)\cong Z(\operatorname{Fib})$ and character \eqref{eqn:firstoption} has been independently discovered in unpublished work \cite{CLY24}. There, it appears as an intermediate step in a ``non-classical'' construction of the moonshine module $V^\natural$.

In \cite{HLY12}, it was shown that every element of the 3A conjugacy class of $\mathbb{M}$ fixes a $\Z_3$-parafermion \voa{} $\mathcal{P}(3)\cong L(\sfrac45)\oplus L(\sfrac45,3)\subset \smash{V^\natural}$ of central charge $c_{5,6}=4/5$ (see \cite[\autoref*{MR1ex:para}]{MR24a} for the definition). One can show that $\Rep(\mathcal{P}(3))\cong (E_6,1)\boxtimes (G_2,1)$ as a modular tensor category. Indeed, Tables~3 and~4 of \cite{NRWW23} imply that a modular tensor category with $c=4/5\pmod{8}$ and rank $6$ has modular data which agree with those of $(E_6,1)\boxtimes(G_2,1)$. Since $\mathcal{P}(3)$ has central charge $c=4/5$ and possesses six irreducible modules \cite{KMY00}, it follows that the modular data of $\Rep(\mathcal{P}(3))$ and $(E_6,1)\boxtimes (G_2,1)$ match. Since the modular data factorize, Verlinde's formula \cite{Ver88,Hua08} coupled with Theorem~4.2 in \cite{Mue03b} then implies that $\Rep(\mathcal{P}(3))\cong\mathcal{C}\boxtimes \mathcal{D}$ with $\mathcal{C}$ a modular tensor category agreeing with $(E_6,1)$ at the level of modular data and $\mathcal{D}$ a modular tensor category agreeing with $(G_2,1)$ at the level of modular data. Then the classification of modular tensor categories with rank at most 4 \cite{RSW09} shows that $\mathcal{C}$ must be $(E_6,1)$ and $\mathcal{D}$ must be $(G_2,1)$.

The commutant $\smash{\textsl{VF}_{24}^\natural=V^\natural/\mathcal{P}(3)}$ was studied in \cite{HLY12} and shown to possess an action of the Fischer sporadic group $\operatorname{Fi}_{24}$. Below, we prove that $\smash{\textsl{VF}_{24}^\natural}$ is \strat{}. It follows from, e.g., \cite{FFRS06} that $\Rep(\textsl{VF}_{24}^\natural) \cong\overline{\Rep(\mathcal{P}(3))}$. Therefore, the tensor product $\smash{\mathcal{P}(3)\otimes \textsl{VF}_{24}^\natural}$ occurs as a \strat{} conformal subalgebra of $\smash{V^\natural}$ with representation category given by
\begin{equation*}
\smash{\Rep(\mathcal{P}(3)\otimes\textsl{VF}_{24}^\natural)\cong \Rep(\mathcal{P}(3))\boxtimes\overline{\Rep(\mathcal{P}(3))}\cong  Z(\Vect_{\Z_3}\boxtimes\operatorname{Fib})},
\end{equation*}
and $\smash{V^\natural}$ is obtained from $\smash{\mathcal{P}(3)\otimes\textsl{VF}_{24}^\natural}$ by condensing the canonical Lagrangian algebra of $Z(\Vect_{\mathbb{Z}_3}\boxtimes \operatorname{Fib})$. So, the subalgebra $\smash{\mathcal{P}(3)\otimes \textsl{VF}_{24}^\natural\subset V^\natural}$ defines a symmetry
\begin{equation*}
\smash{\operatorname{Ver}(V^\natural\,/\,\mathcal{P}(3)\otimes\textsl{VF}_{24}^\natural)}=\Vect_{\Z_3}\boxtimes \operatorname{Fib}
\end{equation*}
of $\smash{V^\natural}$, with the $\Vect_{\Z_3}$ living in the 3A conjugacy class of $\mathbb{M}$. This symmetry can be rigorously formulated as a hypergroup action in the sense of \cite{Rie25,DNRX25}.

Let us study the predictions of quantum Galois theory. \autoref{prop:galois} says that the \vosa{}s of $\smash{V^\natural}$ containing $\smash{\mathcal{P}(3)\otimes \textsl{VF}_{24}^\natural}$ are in bijection with fusion subcategories of $\smash{\Vect_{\Z_3}}\boxtimes\operatorname{Fib}$. Besides $\smash{\Vect_{\Z_3}}\boxtimes\operatorname{Fib}$ and $\Vect$, which correspond to $\smash{\mathcal{P}(3)\otimes \textsl{VF}_{24}^\natural}$ and $\smash{V^\natural}$, respectively, there are the fusion subcategories $\Vect_{\Z_3}$ and $\operatorname{Fib}$. We call the corresponding \vosa{}s predicted by \autoref{prop:galois} $U$ and $W$, respectively. There are thus the following Hasse diagrams:
\begin{equation*}
\begin{tikzcd}
& V^\natural \\
U\arrow[hookrightarrow]{ur}&& W\arrow[hookrightarrow]{ul}\\
& \mathcal{P}(3)\otimes \textsl{VF}_{24}^\natural \arrow[hookrightarrow]{ul}\arrow[hookrightarrow]{ur}
\end{tikzcd}\hspace{7pt}\leftrightarrow\hspace{7pt}
\begin{tikzcd}
& \Vect\arrow[hookrightarrow]{dl}\arrow[hookrightarrow]{dr} \\
\Vect_{\Z_3}\arrow[hookrightarrow]{dr}&& \operatorname{Fib}\arrow[hookrightarrow]{dl}\\
& \Vect_{\Z_3}\boxtimes \operatorname{Fib}
\end{tikzcd}
\end{equation*}

Because $U$ and $W$ contain $\smash{\mathcal{P}(3)\otimes\textsl{VF}_{24}^\natural}$ as a conformal subalgebra, we may decompose them into modules for it. For example, $U$ decomposes as
\begin{equation*}
U\cong \mathcal{P}(3)\otimes \textsl{VF}_{24}^\natural \oplus \mathcal{P}(3,[2,0])\otimes \textsl{VF}_{24}^\natural(2,0),
\end{equation*}
where we have used the notation of \cite{BHLLR21} to label the modules of $\mathcal{P}(3)$ and $\smash{\textsl{VF}_{24}^\natural}$. Let us do the same for $W$ in more detail. We may obtain it and its irreducible modules as simple-current extensions of $\smash{\mathcal{P}(3)\otimes\textsl{VF}_{24}^\natural}$,
\begin{equation}\label{eqn:Wmodules}
\begin{split}
W_{(\mathbf{1},\mathbf{1})} &= \mathcal{P}(3)  \textsl{VF}_{24}^\natural \oplus \mathcal{P}(3,[3,1])  \textsl{VF}_{24}^\natural(3,1)\oplus \mathcal{P}(3,[3,-1])  \textsl{VF}_{24}^\natural(3,-1),\\
W_{(\mathbf{1},\tilde\tau)} &= \mathcal{P}(3)  \textsl{VF}_{24}^\natural(2,0)\oplus \mathcal{P}(3,[3,1])  \textsl{VF}_{24}^\natural(1,1)\oplus \mathcal{P}(3,[3,-1])  \textsl{VF}_{24}^\natural(2,2),\\
W_{(\tau,\mathbf{1})} &= \mathcal{P}(3,[2,0])   \textsl{VF}_{24}^\natural \oplus \mathcal{P}(3,[1,1])  \textsl{VF}_{24}^\natural(3,1) \oplus \mathcal{P}(3,[2,2])  \textsl{VF}_{24}^\natural(3,-1),\\
W_{(\tau,\tilde\tau)} &= \mathcal{P}(3,[2,0])\textsl{VF}_{24}^\natural(2,0) \!\oplus\! \mathcal{P}(3,[1,1])\textsl{VF}_{24}^\natural(1,1) \!\oplus\! \mathcal{P}(3,[2,2]) \textsl{VF}_{24}^\natural(2,2),
\end{split}
\end{equation}
where to avoid clutter we have suppressed the symbol $\otimes$. Here, we are labeling the irreducible modules of $W=W_{(\mathbf{1},\mathbf{1})}$ by the simple objects $(\mathbf{1},\mathbf{1})$, $(\mathbf{1},\tilde\tau)$, $(\tau,\mathbf{1})$ and $(\tau,\tilde\tau)$ of $\Rep(W)\cong Z(\operatorname{Fib})\cong (G_2,1)\boxtimes (F_4,1)$, where $\tau$ and $\tilde\tau$ are the non-trivial simple objects of $(G_2,1)$ and $(F_4,1)$, respectively.
One may indeed check that the character vector of this $W$ matches equation~\eqref{eqn:firstoption} by using equations (3.20) and (3.25) of \cite{BHLLR21}.\footnote{We note a typo in the transcription of the parafermion characters in version~$2$ of \cite{BHLLR21} on the arXiv and in the published version. The correct characters are $\psi^{(3)}_{1,1}(\tau) = \psi^{(3)}_{2,2}(\tau) = q^{1/30}(1+q+ 2q^2 +3q^3 +5q^4+ 7q^5+\dots)$ and $\psi_{2,0}^{(3)}(\tau) = q^{11/30}(1+2q+2q^2+4q^3+5q^4+8q^5+\dots)$.}

\smallskip

\noindent\emph{Proof of Strong Rationality of $\smash{\textsl{VF}_{24}^\natural}$.} We now give a proof that the Fischer \voa{} $\smash{\textsl{VF}_{24}^\natural}$ is \strat{}. We thank the referee for raising the question, which prompted us to look into this, and also Hiroshi Yamauchi for useful discussions.

First, we argue that $\textsl{VF}_{24}^\natural$ is a conformal extension of a \strat{} \voa{}. To this end, we observe from equation~(1.1) in \cite{DLMN98} that the moonshine module $V^\natural$ possesses a conformal subalgebra of the form
\begin{equation*}
\left(L(\sfrac12)\otimes L(\sfrac7{10})\otimes L(\sfrac45)\right)^{\otimes 12}\subset V^\natural
\end{equation*}
with $c_{3,4}+c_{4,5}+c_{5,6}=1/2+7/10+4/5=2$. In fact, using the results of \cite{KMY00}, one can make the more refined statement that there is a conformal subalgebra of the form
\begin{equation*}
(L(\sfrac12)\otimes L(\sfrac7{10})\otimes \mathcal{P}(3))^{\otimes 12}\subset V^\natural
\end{equation*}
where $\mathcal{P}(3)\cong L(\sfrac45)\oplus L(\sfrac45,3)$ is the unitary, \strat{} $\mathbb{Z}_3$-parafermion \voa{}. Since $\smash{\textsl{VF}_{24}^\natural}$ is defined to be the commutant of $\mathcal{P}(3)$ in $V^\natural$ and since $(L(\sfrac12)\otimes L(\sfrac7{10}))^{\otimes 12}\otimes\mathcal{P}(3)^{\otimes 11}$ is a \strat{} vertex operator subalgebra of the commutant of $\mathcal{P}(3)$, it follows that $\smash{\textsl{VF}_{24}^\natural}$ is a conformal extension of a \strat{} \voa{}.

Thus, from \autoref{prop:CMSY}, it will follow that $\smash{\textsl{VF}_{24}^\natural}$ is itself \strat{} if we can show that it is simple and $\mathbb{Z}_{\geq 0}$-graded. First, it is $\mathbb{Z}_{\geq 0}$-graded, and even of CFT-type, simply because $V^\natural$ is of CFT-type and because the action of $L_0$ on $\smash{\textsl{VF}_{24}^\natural}$ is the restriction of the $L_0$-action on $V^\natural$ to the subalgebra $\smash{\textsl{VF}_{24}^\natural\subset V^\natural}$. Next, we note that it is unitary by Example~5.27 in \cite{CKLW18} since it is the commutant of a unitary \voa{} by another unitary \voa{}. Since it is unitary and of CFT-type, Proposition~5.3 in \cite{CKLW18} implies that it is simple. This concludes the proof.

\smallskip

\noindent\emph{Option 2}. It is also possible to realize a Fibonacci symmetry of $V^\natural$ whose corresponding positive, \strat{} conformal subalgebra has the character vector given by \eqref{eqn:secondoption}.

Let $U_{\mathrm{2A}}$ be the simple-current extension of the simple, unitary Virasoro \voa{}s $L(\sfrac12)\otimes L(\sfrac{7}{10})$ with central charge $c_{3,4}+c_{4,5}=1/2+7/10=6/5$ (cf.\ \autoref{ex:discretewitt}) that decomposes as
\begin{equation*}
U_{\mathrm{2A}}\cong L(\sfrac12)\otimes L(\sfrac{7}{10})\oplus L(\sfrac12,\sfrac12)\otimes L(\sfrac{7}{10},\sfrac32).
\end{equation*}
The representation theory of this \strat{} \voa{} was studied in \cite{LY00}, and one can infer from op.\ cit.\ that
\begin{equation*}
\Rep(U_{\mathrm{2A}})\cong (F_4,1)\boxtimes (D_4,1).
\end{equation*}
As a fusion category, i.e.\ forgetting the braiding and twist, one can verify that $\Rep(U_{\mathrm{2A}})\cong \operatorname{Fib} \boxtimes \Vect_{\mathbb{Z}_2\times\mathbb{Z}_2}$.

One can show that $U_{\mathrm{2A}}$ embeds primitively into $V^\natural$, i.e.\ $U_{\mathrm{2A}}\cong U_{\mathrm{2A}}''$ where $U_{\mathrm{2A}}''\coloneqq\Com_{V^\natural}(\Com_{V^\natural}(U_{\mathrm{2A}}))\supset U_{\mathrm{2A}}$. Indeed, by Example~5.27 in \cite{CKLW18}, $U_{\mathrm{2A}}''$ is a unitary \voa{}. Furthermore, $(U_{\mathrm{2A}}'')_0=\C\vac$ because $V^\natural$ has a unique vacuum and the $L_0$-action on $U_{\mathrm{2A}}''$ is the restriction of the $L_0$-action on $V^\natural$ to $U_{\mathrm{2A}}''\subset V^\natural$. Thus, from Proposition~5.3 in \cite{CKLW18}, it follows that $U_{\mathrm{2A}}''$ is simple. Simple conformal extensions of $U_{\mathrm{2A}}$ are controlled by condensable algebras of $\Rep(U_{\mathrm{2A}})$ by \autoref{prop:algext}, but $\Rep(U_{\mathrm{2A}})$ is completely anisotropic because $U_{\mathrm{2A}}$ does not possess any non-trivial module with integer $L_0$-weights. One concludes that $U_{\mathrm{2A}}''=U_{\mathrm{2A}}$.

It follows from \cite{FFRS06} that the commutant $U'_{\mathrm{2A}}=\Com_{V^\natural}(U_{\mathrm{2A}})$ has its representation category given by $\smash{\Rep(U'_{\mathrm{2A}})\cong\overline{\Rep(U_{\mathrm{2A}})}}$. In equation~(3.55) of \cite{BHLLR21}, the character vector of $U'_{\mathrm{2A}}$ was computed, and it was argued that $U'_{\mathrm{2A}}$ admits an action of ${^2}\!E_6(2).2$ by symmetries, with the Steinberg subgroup ${^2}\!E_6(2)$ acting by inner automorphisms.

Now, the \voa{} $U_{\mathrm{2A}}\otimes U'_{\mathrm{2A}}$ sits as a dual pair inside $V^\natural$ and induces an action of $\operatorname{Ver}(V^\natural\,/\,U_{\mathrm{2A}}\otimes U'_{\mathrm{2A}})\cong \operatorname{Fib}\boxtimes \Vect_{\mathbb{Z}_2\times\mathbb{Z}_2}$ by symmetries.\footnote{It is known that each of the three non-identity involutions in the factor $\Vect_{\mathbb{Z}_2\times\mathbb{Z}_2}$ lives in the 2A conjugacy class of $\Aut(V^\natural)\cong\mathbb{M}$, where we follow the conventions of \cite{CCNPW85} in labeling conjugacy classes of the monster group.} Galois theory then predicts that intermediate subalgebras sitting in between $U_{\mathrm{2A}}\otimes U'_{\mathrm{2A}}$ and $V^\natural$ are in bijective correspondence with fusion subcategories of $\operatorname{Ver}(V^\natural\,/\,U_{\mathrm{2A}}\otimes U'_{\mathrm{2A}})$. In particular, one such fusion subcategory is $\operatorname{Fib}$, and we claim that the corresponding \vosa{}, $W=(V^\natural)^{\operatorname{Fib}}$, is a simple-current extension of $U_{\mathrm{2A}}\otimes U'_{\mathrm{2A}}$, with modules decomposing as
\begin{equation}\label{eqn:secondFibsubalgebra}
\begin{split}
W_{(\mathbf{1},\mathbf{1})}&= U_{\mathrm{2A}} U'_{\mathrm{2A}}\oplus U_{\mathrm{2A}}(2)U'_{\mathrm{2A}}(2)\oplus U_{\mathrm{2A}}(3)U'_{\mathrm{2A}}(3) \oplus U_{\mathrm{2A}}(4)U'_{\mathrm{2A}}(4),\\
W_{(\mathbf{1},\tilde\tau)}&= U_{\mathrm{2A}}(1)U'_{\mathrm{2A}}\oplus U_{\mathrm{2A}}(5)U'_{\mathrm{2A}}(2)\oplus U_{\mathrm{2A}}(6)U'_{\mathrm{2A}}(3) \oplus U_{\mathrm{2A}}(7)U'_{\mathrm{2A}}(4),\\
W_{(\tau,\mathbf{1})}&= U_{\mathrm{2A}} U'_{\mathrm{2A}}(1)\oplus U_{\mathrm{2A}}(2)U'_{\mathrm{2A}}(5)\oplus U_{\mathrm{2A}}(3)U'_{\mathrm{2A}}(6) \oplus U_{\mathrm{2A}}(4)U'_{\mathrm{2A}}(7),\\
W_{(\tau,\tilde\tau)}&= U_{\mathrm{2A}}(1)U'_{\mathrm{2A}}(1)\oplus U_{\mathrm{2A}}(5)U'_{\mathrm{2A}}(5)\oplus U_{\mathrm{2A}}(6)U'_{\mathrm{2A}}(6) \oplus U_{\mathrm{2A}}(7)U'_{\mathrm{2A}}(7),
\end{split}
\end{equation}
where we have labeled modules of $U_{\mathrm{2A}}$ and $U'_{\mathrm{2A}}$ following equations (3.47) and (3.55) of \cite{BHLLR21}. It is straightforward to check that the graded dimensions of the modules appearing in \eqref{eqn:secondFibsubalgebra} are compatible with \eqref{eqn:secondoption} to low orders in the $q$-expansion.
\end{proof}

Before moving on, we note that the characters of $V^\natural$ twisted by the Fibonacci symmetry corresponding to either of these two choices of subalgebra $W$ can be computed in terms of the ordinary characters of $W$. For example, using $\smash{\hat{\mathcal{L}}}$ to denote the operator acting on $\smash{V^\natural}$ corresponding to the non-trivial line $\mathcal{L}$ in the Fibonacci category $\operatorname{Fib}$, one obtains
\begin{equation}\label{eqn:fiblineop}
\begin{split}
\tr_{V^\natural}\hat{\mathcal{L}}q^{L_0-1} &= \varphi \ch_{W_{(\mathbf{1},\mathbf{1})}}(\tau)-\frac{1}{\varphi} \ch_{W_{(\tau,\tilde\tau)}}(\tau),
\end{split}
\end{equation}
where $\varphi=(1+\sqrt{5})/2$. One may then compute the $\mathcal{L}$-twisted sector $\smash{V^\natural_{\mathcal{L}}}$ by taking a modular $S$-transformation of equation~\eqref{eqn:fiblineop}. Alternatively, one may proceed abstractly by noting that $\smash{V^\natural}$ defines a Lagrangian algebra $L$ inside $\Rep(W)$, and the category $\smash{\operatorname{Ver}(V^\natural/W)\cong \operatorname{Fib}}$ is given by $\Rep(W)_L$, the category of $L$-modules inside $\Rep(W)$. It is known that there are precisely two irreducible $L$-modules, $L$ and $M$, in $\Rep(W)$, corresponding to the trivial line and the Fibonacci line $\mathcal{L}$, respectively. They decompose as simple objects of $\Rep(W)$ as
\begin{equation*}
L\cong W_{(\mathbf{1},\mathbf{1})}\oplus W_{(\tau,\tilde\tau)}, \ \ M\cong W_{(\mathbf1,\tilde\tau)}\oplus W_{(\tau,\mathbf{1})}\oplus W_{(\tau,\tilde\tau)}.
\end{equation*}
In particular, the second decomposition reveals that the graded dimension of the $\mathcal{L}$-twisted sector in $V^\natural$ is
\begin{equation*}
\tr_{V^\natural_{\mathcal{L}}}q^{L_0-1}=\ch_{W_{(\mathbf{1},\tilde\tau)}}(\tau) + \ch_{W_{(\tau,\mathbf{1})}}(\tau)+\ch_{W_{(\tau,\tilde\tau)}}(\tau).
\end{equation*}
See \cite{CRZ24} for details on how to obtain more involved genus-1 partition functions that encode the action of a Fibonacci symmetry in $V^\natural$ and in other 2d CFTs. (In the language of op.\ cit., the characters $\Ch_W(\tau)$ are the ``representation basis'' partition functions of $V^\natural$ with respect to a putative Fibonacci symmetry.)
\end{ex}

\begin{exph}[Every Finite Symmetry of the $\mathrm{SU}(2)_1$ RCFT Is Invertible]\label{exph:SU(2)1symmetries}
This example is based on \cite{CLR24}. See also \cite{FGRS07,TW24} for previous work on topological line interfaces in $(c_L,c_R)=(1,1)$ CFTs.

We use \autoref{prop:rcftsymmetries} to argue that every finite (unitary) symmetry category acting on the $\mathrm{SU}(2)_1$ Wess--Zumino--Witten RCFT is invertible, conditional on a mild technical hypothesis, which we state explicitly below.

We use the notation $\mathcal{H}_{\mathrm{SU}(2)_1}$ to denote the Hilbert space of the theory, which can be thought of as a Lagrangian algebra object in $\Rep(V_{A_1})\boxtimes\smash{\overline{\Rep(V_{A_1})}}$. Our strategy will be to show that the Verlinde categories $\operatorname{Ver}(\mathcal{H}_{\mathrm{SU}(2)_1}\,/\,V,\!W)$ are all invertible, where $V$ and $W$ run over \strat{} conformal subalgebras of $\mathsf{A}_{1,1}$. Since an arbitrary finite symmetry of $\mathrm{SU}(2)_1$ embeds into one of these Verlinde categories as a fusion subcategory, this shows that any finite symmetry of $\mathrm{SU}(2)_1$ is invertible.

As a lemma, we note that if $V'\subset V$ and $W'\subset W$ are \strat{} conformal subalgebras, then
\begin{equation*}
\operatorname{Ver}(\mathcal{H}\,/\,V,\!W)\subset \operatorname{Ver}(\mathcal{H}\,/\,V',\!W').
\end{equation*}
This is because the topological lines of an RCFT that commute with the chiral algebras $V$ and $W$ will certainly also commute with the subalgebras $V'$ and $W'$.

Invoking Theorem~3.2 of \cite{CGH19} (cf.\ \autoref{ex:cc1}), and making the mild assumption that every strongly rational conformal subalgebra of a unitary strongly rational \voa{} is again unitary, we learn that every positive, \strat{} conformal subalgebra of $\mathsf{A}_{1,1}$ is either one of three exceptional cases $V(E_6)$, $V(E_7)$ or $V(E_8)$, or else of the form $V_{\ell A_1}$ or $V_{\ell A_1}^+$ for some $\ell\in\Ns$. From the lemma above, we see that we do not need to consider the case that $V$ or $W$ is $V_{\ell A_1}$ because $V_{\ell A_1}^+$ is contained in it as a subalgebra.

We just consider the case that $V=V_{\ell A_1}^+$ and $W=V_{\ell' A_1}^+$ for some $\ell$ and $\ell'$; all other cases can be treated analogously. By invoking the lemma again, we find without loss of generality that we can take $\ell=\ell'$ because
\begin{equation*}
\operatorname{Ver}(\mathcal{H}_{\mathrm{SU}(2)_1}\,/\,V_{\ell A_1}^+,\!V_{\ell' A_1}^+) \subset \operatorname{Ver}(\mathcal{H}_{\mathrm{SU}(2)_1}\,/\,V_{\ell \ell' A_1}^+,\!V_{\ell \ell' A_1}^+).
\end{equation*}
Recall that a simple object of a unitary fusion category is invertible if and only if its quantum dimension is $1$. Using this, we can argue that $\operatorname{Ver}(\mathcal{H}_{\mathrm{SU}(2)_1}\,/\,V_{\ell A_1}^+,\!V_{\ell A_1}^+)$ is an invertible category by computing its rank and global dimension and showing that they agree.

To evaluate its dimension, note that $\operatorname{Ver}(\mathcal{H}_{\mathrm{SU}(2)_1}\,/\,V_{\ell A_1}^+,\!V_{\ell A_1}^+)$ can be thought of as the ``dual symmetry'' obtained by starting with the diagonal RCFT on the orbifold branch built on the chiral algebra $\smash{V^+_{\ell A_1}}$, and gauging the $\smash{\Rep(V^+_{\ell A_1})}$ symmetry. Thus, $\operatorname{Ver}(\mathcal{H}_{\mathrm{SU}(2)_1}\,/\,V_{\ell A_1}^+,\!V_{\ell A_1}^+)$ is Morita equivalent to $\Rep(V^+_{\ell A_1})$, and hence their categorical dimensions agree. In particular, we deduce that
\begin{equation*}
\dim\bigl(\operatorname{Ver}(\mathcal{H}_{\mathrm{SU}(2)_1}\,/\,V_{\ell A_1}^+,\!V_{\ell A_1}^+)\bigr)= \dim\bigl(\Rep(V^+_{\ell A_1})\bigr)=\sum_{i} S_{0i}^2 = 8\ell^2,
\end{equation*}
where the $S$-matrix of $V^+_{\ell A_1}$ can be found in Table~4 of \cite{DVVV89}.

On the other hand, it follows from Claim~3 in \cite{Ost03b} that the rank of the Verlinde category is given by
\begin{equation*}
\rk\bigl(\operatorname{Ver}(\mathcal{H}_{\mathrm{SU}(2)_1}\,/\,V_{\ell A_1}^+,\!V_{\ell A_1}^+)\bigr)=\sum_{i,j}Z_{ij}^2,
\end{equation*}
where $Z_{ij}$ is the non-negative integer matrix describing the decomposition of the Hilbert space of $\mathrm{SU}(2)_1$ into irreducible representations of $V^+_{\ell A_1}$, i.e.
\begin{equation*}
\mathcal{H}_{\mathrm{SU}(2)_1}\cong \bigoplus_{i,j} Z_{ij} V^+_{\ell A_1}(i) \otimes \overline{V^+_{\ell A_1}(j)}.
\end{equation*}
Let $\tilde\alpha=1/\sqrt{2 k}$. Following \cite{DN99b}, the irreducible modules of $V^+_{\sqrt{k} A_1}$ are
\begin{equation*}
\bigl\{ V_{\sqrt{k}A_1}^\pm, V_{\sqrt{k}A_1+k\tilde{\alpha}}^\pm\bigr\}\cup\bigl\{ V_{\sqrt{k}A_1+i\tilde{\alpha}}\bigm| i=1,\dots,k-1\bigr\}\cup\bigl\{V_{\sqrt{k}A_1}^{T_i,\pm}\bigm| i=1,2\bigr\},
\end{equation*}
and it is straightforward to determine that $V_{A_1}$ and its non-trivial irreducible module $V_{A_1}(\sfrac14)$ decompose as follows. If $\ell$ is even, then
\begin{align*}
V_{A_1} &\cong V_{\ell A_1}^+\oplus V_{\ell A_1}^-\oplus V^+_{\ell A_1 + \ell^2\tilde{\alpha}}\oplus V^-_{\ell A_1+\ell^2\tilde{\alpha}} \oplus \bigoplus_{j=1}^{\ell/2-1} 2 \cdot V_{\ell A_1+2j\ell\tilde{\alpha}},\\[-3mm]
V_{A_1}(\sfrac14) &\cong \bigoplus_{j=0}^{\ell/2-1} 2\cdot V_{\ell A_1 + (2j+1)\ell\tilde\alpha},
\end{align*}
whereas if $\ell$ is odd, then
\begin{align*}
V_{A_1}&\cong V_{\ell A_1}^+ \oplus V_{\ell A_1}^- \oplus \bigoplus_{j=1}^{(\ell-1)/2}2\cdot V_{\ell A_1 + 2j\ell\tilde\alpha},\\[-3mm]
V_{A_1}(\sfrac14)&\cong V_{\ell A_1+\ell^2\tilde\alpha}^+ \oplus V_{\ell A_1+\ell^2\tilde\alpha}^-\oplus \bigoplus_{j=0}^{(\ell-3)/2}2\cdot V_{\ell A_1 + (2j+1)\ell\tilde\alpha}.
\end{align*}
Using the fact that
\begin{equation*}
\mathcal{H}_{\mathrm{SU}(2)_1} \cong V_{A_1}\otimes \overline{V_{A_1}}\oplus V_{A_1}(\sfrac14)\otimes \overline{V_{A_1}(\sfrac14)}
\end{equation*}
we find in the first case that $Z=\operatorname{diag}(Z_1,Z_2)$, where $Z_1$ and $Z_2$ are the $(\ell+6)/2\times (\ell+6)/2$ and $\ell/2\times\ell/2$ matrices
\begin{equation*}
Z_1 =
\begin{pmatrix}
1 & 1 & 1 & 1 & 2 & \cdots & 2 \\
1 & 1 & 1 & 1 & 2 &\cdots & 2 \\
1 & 1 & 1 & 1 & 2 & \cdots & 2 \\
1 & 1 & 1 & 1 & 2 &\cdots & 2 \\
2 & 2 & 2 & 2 & 4 & \cdots & 4 \\
\vdots & \vdots & \vdots & \vdots & \vdots & \ddots & \vdots \\
2 & 2 & 2 & 2 & 4 & \cdots & 4
\end{pmatrix},\quad
Z_2 =
\begin{pmatrix}
4 & \cdots & 4 \\
\vdots & \ddots & \vdots \\
4 & \cdots & 4
\end{pmatrix},
\end{equation*}
and in the second case that $Z=\operatorname{diag}(Z',Z')$, where $Z'$ is the $(\ell+3)/2\times (\ell+3)/2$ matrix
\begin{equation*}
Z'=
\begin{pmatrix}
1 & 1 & 2 & \cdots & 2 \\
1 & 1 & 2 & \cdots & 2 \\
2 & 2 & 4 & \cdots & 4 \\
\vdots & \vdots & \vdots & \ddots & \vdots \\
2 & 2 & 4 & \cdots & 4
\end{pmatrix}.
\end{equation*}
In both cases, one finds that
\begin{equation*}
\sum_{i,j} Z_{ij}^2 = 8\ell^2.
\end{equation*}
In particular, the rank and dimension of $\operatorname{Ver}(\mathcal{H}_{\mathrm{SU}(2)_1}\,/\,V_{\ell A_1}^+,\!V_{\ell A_1}^+)$ are equal to one another, so it is an invertible category.

One can repeat this analysis for the remaining Verlinde categories and conclude that all finite symmetries of $\mathrm{SU}(2)_1$ are invertible, and hence must be finite subgroups of $\Aut(\mathcal{H}_{\mathrm{SU}(2)_1})=\mathrm{SO}(4)$.
\end{exph}


\section{Neighborhood}\label{sec:neighborhood}

In this section, we introduce the notion of $n$-neighborhood of \strat{} \voa{}s. The connection between the (bulk) genus and neighborhood will turn out to be less clear than in the lattice context (see \autoref{sec:lat}).

We discuss the holomorphic case in detail in \autoref{subsec:holneigh}, and concrete examples in \autoref{subsec:neighborexamples}. See \cite{Mon98,HM23} for some earlier work for $n=2,3$.


\subsection{Neighborhood and Orbifolds}\label{subsec:orbneighborhood}

Inspired by the corresponding lattice definition, \autoref{defi:latneigh}, we shall define the notion of $n$-neigh\-bor\-hood for \voa{}s, which for holomorphic \voa{}s coincides with the $\Z_n$-orbifold construction studied in \cite{FLM88,EMS20a,Moe16}. This notion is a specialization of orbifold equivalence, \autoref{defi:orbequiv}, where we allow only one orbifold step and specify the index of the conformal extensions. Often, we shall take $n=p$ where $p$ is a prime.

\begin{defi}[Neighborhood]\label{defi:neighbor}
Let $V$ and $V'$ be \strat{} \voa{}s and $n\in\Ns$. We say $V$ and $V'$ are \emph{$n$-neighbors} if there is a \strat{} \voa{} $W$ such that both $V$ and $V'$ are (isomorphic to) conformal extensions of $W$ by a simple-current module of order $n$ (i.e.\ $\Z_n$-extensions).
\end{defi}

The definition is depicted in \autoref{fig:neighbor}. Equivalently, the \voa{}s $V$ and $V'$ admit automorphisms $g$ and $g'$, respectively, of order~$n$ whose \fpvosa{}s (also called orbifolds) are both isomorphic to $W$, recalling that the strong rationality of $V$ or $V'$ implies that of $W$ \cite{Miy15,CM16} (see also \cite{McR21b,McR26}) for solvable and in particular cyclic groups of automorphisms.

\begin{figure}[ht]
\centering
\begin{tikzcd}
V\arrow[-,dashed]{rr}{n\text{-neighbor}}& & V'\\
& W\arrow[hookrightarrow]{ul}{\Z_n\text{-ext.}}\arrow[hookrightarrow]{ur}[swap]{\Z_n\text{-ext.}}
\end{tikzcd}
\caption{Neighborhood of \strat{} \voa{}s $V$ and $V'$.}
\label{fig:neighbor}
\end{figure}

We shall investigate in what way $n$-neigh\-bor\-hood is compatible with the various notions of genera of \voa{}s that we introduced in \cite{MR24a}, in particular in the holomorphic case discussed in \autoref{subsec:holneigh}. We begin with some initial suppositions and observations related to the bulk genus:

\begin{rem}\label{rem:neighborbulk}
Recall from \autoref{sec:lat} that two even lattices $L$ and $M$ that are (iterated) $p$-neighbors of each other for some prime $p$ are in the same lattice genus as long as $p$ has nothing to do with $L'/L$ and $M'/M$ in the sense that $p$ does not divide the discriminant $|L'/L|$ of $L$, necessarily equal to the discriminant of $M$.

One is tempted to formulate an analogous statement for \voa{}s, perhaps along the lines of: \emph{for \strat{} \voa{}s $V$ and $V'$ and a prime $p$ not dividing $|G(\Rep(V))|$, where $G(\Rep(V))$ is the group of simple currents, if $V$ and $V'$ are iterated $p$-neighbors, then $V$ and $V'$ are in the same bulk genus, i.e.\ $\Rep(V)\cong\Rep(V')$.}

While we are not certain that such a statement is true in general, it does hold when $\Rep(V)\cong\Vect$. In that case, the representation category $\Rep(W)$ of the common conformal subalgebra $W$ of $V$ and $V'$ must be a twisted Drinfeld double $\mathcal{D}_\omega(\Z_p)$, and any condensation by a simple current of order~$p$ with trivial twist (defining a condensable algebra) yields again $\Vect$. This case is discussed in detail in \autoref{subsec:holneigh}.
\end{rem}

\begin{rem}
In contrast to even lattices, a converse of the above statement is certainly false. Indeed, as we shall mention in \autoref{ex:schellekens3}, the $p$-neigh\-bor\-hood graph of the \strat{}, holomorphic \voa{}s of central charge $24$ is known to be disconnected for any prime $p$.

However, one might conjecture: \emph{if one is allowed to consider $p$-neighbors for different primes $p$, then any bulk genus becomes connected under $p$-neigh\-bor\-hood.}

This is known to be true for the bulk genus $(\Vect,24)$ assuming moonshine uniqueness, \autoref{conj:moonshineuniqueness} (cf.\ \autoref{ex:deepholes}). Moreover, note that this conjecture would already imply \autoref{conj:anisobulkorb} and hence \autoref{conj:wittorb} because iterated $p$-neighbors are orbifold equivalent.
\end{rem}

Carving out a clearer relationship between neighborhood of \voa{}s and the bulk genus is left for future work, but we shall address the problem experimentally in \autoref{subsec:neighborexamples}.

\medskip

Not surprisingly, $n$-neigh\-bor\-hood for \voa{}s generalizes lattice $n$-neigh\-bor\-hood, \autoref{defi:latneigh}, in the following sense:

\begin{ex}[Lattice \VOA{}s]
Let $L$ and $M$ be two positive-definite, even lattices that are $n$-neighbors for some $n\in\Ns$, via the index-$n$ sublattice $K$. Then the corresponding \strat{} \voa{}s $V_L$ and $V_M$ are $n$-neighbors in the sense of \autoref{defi:neighbor}, with $W=V_K$.

In fact, they are inner $n$-neighbors, with trivial Heisenberg commutant $C$, as in \autoref{defi:innerneighbor} below. If we further assume that $L$ and $M$ are unimodular (and that $n=p$ is prime), then $V_L$ and $V_M$ are inner $p$-neighbors of type~I, as described in \autoref{subsec:primeinner}.
\end{ex}


\subsection{Inner Neighborhood}\label{subsec:innerneighborhood}

If we restrict attention to neighbors that are also inner orbifold equivalent in the sense of \autoref{defi:innerorbequ}, it becomes feasible to make contact with the notion of hyperbolic genus in the commutant formulation of \cite[\autoref*{MR1defi:commgenus}]{MR24a}. In \autoref{subsec:holneigh}, we shall describe this connection in detail under the simplifying assumption that $V$ and $V'$ are holomorphic (see \autoref{conj:hypgenpneighbor}, \autoref{prop:hypgenpneighbor1} and \autoref{prop:hypgenpneighbor2}).

\medskip

We begin with a short reminder about inner automorphisms. Recall from \cite[\autoref*{MR1sec:innerauts}]{MR24a} that an automorphism of a, say \strat{}, \voa{} $V$ is called \emph{inner} if it is a product of elements of the form $e^{a_0}$ for $a\in V_1$, the weight-$1$ space of $V$. The inner automorphisms form a normal subgroup of the full automorphism group, $\Inn(V)<\Aut(V)$.

If $V$ is \strat{}, then $V_1$ is a finite-dimensional, reductive Lie algebra \cite{DM04b}. Let $\hh$ be a choice of Cartan subalgebra of $V_1$ or equivalently of $V$. Then, there is the abelian subgroup $T=T(\hh)\coloneqq\langle\{e^{a_0}\mid a\in\hh\}\rangle=\{e^{a_0}\mid a\in\hh\}<\Inn(V)$.

Under reasonably mild assumptions (e.g., when the kernel of the restriction map $\Aut(V)\to\Aut(V_1)$ is contained in $T$, which holds for all lattice \voa{}s \cite{DN99} and for the Schellekens \voa{}s in central charge $24$ with $V_1\neq\{0\}$ \cite{LS20b}), any finite-order, inner automorphism in $\Inn(V)$ is conjugate in $\Inn(V)$ to an automorphism in $T$ \cite{HM22}, for one and hence any choice of Cartan subalgebra $\hh$.

\medskip

We now specialize the notion of neighborhood of \voa{}s, analogously to the definition of inner orbifold equivalence, \autoref{defi:innerorbequ}:

\begin{defi}[Inner Neighborhood]\label{defi:innerneighbor}
Let $V$ and $V'$ be \strat{} \voa{}s. Let $n\in\Ns$. We say $V$ and $V'$ are \emph{inner $n$-neighbors} if there is a \strat{} \voa{} $W$ such that both $V$ and $V'$ are (isomorphic to) conformal extensions of $W$ by a simple-current module of order $n$ (i.e.\ $\Z_n$-extensions) and furthermore such that the concrete realization of $W$ inside $V$ (and similarly for $V'$) can be chosen such that $W$ and $V$ have the same Heisenberg commutant.
\end{defi}
In the orbifold picture, where $W\cong V^g$ for some automorphism $g\in\Aut(V)$, the condition that $W$ have the same Heisenberg commutant as $V$ implies that there is a choice of Cartan subalgebra of $V$ that is fixed point-wise by $g$ and is also a Cartan subalgebra of $W$ (and analogously for $V'$).

Based on this, we can characterize the two \voa{} automorphisms corresponding to a pair of inner neighbors $V$ and $V'$.
\begin{prop}\label{prop:innerneighchar}
Let $V$ and $V'$ be \strat{} \voa{}s. Then $V$ and $V'$ are inner $n$-neighbors if and only if the corresponding automorphisms $g\in\Aut(V)$ and $g'\in\Aut(V')$ of order $n$ are both (up to algebraic conjugacy) in $T$, i.e.\ of the form $g=e^{h_0}$ and $g'=e^{h'_0}$ for elements $h$ and $h'$ in some choices of Cartan subalgebras of $V$ and $V'$, respectively.
\end{prop}
\begin{proof}
The backward direction is immediate as $g=e^{h_0}$ and $g'=e^{h'_0}$ act trivially on the Cartan subalgebras and Heisenberg commutants. (See \cite{HM23} for a more detailed discussion in the case of $n=2$. Below we even state an explicit formula for the associated lattice of $V^g$.)

We now prove the forward direction. Let $W=V^g$ be the \fpvosa{} under an automorphism $g$ of $V$ of order~$n$. Suppose that both $V$ and $V^g$ have the Heisenberg commutant $C$ and, by the comment above, the same Cartan subalgebra $\hh$ (and hence the same Heisenberg \voa{} $M_{\hat\hh}(1,0)$ generated by it). That is, $g$ fixes $C\otimes M_{\hat\hh}(1,0)\subset V$ point-wise. In particular, $g$ commutes with $h_0$ for all $h\in\hh$.

Recall from \cite[\autoref*{MR1prop:asslatdecomp}]{MR24a} that the \voa{} $V$ decomposes (without multiplicities) into irreducible modules for $C\otimes M_{\hat\hh}(1,0)$ as
\begin{equation*}
V=\bigoplus_{\alpha\in P}C^{\tau(\alpha+L)}\otimes M_{\hat\hh}(1,\alpha),
\end{equation*}
where $P\subset L'$ is a sublattice with $A=P/L$.

We consider how $V$ splits up into eigenspaces $V^g=V^{(0)},V^{(1)},\ldots,V^{(n-1)}$ for~$g$, which are modules for $V^g$ and hence for $\smash{C\otimes M_{\hat\hh}(1,0)\subset V^g}$. Since $g$ commutes with $h_0$ for all $h\in\hh$, its action restricts to each summand $C^{\tau(\alpha+L)}\otimes M_{\hat\hh}(1,\alpha)$ in the above direct sum. On the other hand, each summand can only belong to one eigenspace $V^{(i)}$ of $g$. If not, $C^{\tau(\alpha+L)}\otimes M_{\hat\hh}(1,\alpha)$ would decompose into a direct sum of more than one module for $C\otimes M_{\hat\hh}(1,0)\subset V^g$, contradicting the irreducibility of $C^{\tau(\alpha+L)}\otimes M_{\hat\hh}(1,\alpha)$. In other words,
\begin{equation*}
g|_{C^{\tau(\alpha+L)}\otimes M_{\hat\hh}(1,\alpha)}=\lambda(\alpha)\id,\quad\alpha\in P,
\end{equation*}
for some function $\lambda\colon P\to\C^\times$. As $V$ is $P$-graded as a \voa{} and $g$ is a \voa{} automorphism, $\lambda$ must be a homomorphism. It finally follows that we may write $g=e^{h_0}$ for some $h\in\hh$. In fact, since $g$ has order~$n$, $h\in2\pi i(1/n)P'$. This proves the assertion for $V$, but for $V'$ one proceeds exactly in the same manner.
\end{proof}
\begin{rem}\label{rem:nameinner}
The proposition gives an a-posteriori justification for the name \emph{inner orbifold equivalence} in \autoref{defi:innerorbequ}, which generalizes inner neighborhood. Indeed, $g=e^{h_0}$ and $g'=e^{h'_0}$ are examples of inner automorphisms in $T\subset\Inn(V)$.

Moreover, as we explained above, it is often the case that any finite-order, inner automorphism of $V$ or $V'$ can be conjugated into an automorphism of this form.
\end{rem}

\autoref{prop:innerneighchar} allows us to very explicitly describe the relation between \voa{}s that are inner neighbors, in particular with regard to their lattice structure. This will be useful for the explicit computations in \autoref{subsec:neighborexamples}. Suppose that $V$ is a \strat{} \voa{} with associated lattice $L$, Heisenberg commutant $C$ and corresponding dual pair $C\otimes V_L$ with decomposition
\begin{equation*}
V=\bigoplus_{\alpha+L\in A}C^{\tau(\alpha+L)}\otimes V_{\alpha+L}
\end{equation*}
as in \cite[\autoref*{MR1prop:asslatdecomp}]{MR24a}, where $A=P/L$ for some sublattice $P\subset L'$. Then, the \fpvosa{} under the automorphism $g=e^{2\pi i h_0}$ for some $h$ in the Cartan subalgebra $\hh=\C\vac\otimes\{k(-1)\ee_0\mid k\in L\otimes_\Z\C\}$,\footnote{Here, $\vac$ denotes the vacuum vector of $C$ and $\ee_0$ is the vacuum vector of $V_L$, following standard notation for lattice \voa{}s.} which has order~$n$ if and only if $h$ has order $n$ in $P'\otimes_\Z\Q/P'$, has the following form:
\begin{equation*}
V^g=\bigoplus_{\alpha+L\in A}C^{\tau(\alpha+L)}\otimes V_{(\alpha+L)^h}
\end{equation*}
with $(\alpha+L)^h\coloneqq\{\beta\in\alpha+L\,|\,\langle\beta,h\rangle\in\Z\}$. The associated lattice of $V^g$ is $L^h$ so that $V^g$ is a simple-current extension of the dual pair $C\otimes V_{L^h}$ (see also \cite{HM23}). The sets $(\alpha+L)^h$ are either empty or cosets in $(L^h)'/L^h$. $L^h$ has index dividing $n$ in $L$, but could be equal to $L$. The $\Z_n$-extension of $V^g$ to the neighbor $V'$ is analogous.

We shall discuss further details in \autoref{subsec:holneigh} in the case when the neighbors $V$ and $V'$ are holomorphic (and when $n=p$ is a prime).


\subsection{Holomorphic Neighbors}\label{subsec:holneigh}

In the following, we specialize the discussion to the situation where the two $n$-neighbors $V$ and $V'$ (for now, not necessarily inner) are \strat{} and holomorphic.

In fact, it suffices to assume that $V$ is holomorphic, i.e.\ that $\Rep(V)\cong\Vect$. Then also $\Rep(V')\cong\Vect$ so that $V$ and $V'$ are in the same bulk genus and the conjecture in \autoref{rem:neighborbulk} is true (with the condition on $p$ becoming vacuous).

This situation, where $V$ and $V'$ are holomorphic \voa{}s and $W\coloneqq V^g\cong V'^{g'}$ for some automorphisms $g$ and $g'$ of order~$n$, was already studied in \cite{EMS20a,Moe16,DRX17}, where it is called \emph{cyclic orbifold construction}. There, it is shown that $\Rep(W)$ is a twisted Drinfeld double $\mathcal{D}_\omega(\Z_n)$ for some $3$-cocycle $\omega\in H^3(\Z_n,\C^\times)\cong\Z_n$, a certain pointed modular tensor category (see also \cite{DPR90,DNR25,GR24}). The cohomology class $\omega\in\Z_n$ (called the anomaly, or in \cite{EMS20a,Moe16} the type) can be determined as $n^2\rho(V(g))\pmod{n}$, where $\rho(V(g))$ is the smallest $L_0$-weight of the unique irreducible, $g$-twisted $V$-module $V(g)$.

\begin{rem}\label{rem:type0}
Given two holomorphic $n$-neighbors $V$ and $V'$, we may without loss of generality assume that
\begin{enumerate}
\item the cohomology class $\omega\in\Z_n$ is trivial, in which case $\Rep(W)$ is the pointed modular tensor category (see, e.g., \cite{EGNO15}) associated with the metric group $D_n\coloneqq\Z_n\times\Z_n$ with quadratic form $(i,j)\mapsto ij/n+\Z$, and
\item the isotropic subgroups in $D_n$ corresponding to $V$ and $V'$ (i.e.\ to the condensable algebras they define) have trivial intersection (and they are maximal as $V$ and $V'$ are holomorphic).
\end{enumerate}
If not, $V$ and $V'$ are $d$-neighbors for some $d\mid n$ with $d<n$ such that these conditions are satisfied \cite{Moe16}.
\end{rem}

The group $D_n=\Z_n\times\Z_n$ with quadratic form $(i,j)\mapsto ij/n+\Z$ is the discriminant form of $\II_{1,1}(n)$, the unique even, unimodular lattice $\II_{1,1}$ of signature $(1,1)$ with the quadratic form scaled by $n$. For example, for $n=p$ prime, the Jordan decomposition \cite{CS99} of $D_p$ is $2_{\II}^{+2}$, $3^{-2}$, $5^{+2}$, $7^{-2}$, $11^{-2}$, etc.

The automorphism $g\in\Aut(V)$ characterizing the $\Z_n$-extension of $W$ to $V$ uniquely determines the automorphism $g'\in\Aut(V')$ describing the $\Z_n$-extension of $W$ to $V'$ (called the inverse-orbifold automorphism) and vice versa.

\medskip

After these general remarks, we return to the study of inner $n$-neigh\-bor\-hood, now assuming that both $V$ and $V'$ are holomorphic. In this situation, we want to state a precise relation to the hyperbolic genus, in the Heisenberg commutant formulation of \cite[\autoref*{MR1defi:commgenus}]{MR24a}. The following is immediate:
\begin{prop}\label{prop:hypgenpneighbor1}
Let $V$ and $V'$ be \strat{}, holomorphic \voa{}s. For fixed $n\in\Ns$, if $V$ and $V'$ are iterated inner $n$-neighbors, then they are in the same hyperbolic genus.
\end{prop}
\begin{proof}
If $V$ and $V'$ are iterated inner $n$-neighbors, then, by definition, they have the same Heisenberg commutant. Then, \cite[\autoref*{MR1cor:holhyp}]{MR24a} implies that they are in the same hyperbolic genus.
\end{proof}

We also want to formulate the converse direction, but this will turn out to be more difficult (see \autoref{conj:hypgenpneighbor} and \autoref{prop:hypgenpneighbor2}).

Recall from \cite[\autoref*{MR1sec:holhyp}]{MR24a} that, as $V$ is holomorphic, the statement of \cite[\autoref*{MR1prop:asslatdecomp}]{MR24a} simplifies and $V$ (and similarly $V'$) is a simple-current extension of the dual pair $C\otimes V_L$ of the form
\begin{equation}\label{eqn:holdecomp}
V=\bigoplus_{\alpha+L\in L'/L}C^{\tau(\alpha+L)}\otimes V_{\alpha+L}
\end{equation}
for some ribbon-reversing tensor equivalence $\tau\colon\Rep(V_L)\to\Rep(C)$. In particular, $\Rep(C)$ is pointed and pseudo-unitary, i.e.\ $\Rep(C)\cong\mathcal{C}(R_C)$ for the metric group $\smash{R_C\cong\overline{L'/L}}$. Here, for a metric group $D$, we denote by $\overline{D}$ the same abelian group with the quadratic form multiplied by $-1$.

\begin{rem}\label{rem:samelatgen}
Because both $V$ and $V'$ have the same Heisenberg commutant $C$ and are holomorphic, their associated lattices $L$ and $M$, respectively, must have the same discriminant form, namely the one dual to the pointed modular tensor category $\Rep(C)$. Hence, $L$ and $M$ must lie in the same lattice genus, regardless of whether $n$ is coprime to $d(L)=d(M)$ or not (cf.\ the discussion in \autoref{sec:lat}).
\end{rem}
In the following, we describe the relation between the associated lattices $L$ and $M$ of $V$ and $V'$ more thoroughly, in particular in \autoref{subsec:primeinner} in the prime case.

To this end, we specialize the inner orbifold picture, \autoref{prop:innerneighchar}, to the holomorphic case. Let $g$ be an inner automorphism in $T=T(\hh)$ of the form $g=e^{2\pi i h_0}$ for some $h$ in the Cartan subalgebra $\hh\coloneqq\C\vac\otimes\{k(-1)\ee_0\mid k\in L\otimes_\Z\C\}$. Assume that $g$ has order $n$ or equivalently that $h\in L\otimes_\Z\Q/L$ has order~$n$, necessarily implying that $h\in(1/n)L/L$. The \fpvosa{} $W=V^g$ decomposes as
\begin{equation*}
W=V^g=\bigoplus_{\alpha+L\in L'/L}C^{\tau(\alpha+L)}\otimes V_{(\alpha+L)^h}
\end{equation*}
where $(\alpha+L)^h\coloneqq\{\beta\in\alpha+L\mid\langle\beta,h\rangle\in\Z\}$. The latter are cosets in $(L^h)'/L^h$ if they are non-empty, i.e.\ the $V_{(\alpha+L)^h}$ are irreducible modules for $V_{L^h}$. Moreover,
\begin{equation*}
K\coloneqq L^h=\{\beta\in L\mid\langle\beta,h\rangle\in\Z\}
\end{equation*}
is the associated lattice of $W=V^g$ and $C$ the Heisenberg commutant, so that $C\otimes V_{L^h}\subset V^g$ is a dual pair in $W$. The lattice index $[L:L^h]$ is some divisor of $n$.

We further assume without loss of generality (see \autoref{rem:type0}) that $g=e^{2\pi i h_0}$ has type~$0$, i.e.\ that the unique irreducible $g$-twisted $V$-module $V(g)$ has $L_0$-weights in $(1/n)\Z$. As can be seen from the structure of $V(g)$, which we describe momentarily, this is equivalent to $\langle h,h\rangle/2\in(1/n)\Z$ (see also \cite{HM23} for more details in the case of $n=2$). Then, as explained above, there is a unique extension of $W$ to a holomorphic \voa{} $V'$, necessarily a $\Z_n$-extension, such that $V'$ and $V$ have minimal intersection. This is the $n$-neighbor or orbifold construction $V'=V^{\orb(g)}$. $V'$ is composed of parts of the unique irreducible $g^i$-twisted $V$-modules $V(g^i)$ for $i\in\Z_n$. Now, for an inner automorphism $g=e^{2\pi i h_0}$ these modules take the simple form
\begin{equation*}
V(g^i)=\bigoplus_{\alpha+L\in L'/L}C^{\tau(\alpha+L)}\otimes V_{ih+\alpha+L}
\end{equation*}
for $i\in\Z_n$, where we may view $V_{ih+\alpha+L}$ as a (not necessarily irreducible) module over $V_{L^h}\subset V^g$ because $ih+L'\subset(L^h)'$ \cite{Li96}.


\subsection{Prime Order}\label{subsec:primeinner}

We now, for simplicity, restrict to the case that $n=p$ is prime. We continue the discussion above, studying inner $p$-neigh\-bor\-hood for \strat{}, holomorphic \voa{}s.

The condition that $g$ have order $p$ is now equivalent to $h\in(1/p)L\setminus L$. Since $p$ is prime, the orbifold construction $V'=V^{\orb(g)}$ can be written as
\begin{equation*}
V'=V^{\orb(g)}=V^g\oplus\bigoplus_{i\in\Z_p\setminus\{0\}}V(g^i)_\Z,
\end{equation*}
where $V(g^i)_\Z$ denotes the irreducible $V^g$-submodule of $V(g^i)$ obtained by keeping only the elements of integral $L_0$-weights. Hence,
\begin{equation*}
V'=V^{\orb(g)}=\bigoplus_{\alpha+L\in L'/L}C^{\tau(\alpha+L)}\otimes\Bigl(V_{(\alpha+L)^h}\oplus\bigoplus_{i\in\Z_p\setminus\{0\}} V_{(ih+\alpha+L)_\Z}\Bigr)
\end{equation*}
with $(ih+\alpha+L)_\Z\coloneqq\{\beta\in ih+\alpha+L\mid\langle\beta,\beta\rangle/2\equiv\langle\alpha,\alpha\rangle/2\pmod{1}\}$. These sets decompose into cosets of $L^h$ so that $V_{(ih+\alpha+L)_\Z}$ is a (not necessarily irreducible) module over $V_{L^h}\subset V^g$. Because we assumed that $C$ is also the Heisenberg commutant of $V'=V^{\orb(g)}$, the associated lattice of $V'$ is given by
\begin{equation*}
M\coloneqq L^h\cup\bigcup_{i\in\Z_p\setminus\{0\}}(ih+L)_\mathrm{ev}
\end{equation*}
where $(\cdot)_\mathrm{ev}$ denotes the subset of vectors $v$ of norm $\langle v,v\rangle/2\in\Z$. That the above expression is indeed an (even) lattice is a slightly non-trivial observation that uses $\langle h,h\rangle/2\in(1/p)\Z$.

Generalizing the discussion in \cite{HM23} to arbitrary primes, we distinguish two types of inner automorphisms:
\begin{prop}\label{prop:neightwocases}
Let $V$ and $V'$ be inner $p$-neighbors for some prime $p$. Let $g=e^{2\pi i h_0}\in\Aut(V)$ for $h\in\hh$, with $h\in(1/p)L\setminus L$ and $\langle h,h\rangle/2\in(1/p)\Z$, be the corresponding inner automorphism on the $V$-side. Then, one of two cases can occur, where $L$ denotes the associated lattice of $V$:
\begin{enumerate}
\item[(I)] $h\notin L'$ (if and only if $[L:L^h]=p$),
\item[(IIa)] $h\in L'$ (if and only if $[L:L^h]=1$) and $\langle h,h\rangle/2\in(1/p)\Z\setminus\Z$.
\end{enumerate}
The automorphism $g'=e^{2\pi i h'_0}\in\Aut(V')$ describing the $p$-neigh\-bor\-hood on the $V'$-side must satisfy the analogous condition.
\end{prop}
\begin{proof}
The proof given in \cite{HM23} (see in particular Table~4) for $p=2$ holds verbatim for arbitrary primes $p$.
\end{proof}
There can also be inner automorphisms of $V$ describing $p$-neighbors satisfying (IIb) $h\in L'$ (if and only if $[L:L^h]=1$) and $\langle h,h\rangle/2\in\Z$, but they can only correspond to $g'\in\Aut(V')$ being non-inner, so they do not appear in the context of inner neighbors.

\medskip

We describe the precise relation between the associated lattices $L$, $K=L^h$ and $M$ of $V$, $W$ and $V'$, respectively, and between the corresponding simple-current decompositions for types I and IIa. It will be convenient to base this description on $W$ with the associated lattice $K$ and describe $V$ and $V'$ as extensions thereof. The discussion is borrowed from Section~6 in \cite{HM23}, with all powers of $2$ replaced by the corresponding $p$-power.

Recall the decomposition \eqref{eqn:holdecomp} of $V$, and in particular that the Heisenberg commutant $C$ of $V$ and $V'$ is pointed and pseudo-unitary with $\Rep(C)\cong\mathcal{C}(R_C)$ for the metric group $R_C\cong\overline{L'/L}$. The \voa{} $C$ also being the Heisenberg commutant of $W$, the latter is of the form
\begin{equation*}
W=\bigoplus_{\alpha+K\in A}C^{\tau(\alpha+K)}\otimes V_{\alpha+K}
\end{equation*}
for some subgroups $A<K'/K$ and $A'<R_C$ and an anti-isometry $\tau\colon A\to A'$. The \voa{} $W$ is characterized as extension of the dual pair $C\otimes V_K$ by the isotropic subgroup
\begin{equation*}
H=H_W=\{(\tau(x),x)\,|\,x\in A\}<R_C\times K'/K,
\end{equation*}
implying that $\Rep(W)\cong\mathcal{C}(H^\bot/H)$, which by assumption must be the (untwisted) Drinfeld double $\mathcal{D}_0(\Z_p)\cong\mathcal{C}(D_p)$, with $D_p\cong\Z_p\times\Z_p$ as groups.

The holomorphic \voa{}s $V$ and $V'$ are extensions of $W$, and so they correspond to the two (trivially intersecting) maximal isotropic subgroups of $H^\bot/H\cong D_p$, which are of order~$p$. This implies that $[R_C:A'][K'/K:A]=p^2$. Based on this equation, one can distinguish four cases (cf.\ \cite{HM23}), of which only two, namely I and IIa are relevant in the setting of inner neighbors:
\begin{enumerate}
\item[(I)] $[R_C:A']=1$ and $[K'/K:A]=p^2$ (so that $|K'/K|=p^2\,|R_C|$),
\item[(IIa)] $[R_C:A']=p$ and $[K'/K:A]=p$ (so that $|K'/K|=|R_C|$).
\end{enumerate}

\subsubsection*{Type I}
In this case, $A$ has index $p^2$ in $K'/K$ so that $D_p\cong H^\bot/H\cong\smash{A^\bot}$ (and $\smash{{A'}^\bot}$ is trivial). Hence, $\smash{A^\bot}$ has two maximal isotropic subgroups (each of order~$p$), say generated by the elements $\gamma_1$ and $\gamma_2$ of norm~$0$. Then $V$ and $V'$ are of the form
\begin{equation*}
V^{(}{'}{}^{)}=\bigoplus_{\alpha+K\in A}C^{\tau(\alpha+K)}\otimes\bigl(\bigoplus_{i\in\Z_p}V_{i\gamma_j+\alpha+K}\bigr)
\end{equation*}
with dual pairs $C\otimes(\bigoplus_{i\in\Z_p}V_{i\gamma_j+K})=C\otimes V_{\cup_{i\in\Z_p} (i\gamma_j+K)}$, $j=1,2$. The even lattices
\begin{equation*}
L\coloneqq\cup_{i\in\Z_p}(i\gamma_1+K)\quad\text{and}\quad M\coloneqq\cup_{i\in\Z_p}(i\gamma_2+K)
\end{equation*}
are the associated lattices of $V$ and $V'$, respectively. They are $p$-extensions of $K$, the associated lattice of $W$.
\begin{rem}\label{rem:typeI}
We have shown that for inner neighborhood of type~I, the associated lattices $L$ and $M$ are usual lattice $p$-neighbors (and recall from \autoref{rem:samelatgen} that $L$ and $M$ are in the same lattice genus). Hence, classical lattice results (see \autoref{sec:lat}) yield restrictions on the existence of inner neighbors of type~I.
\end{rem}

\subsubsection*{Type IIa}
This case is slightly more complicated. Here, both $A$ and $A'$ have index~$p$ in $K'/K$ and $R_C$, respectively. In contrast to the previous case, $H^\bot$ is now strictly larger than $\smash{H+{A'}^\bot\times A^\bot}$. Here, $\smash{(H+{A'}^\bot\times A^\bot)/H}$ forms a (non-isotropic) subgroup of $H^\bot/H$ of order~$p$. On the other hand, there are two elements $\gamma_1=(c'_1,c_1)+H$ and $\gamma_2=(c'_2,c_2)+H$ of order~$p$ and norm~$0$ (generating isotropic subgroups) in $H^\bot/H\cong D_p$ that correspond to the extensions of $W$ to $V$ and $V'$, respectively,
\begin{equation*}
V^{(}{'}{}^{)}=\bigoplus_{\alpha+K\in A}\bigoplus_{i\in\Z_p}\bigl(C^{ic'_j+\tau(\alpha+K)}\otimes V_{ic_j+\alpha+K}\bigr),
\end{equation*}
$j=1,2$, each with dual pair $C\otimes V_K$. The even lattice $L=M=K$ is the associated lattice of $V$, $V'$ and $W$ (compatible with \autoref{rem:samelatgen}). We can extend the anti-isometry $\tau\colon A\to A'$ to $\tilde\tau_j\colon K'/K\to R_C$ by demanding $\tilde\tau_j(c_j+K)=c'_j\in R_C$ so that
\begin{equation*}
V^{(}{'}{}^{)}=\bigoplus_{\alpha+K\in K'/K}C^{\tilde\tau_j(\alpha+K)}\otimes V_{\alpha+K}
\end{equation*}
for $j=1,2$, recovering the usual decomposition for holomorphic \voa{}s.

\begin{rem}\label{rem:typeIIa}
By definition, inner $p$-neighbors of type~IIa can only exist if $p$ divides the determinant of the associated lattice (which is the same for both neighbors, namely $K$). This heavily constrains the existence of inner $p$-neighbors of type~IIa.
\end{rem}

To summarize, for type~I, the extension from $W$ to $V$ and $V'$ is along the associated lattice $K$ (corresponding to usual lattice neighborhood), while the gluing map $\tau$ is unchanged. On the other hand, for type~IIa, the extension is along the gluing map $\tau$, while the associated lattice remains the same.

\medskip

Based on experimental evidence (see, e.g., \autoref{ex:schellekens3b}) we conjecture, in analogy to the corresponding result for even lattices in \autoref{sec:lat}:
\begin{conj}\label{conj:hypgenpneighbor}
Let $V$ and $V'$ be \strat{}, holomorphic \voa{}s. Fix a prime $p$ large enough. Then $V$ and $V'$ are in the same hyperbolic genus if and only if they are iterated inner $p$-neighbors. (In fact, for the forward direction, we add the assumption that the genus of the associated lattice of $V$, which is also the genus of that of $V'$ by \autoref{rem:samelatgen}, does not contain more than one spinor genus.)
\end{conj}

The main result of this section will be a somewhat weaker statement on the relation between hyperbolic equivalence and (inner) neighborhood.

Recall that in \autoref{prop:hypgenpneighbor1} we already showed that inner neighborhood of holomorphic \voa{}s preserves the hyperbolic genus (cf.\ \autoref{rem:neighborbulk}). Our goal is to prove a statement in the converse direction.

We first need the following statement on metric groups:
\begin{lem}\label{lem:isotropicneighbor}
Let $D$ be a metric group. Suppose there are self-dual, isotropic subgroups $H$ and $H'$ of $D$. Then there is a sequence of self-dual, isotropic subgroups $H_0=H,H_1,\dots,H_n=H'$ such that $[H_i:H_i\cap H_{i+1}]=[H_{i+1}:H_i\cap H_{i+1}]$ is prime for all $i=0,\dots,n-1$.
\end{lem}
Two self-dual, isotropic subgroups $H$ and $H'$ (necessarily satisfying $|H|=|H'|$) with $[H:H\cap H']=[H':H\cap H']=p$ for some prime $p$ are called \emph{$p$-neighbors}. We note that the primes appearing in the above lemma are not necessarily all the same. On the other hand, they must certainly divide $|D|$.
\begin{proof}[Sketch of Proof]
The proof is very similar to the one for the analogous result for lattices, as given in Abschnitt~28 of \cite{Kne02}. Given two self-dual, isotropic subgroups $H$ and $H'$, denote by $m(H,H')\coloneqq[H:H\cap H']=[H':H\cap H']$ the index of their intersection in each of them. Then, for any prime $p$, the \emph{$p$-distance} $d_p(H,H')$ is defined as the $p$-adic valuation $d_p(H,H')\coloneqq\nu_p(m(H,H'))\in\N$. Clearly, $d_p(H,H')>0$ is only possible if $p$ divides $|D|$, which holds for finitely many primes $p$.

The crucial step is to show that for any fixed prime~$p$, there is a $p$-neighbor $\tilde{H}$ of~$H$ such that $d_p(\tilde{H},H')<d_p(H,H')$, while not changing the $p'$-distances $d_{p'}(\tilde{H},H')$ for all primes $p'\neq p$.

After repeating this step $d_p(H,H')$ many times for each prime $p$, one arrives at a self-dual, isotropic subgroup $\tilde{H}$ of $D$ that is an iterated neighbor of $H$ and satisfies $d_p(\tilde{H},H')=0$ for all primes $p$. But then $\tilde{H}=H'$, and we are done.
\end{proof}

We come to the main result:

\begin{prop}\label{prop:hypgenpneighbor2}
Let $V$ and $V'$ be \strat{}, holomorphic \voa{}s. If $V$ and $V'$ are in the same hyperbolic genus, then they are iterated (not necessarily inner) $p$-neighbors for possibly different primes $p$.
\end{prop}

Of course, we already hypothesized at the end of \autoref{subsec:orbneighborhood} that allowing iterated $p$-neighbors, possibly non-inner and for possibly different primes $p$, actually connects all \voa{}s in a \emph{bulk} genus, but this statement seems currently out of reach, and would imply \autoref{conj:anisobulkorb} and \autoref{conj:wittorb}.

\begin{proof}
Suppose that $V$ and $V'$ are hyperbolically equivalent. Then, in particular, they have the same Heisenberg commutant $C$ and their associated lattices $L$ and $M$, respectively, are in the same lattice genus by \cite[\autoref*{MR1thm:hypcomm}]{MR24a}.

We proceed in two steps. First, recall from \autoref{sec:lat} that $L$ and $M$ are iterated $p$-neighbors for any fixed prime $p$ not dividing $d(L)=d(M)$. We suppose without loss of generality that $L$ and $M$ are actual and not iterated $p$-neighbors, and that $L\neq M$. This means that they have a common index-$p$ sublattice $K=L\cap M$. Any metric group can be split into an orthogonal direct sum of its $p$-adic components. Applying this to $K'/K$ with $d(K)=p^2d(L)=p^2d(M)$ and $p\nmid d(L)=d(M)$, it follows that $K'/K\cong L'/L\times D_p\cong M'/M\times D_p$ with the metric group $D_p=\Z_p\times\Z_p$ generated by the isotropic elements $\gamma_j\in K'/K$ for $j=1,2$ of order $p$ with bilinear form $B(\gamma_1,\gamma_2)=1/p+\Z$. We fix representatives $\gamma_j\in K'$. They satisfy $\gamma_j\in(1/p)K$ and $\langle\gamma_j,\gamma_j\rangle/2\in\Z$. As we explained in equation~\eqref{eqn:lattice_ext}, $L=\cup_{i\in\Z_p}(i\gamma_1+K)$ and $M=\cup_{i\in\Z_p}(i\gamma_2+K)$.

Now, consider the element $h\coloneqq\gamma_1$. Then, because $\langle\gamma_1,\gamma_2\rangle\in1/p+\Z$, the even lattice $K$ can be written as $K=M^h=\{\beta\in M\,|\,\langle\beta,h\rangle\in\Z\}$. We then define the inner automorphism $g\coloneqq e^{2\pi i h_0}\in\Aut(V')$. As $h\in(1/p)K\subset(1/p)M$ (and not in $M$ itself) and $\langle h,h\rangle/2\in\Z\subset(1/p)\Z$, the automorphism $g$ has order $p$ and type~$0$. Hence, $g$ defines an inner $p$-neighbor $\smash{\tilde{V}\coloneqq{V'}^{\orb(g)}}$ of the \voa{}~$V'$. Moreover, as $h\notin M'$, \autoref{prop:neightwocases} implies that this neighborhood is of type~I. By construction, $\smash{\tilde{V}={V'}^{\orb(g)}}$ has $L$ as associated lattice, which is also the associated lattice of $V$.

\smallskip

So far, we proved that $V'$ is connected by iterated inner $p$-neigh\-bor\-hood (for a fixed prime $p$, of type~I) to the \strat{}, holomorphic \voa{}~$\tilde{V}$. Comparing $V$ and $\tilde{V}$, we note that they are both simple-current extensions of the dual pair $C\otimes V_L$, i.e.
\begin{equation*}
V=\!\!\bigoplus_{\alpha+L\in L'/L}\!\! C^{\tau(\alpha+L)}\otimes V_{\alpha+L},\quad\tilde{V}=\!\!\bigoplus_{\alpha+L\in L'/L}\!\! C^{\sigma(\alpha+L)}\otimes V_{\alpha+L}
\end{equation*}
for some anti-isometries $\tau,\sigma\colon L'/L\to R_C$. It remains to prove that $V$ and $\tilde{V}$ are iterated $p$-neighbors, now possibly for several different primes $p$.

We denote by $H=\{(\tau(x),x)\,|\,x\in L'/L\}$ and $\tilde{H}=\{(\sigma(x),x)\,|\,x\in L'/L\}$ the self-dual, isotropic subgroups of $R_C\times L'/L$ corresponding to $V$ and $\smash{\tilde{V}}$, respectively. By the theory of simple-current extensions \cite{Yam04}, the automorphisms of $V$ that fix $C\otimes V_L$ point-wise are naturally isomorphic to $\smash{\Aut(V)_{C\otimes V_L}}\cong\smash{\hat{H}}\eqqcolon G$, the dual group of $H$, and analogously for $\tilde{V}$. Moreover, $V^G=C\otimes V_L=\smash{{\tilde{V}}^{\tilde{G}}}$. This shows that $\tilde{V}$ is an orbifold construction of $V$ for the abelian group $G<\Aut(V)$.

In fact, as both $V$ and $\tilde{V}$ are simple-current extensions of the same \voa{} $C\otimes V_L$, by interpolating the isotropic subgroups $H$ and $\tilde{H}$ as described in \autoref{lem:isotropicneighbor} we deduce that the orbifold construction of $\tilde{V}$ from $V$ by the abelian group $G$ can be split up into a sequence of $\Z_p$-orbifold constructions (or $p$-neigh\-bor\-hood) for some primes $p$ dividing $|G|$. This completes the proof.
\end{proof}

We note that the $\Z_p$-orbifold constructions in the very last step are not necessarily inner. Or in other words, the intermediate isotropic subgroups do not have to correspond to holomorphic extensions of $V$ with Heisenberg commutant $C$ (cf.\ \cite[\autoref*{MR1rem:wrongext}]{MR24a}), even though both $H$ and $\tilde{H}$ do. We shall actually observe this phenomenon in \autoref{ex:schellekens3b} for the hyperbolic genus J, where even though $|R_C\times L'/L|$ and hence any $|G|=|H|$ has only prime factors $2$ and $3$, the hyperbolic genus is not connected when combining inner orbifolds of order $2$ and $3$.

On the other hand, if the $\Z_p$-orbifold constructions in the sequence in the last step can be chosen to be all inner, meaning that $C$ is the Heisenberg commutant in each intermediate step, then they are already of type~IIa. Indeed, by \autoref{prop:neightwocases}, each step describes an inner $p$-neigh\-bor\-hood of type I or IIa, but as $C\otimes V_L$ is fixed point-wise, and $C$ is the Heisenberg commutant, it cannot be of type I.


\subsection{Examples}\label{subsec:neighborexamples}

We study some examples of neighborhood of holomorphic \voa{}s. As a good testing ground, we look at the bulk genus $(\Vect,24)$, i.e.\ the \strat{}, holomorphic \voa{}s of central charge $24$, which we already discussed in \cite[\autoref*{MR1ex:schellekens}]{MR24a} and \cite[\autoref*{MR1ex:schellekens2}]{MR24a} in the context of bulk and hyperbolic equivalence, respectively. (We also recall that in \autoref{ex:c=16orbifolds}, we already discussed $2$-neighborhood for the case $c=16$, which is much simpler.) We begin by considering not necessarily inner neighborhood.

\begin{ex}[Schellekens \VOA{}s, Neighborhood]\label{ex:schellekens3}~
\begin{enumerate}[wide]
\item We recalled in \autoref{ex:deepholes} that all the \voa{}s in the bulk genus $(\Vect,24)$ are orbifold equivalent if we ignore potential ``fake'' moonshine modules, i.e.\ assuming \autoref{conj:moonshineuniqueness}. In fact, each such \voa{} is a direct (often non-inner) $n$-neighbor of the Leech lattice \voa{} $V_\Lambda$ for some $n\in\Ns$.

\smallskip

\item On the other hand, we argued in \autoref{rem:fakemoonshine} that if there were any ``fake'' moonshine modules, they could not be connected to the Leech lattice \voa{} or any other of the $71$ Schellekens \voa{}s by (repeated) cyclic orbifold constructions.

\smallskip

\item Another related statement is the fact that for any fixed prime $p$, the $p$-neigh\-bor\-hood graph of $(\Vect,24)$ is disconnected, independently of whether there are ``fake'' moonshine modules or not. Indeed, this was already rigorously observed in \cite{Mon98} for $p=2,3$. The result for $p=2$ was confirmed in \cite{HM23}, where, in fact, the complete $2$-neigh\-bor\-hood graph was computed.

If one only wants to prove the disconnectedness, it suffices to look at the weight-$1$ Lie algebras of the Schellekens \voa{}s and consider the classification of possible fixed-point Lie subalgebras from \cite{Kac90}. It is then not difficult to show that the $p$-neigh\-bor\-hood graph is disconnected for all primes $p$.

Indeed, as the outer automorphism groups of all Schellekens \voa{}s are finite \cite{BLS23}, for large enough primes $p$ there are only inner $\Z_p$-orbifold constructions, and they preserve the Heisenberg commutant and hence the hyperbolic genus (see \autoref{prop:hypgenpneighbor1}), of which there is more than one in $(\Vect,24)$. Hence, we only need to explicitly check the assertion for finitely many primes $p$, which is readily performed.
\end{enumerate}
\end{ex}

In view of \autoref{conj:hypgenpneighbor}, we turn our attention to inner neighborhood. Looking at \autoref{rem:typeI} and \autoref{rem:typeIIa}, and at the proof of \autoref{prop:hypgenpneighbor2}, there are two effects that can cause a hyperbolic genus to split up into more than one connected component under inner $p$-neigh\-bor\-hood for a fixed prime $p$:
\begin{enumerate}[wide]
\item The genus of the associated lattices may not be connected under lattice $p$-neigh\-bor\-hood for certain primes $p$ (see \autoref{sec:lat}), which typically can only happen if $p$ divides the determinant (or when the spinor genus and genus do not agree). On the other hand, every lattice in the genus must appear as associated lattice of some \voa{} in the hyperbolic genus. Hence, recalling that type-I inner $p$-neigh\-bor\-hood simply corresponds to lattice $p$-neigh\-bor\-hood and that type-IIa inner neighborhood does not change the associated lattice at all, we conclude that the hyperbolic genus will also be disconnected for these (finitely many) primes~$p$.

\item For a fixed associated lattice $L$, different anti-isometries $\tau\colon L'/L\to R_C$ may result in non-isomorphic holomorphic \voa{}s (cf.\ the discussion at the end of \cite[\autoref*{MR1ex:schellekens2}]{MR24a}). A more precise account of the exact number of non-isomorphic extensions of $C\otimes V_L$ is given in \cite{Hoe17,HM23}. We saw that these holomorphic extensions are connected by a chain of $p$-neighbors for certain (possibly different) primes $p$ that divide $d(L)$. These steps are not necessarily inner, but if they are all inner, they are of type~IIa. In any case, this may cause the hyperbolic genus to become disconnected under inner $p$-neigh\-bor\-hood.

On the other hand, these holomorphic \voa{}s may also be connected by type-I inner neighborhood, as we shall see in the following, and it seems that this causes the hyperbolic genus to become connected under inner $p$-neigh\-bor\-hood if the prime $p$ is large enough (cf.\ \autoref{conj:hypgenpneighbor}).
\end{enumerate}

\begin{ex}[Schellekens \VOA{}s, Inner Neighborhood]\label{ex:schellekens3b}
We investigate the Schellekens \voa{}s from the point of view of inner $p$-neigh\-bor\-hood for $p$ prime, especially in view of \autoref{conj:hypgenpneighbor}.

We find that each of the hyperbolic genera A, B, C, E, F, G, H, I, K and L is connected under inner $p$-neigh\-bor\-hood for all primes $p$. For all these hyperbolic genera, there is only one \voa{} for each associated lattice so that the connectedness of the lattice genus under $p$-neigh\-bor\-hood implies the connectedness of the hyperbolic genus. The former is the case for all primes $p$ that are large enough (see \autoref{sec:lat}) so that it suffices to explicitly check the connectedness for finitely many primes.

\smallskip

In the following, we study the more interesting genera D and J in detail. The hyperbolic genus D consists of nine \voa{}s, of which six have one associated lattice and three another one in the genus $\gDD$, denoted by D1 and D2, respectively, in \cite{Hoe17}.

\setlist[description]{font=\normalfont\itshape}
\begin{description}[wide]
\item[Order 2] The lattice genus $\gDD$ is not connected under $2$-neigh\-bor\-hood. In fact, there are no $2$-neighbors whatsoever, not even self-neighbors. Consequently, by \autoref{rem:typeI}, there are no inner $2$-neighbors of type~I for the nine holomorphic \voa{}s in this hyperbolic genus.

On the other hand, because $2$ divides the order of the metric group $2_{\II}^{-10}4_{\II}^{-2}$, it is plausible (cf.\ \autoref{rem:typeIIa}) and does in fact happen that the \voa{}s in each subgenus D1 or D2 are connected under inner $2$-neigh\-bor\-hood of type~IIa. But D1 and D2 remain disconnected from one another.
\item[Order $\geq$3] The lattice genus $\gDD$ is connected under $p$-neigh\-bor\-hood for primes $p\geq3$. Hence, there must be a type-I orbifold construction of order~$p$ that connects a \voa{} from D1 with one from D2. In fact, we observe that all the \voa{}s in the hyperbolic genus D become connected under inner $p$-neigh\-bor\-hood of type~I for $p\geq3$.

On the other hand, as primes $p\neq2$ do not appear in the metric group $2_{\II}^{-10}4_{\II}^{-2}$, there are no inner $p$-neighbors of type~IIa at all by \autoref{rem:typeIIa}.
\end{description}

\smallskip

We then come to the hyperbolic genus J. It consists of two \voa{}s that both have the same associated lattice, which is the unique lattice in the genus $\gJJ$. Even though lattice neighborhood does not matter here, we note that the unique lattice in the genus may or may not be a self-neighbor depending on the prime $p$.

\begin{description}[wide]
\item[Order 2, 3] There are no lattice $2$- or $3$-neighbors whatsoever (not even self-neighbors). Consequently, there are no inner $2$- or $3$-neighbors of type~I for the two \voa{}s in the genus.

On the other hand, as both primes $2$ and $3$ appear in the metric group $2_{\II}^{+4}4_{\II}^{-2}3^{+5}$, it is plausible that there are inner $2$- or $3$-neighbors of type~IIa. However, the hyperbolic genus J stays disconnected (i.e.\ these neighbors are all self-neighbors).
\item[Order $\geq$5] For primes $p\geq 5$ there are lattice $p$-neighbors, which must be self-neighbors as there is only one lattice in the genus. And in fact, we observe that there are now inner $p$-neighbors of type~I. Moreover, we observe that these $p$-neighbors connect the two \voa{}s in the hyperbolic genus.

On the other hand, as primes $p\geq5$ do not appear in the metric group $2_{\II}^{+4}4_{\II}^{-2}3^{+5}$, there are no inner $p$-neighbors of type~IIa at all.
\end{description}
\end{ex}


\section{Future Directions}\label{sec:future}

We conclude by listing a few directions for future research.
\begin{enumerate}[wide]
\item In \autoref{exph:SU(2)1symmetries}, we have argued that every finite symmetry of the $\mathrm{SU}(2)_1$ RCFT is invertible. The methods employed there can be used to systematically classify the unitary, finite symmetries of all $(c_L,c_R)=(1,1)$ unitary RCFTs, building on the work of \cite{FGRS07,TW24}. Results in this direction will be reported in \cite{GR26,CLR}.
\item We have developed symmetry-subalgebra duality (\autoref{subsec:symmetrysubalgebraduality}) and quantum Galois theory (\autoref{subsec:galois}) for holomorphic \voa{}s. It is desirable to generalize these ideas so that they apply to arbitrary \strat{} \voa{}s. See, e.g., \cite{Rie25,DNRX25,GR26} for results along these lines, and also \cite{Bis17,BDG23} for analogous results in the conformal net literature.
\item Certainly the strongest conjectures of this manuscript are \autoref{conj:wittorb} and \autoref{conj:wittinterface}. It would be fruitful to develop any techniques which may eventually lead to proofs. Even discovering a counterexample, e.g., to \autoref{conj:wittinterface} would be extremely interesting: by \autoref{thmph:recharacterizationmainconj}, it would lead to a CFT with vanishing gravitational anomaly that does not admit a boundary condition with finite $g$-function, suggesting that there may be further hitherto undiscovered anomalies (or obstructions) to the existence of a boundary condition.
\item It would be interesting to develop the mathematics of gauging algebra objects of finite quantum dimension inside categories of topological line operators which have infinite rank. As we have described in \autoref{rem:nottransRCFT}, doing so appears necessary to make orbifolding an equivalence relation.
\item In \autoref{sec:neighborhood}, we studied $p$-neighborhood of holomorphic \voa{}s. In \cite{CL19}, Chapter~5, the authors connect the adjacency matrices of the $p$-neighborhood graphs of the even, unimodular lattices in a genus to certain Hecke operators acting on their theta series. It would be insightful to see if one can analogously define Hecke operators for the \strat{}, holomorphic \voa{}s in a bulk or hyperbolic genus. This would hint at a local theory of \voa{}s. We thank Nils Scheithauer for suggesting this question. See \cite{FM23,FM26,FM24} for recent developments on $p$-adic vertex algebras. See also the discussion in \cite[\autoref*{MR1sec:future}]{MR24a} regarding a possible connection to a hypothetical third version of a \voa{} genus.
\end{enumerate}


\bibliographystyle{alpha_noseriescomma}
\bibliography{references}{}

\end{document}